%% file: twofluid.tex
\documentclass[final,3p,times]{elsarticle}
\pdfoutput=1

\usepackage{amsmath}  
\usepackage{amsfonts}  
\usepackage{amsthm}  
\usepackage{hyperref}

\makeatletter
\providecommand{\doi}[1]{%
  \begingroup
    \let\bibinfo\@secondoftwo
    \urlstyle{rm}%
    \href{http://dx.doi.org/#1}{%
      doi:\discretionary{}{}{}%
      \nolinkurl{#1}%
    }%
  \endgroup
}
\makeatother

\usepackage{glstyle}

\begin{document}

\title{A Hierarchy of Non-Equilibrium Two-Phase Flow Models}
\author[sintef,nbi]{Gaute Linga}
\address[sintef]{SINTEF Energy Research, Trondheim, Norway}
\address[nbi]{Niels Bohr Institute, University of Copenhagen, Copenhagen, Denmark.}
\ead{gaute.linga@gmail.com}


\date{\today}

\begin{abstract}
  We consider a hierarchy of relaxation models for two-phase flow.
  The models are derived from the non-equilibrium Baer--Nunziato model, which is endowed with relaxation source terms to drive it towards equilibrium.
  The source terms cause transfer of volume, heat, mass and momentum due to differences between the phases in pressure, temperature, chemical potential and velocity, respectively. 
  The subcharacteristic condition is closely linked to the stability of such relaxation systems, and in the context of two-phase flow models, it implies that the sound speed of an equilibrium system can never exceed that of the relaxation system.
  Here, previous work by Flåtten and Lund~[\emph{Math.\ Models Methods Appl.\ Sci.}, 21 (12), 2011, 2379--2407] and Lund~[\emph{SIAM J.~Appl.~Math.\ } { 72}, 2012, 1713--1741] is extended to encompass two-fluid models, i.e.~models with separately governed velocities for the two phases.
  Each remaining model in the hierarchy is derived, and analytical expressions for the sound speeds are presented.
  Given only physically fundamental assumptions, the subcharacteristic condition is shown to be satisfied in the entire hierarchy, either in a weak or in a strong sense.
\end{abstract}


\begin{keyword}
two-phase flow \sep relaxation systems \sep subcharacteristic condition
\PACS 76T10 \sep 35L60
\end{keyword}

\maketitle


\input{introduction.sec}
\input{basic_model.sec}
\input{v_model.sec}
\input{p_model.sec}
\input{T_model.sec}
\input{mu_model.sec}
\input{pmu_model.sec}
\input{Tmu_model.sec}
\input{comparison.sec}
\input{conclusion.sec}

\section*{Acknowledgements}
The author is greatly indebted to T.~Flåtten and H.~Lund for invaluable advice and useful comments on the manuscript.
Further, U.~J.~F.~Aarsnes, S.~Bore, A.~Bolet, A.~Aasen, A.~Morin, S.~T.~Munkejord, and M.~Hammer are thanked for helpful discussions.

This work has been carried out with support from the BIGCCS Centre, performed under the Norwegian research program Centres for Environment-friendly Energy Research (FME), from the 3D Multifluid Flow project at SINTEF Energy Research, and from the European Union's Horizon 2020 research and innovation program through Marie Curie initial training networks under grant agreement 642976 (NanoHeal).
The following partners are acknownledged for their contributions: ConocoPhillips, Gassco, Shell, Statoil, TOTAL, GDF SUEZ and the Research Council of Norway (193816/S60).

\appendix
\input{appendix.sec}

\bibliographystyle{elsarticle-num-names}
\bibliography{references}

\end{document}

%% file: introduction.sec.tex
\section{Introduction}
The concurrent flow of two fluid phases occurs in a wide range of industrially relevant settings, including in nuclear reactors \cite{bestion_physical_1990,aarsnes2016review}, petroleum production \cite{bendiksen_dynamic_1991}, heat exchangers \cite{pettersen98}, cavitating flows \cite{saurel_modelling_2008}, and within carbon capture, transport and storage (CCS) \cite{berstad_co2_2011,munkejord2016co2,linga2016two}.
However, for most simulation purposes, resolving the full three-dimensional flow field may be too cumbersome, due to the complex interaction between the phases.
In particular, this encompasses calculating the temporal evolution of the interface between the phases, and the transfer of mass, heat and momentum across it.
Averaging methods (see e.g.\ Ishii and Hibiki \cite{ishii_thermo-fluid_2011} or Drew and Passman \cite{drew99}) may therefore be applied to avoid direct computation of the interface.
The resulting coarse-grained models may often be expressed as hyperbolic relaxation systems with source terms accounting for the interactions between the phases, driving them asymptotically towards equilibrium at a finite rate.
In a quasi-linear form, one-dimensional versions of such systems may be written as
\begin{align}
  \pd {\v U} t + \v A( \v U ) \pd {\v U} x = \frac{1}{\epsilon} \v Q ( \v U ),
  \label{eq:relaxation_system}
\end{align}
wherein $\v U (x,t) \in G \subseteq \mathbb{R}^N$ is the (smooth) vector of unknowns and $\v A (\v U )$ is a matrix which we shall call the \emph{Jacobian} of the system, in analogy to conservative systems.\footnote{In systems which can be written on the conservative form $\pd {\v U} {t} + \pd {\v F (\v U)}{x} = 0$, we have that in the weak form \cref{eq:homogeneous_system}, $\v A = \pd {\v F}{\v U}$ is the actual Jacobian of a flux $\v F$.} Further, $\epsilon$ is a characteristic time associated with the relaxation process described by $\v Q(\v U)$.
For an extensive review of the existing literature on such systems, see e.g.~Natalini \cite{natalini_recent_1998}, or, for a more up-to-date summary, consider the first few sections of Solem et al.~\cite{solem14} and the references therein.

Two limits of the relaxation system \eqref{eq:relaxation_system} will be considered in this paper:
\begin{itemize}[]
\item
The \emph{non-stiff limit}, corresponding to the limit
  $\epsilon \to \infty$. In this case, we may write
  \cref{eq:relaxation_system} as
  \begin{align}
    \pd {\v U} t + \v A(\v U) \pd {\v U} x = 0.
    \label{eq:homogeneous_system}
  \end{align}
  We will refer to \cref{eq:homogeneous_system} as the \emph{homogeneous system}.
\item
  The formal \emph{equilibrium limit}, which is characterized by $\v Q (\v U ) \equiv 0$. This defines an equilibrium manifold \cite{chen_hyperbolic_1994} through $\mathcal M = \left\lbrace \v U \in G : \v Q ( \v U ) = 0 \right\rbrace$.
  We now assume that the reduced vector of variables $\v u (x, t) \in \mathbb R^n$, where $n \leq N$, uniquely defines an equilibrium value $\v U = \mathcal E (\v u ) \in \mathcal M$.
  We may then express \cref{eq:relaxation_system} as
  \begin{gather}
    \pd {\v u} t + \v B (\v u ) \pd {\v u} x = 0,\label{eq:equilibrium_system} \quad \v U = \mathcal E ( \v u ),
  \end{gather}
  where
  $\v B (\v u ) = \v P(\v u) \v A (\mathcal E (\v u)) \partial_{\v u} \mathcal E (\v u )$ is the Jacobian of the reduced system. Herein, we have defined the operator $\v P(\v u) : \mathbb R^N \to \mathbb R^n$ through $ \v P (\v u) \partial_{\v u} \mathcal E (\v u )= \v I_n$, i.e.~the identity matrix.
  We will refer to \cref{eq:equilibrium_system} as the \emph{equilibrium system}.
\end{itemize}
We expect solutions of \cref{eq:relaxation_system} to approach solutions of \cref{eq:equilibrium_system} as $\epsilon \to 0$, i.e.~in the \emph{stiff limit}, where the relaxation towards equilibrium is assumed to be instantaneous.

\subsection{The subcharacteristic condition}
An essential concept which arises in the study of relaxation systems and their stability, is the so-called \emph{subcharacteristic condition}.
It was first introduced by Leray \cite{leray53}, subsequently independently found by Whitham \cite{whitham74}, and later developed by Liu \cite{liu_hyperbolic_1987} for conservative $2 \times 2$ systems.
For more general systems, Chen et al.~\cite{chen_hyperbolic_1994} defined an entropy condition which they showed implies the subcharacteristic condition.
Yong \cite{yong_basic_2001} proved that for $n = N-1$, the subcaracteristic condition is necessary for the linear stability of the equilibrium system.
Solem et al.~\cite{solem14} proved that it is also sufficient.
Hence, for rank 1 relaxation processes, \emph{the subcharacteristic condition is equivalent to linear stability}.

For a general $N\times N$ relaxation system, such as \cref{eq:relaxation_system}, the condition may be formulated as follows.
\begin{dfn}[Subcharacteristic condition]\label{dfn:subchar}
  Let the $N$ eigenvalues of the matrix $\v A$ of the homogeneous
  system \eqref{eq:homogeneous_system} be given by $\Lambda_i$, sorted
  in ascending order as
  \begin{align}
    \Lambda_1 \leq \ldots \leq \Lambda_i \leq \Lambda_{i+1} \leq \ldots \leq \Lambda_N .
    \label{eq:sorted_eigenvalues_homogeneous}
  \end{align}
  Similarly, let $\lambda_j$ be the $n$ eigenvalues of the matrix
  $\v B$ of the equilibrium system
  \eqref{eq:equilibrium_system}. Herein, the homogeneous system
  \eqref{eq:homogeneous_system} is applied to a local equilibrium state
  $\v U = \mathcal E (\v u)$, such that $\Lambda_i = \Lambda_i ( \mathcal E (\v u) )$, and $\lambda_j = \lambda ( \v u )$.
  Now, the equilibrium system \eqref{eq:equilibrium_system} is said to
  satisfy the \emph{subcharacteristic condition} with respect to the
  homogeneous system \eqref{eq:homogeneous_system} when (i) all $\lambda_j$ are real, and (ii) if the $\lambda_j$ are sorted in ascending order as
  \begin{align}
    \lambda_1 \leq \ldots \leq \lambda_j \leq \lambda_{j+1} \leq \ldots \leq \lambda_n ,
    \label{eq:sorted_eigenvalues_equilibrium}
  \end{align}
  then the eigenvalues of the equilibrium system are
  \emph{interlaced} with the eigenvalues of the homogeneous system,
  in the sense that $\lambda_j \in [ \Lambda_j , \Lambda_{j+N-n} ]$.
\end{dfn}

The subcharacteristic condition has been shown to be an important trait of many physical models \cite{baudin05a,baudin05b,flatten_wave_2010}, since the eigenvalues then have a direct physical interpretation as the characteristic wave speeds of the system.
In the context of relaxation models for two-phase flow, the fastest wave speeds are the speeds of pressure waves, which are identified as the \emph{fluid-mechanical speeds of sound}.
The subcharacteristic condition then implies in particular that \emph{the sound speeds of an {equilibrium model} can never exceed that of the relaxation model it is derived from.}
This is consistent with the folklore knowledge in the fluid dynamics community that the ``frozen'' speed of sound is higher than the equilibrium speed of sound \cite{gvozdeva1969triple,picard1987calculation,flatten_relaxation_2011}. 

\subsection{The model hierarchy}
In a general averaged two-phase flow model, the two-phase mixture will consist of two fluids which evolve independently.
We assume \emph{local thermodynamic equilibrium} in each phase, i.e.~that each of the phases may be described by an equation of state, such that specifying two thermodynamic quantities completely determines all thermodynamic properties of that phase.
Herein lies also the assumption that the thermodynamic quantities are unaffected by the local velocity field.
Each phase $k$ may then be thought of as having separate pressures $p_k$, temperatures $T_k$, chemical potentials $\mu_k$, and velocities $v_k$. Since the two-phase mixture will move towards phase equilibrium in each of the mentioned variables, we may model these interactions by employing \emph{relaxation source terms} corresponding to the following \emph{relaxation processes}:
\begin{enumerate}
  \item [$p$] volume transfer, i.e.~relaxation towards mechanical equilibrium due to pressure differences between the phases, i.e.~expansion or compression;
  \item [$T$] heat transfer, i.e.~relaxation towards thermal equilibrium, due to temperature differences between the phases; 
  \item[$\mu$] mass transfer: relaxation towards chemical equilibrium due to differences between the phases in chemical potential;\footnote{See also \cref{rmk:mu-equil}.} and
  \item[$v$] momentum transfer, i.e.~relaxation towards velocity equilibrium, due to velocity differences between the phases, i.e.~friction.
\end{enumerate}

The starting point of the forthcoming analysis will be the celebrated Baer--Nunziato (BN) model \cite{baer_two-phase_1986}, which is the most general available formulation of a \emph{two-fluid model}, i.e.~models where the phases have separately governed velocities.
The BN model is endowed with appropriate relaxation terms corresponding to each of these processes presented above.
By considering the \emph{homogeneous} and \emph{equilibrium} limits of each relaxation process, i.e.~assuming all combinations of zero or more of them to be instantaneous, we obtain a \emph{hierarchy of models}, each with partial equilibrium in one or more of the aforementioned variables.

This hierarchy can be represented as a four-dimensional hypercube, as illustrated in \cref{fig:model_hierarchy}.
\begin{figure}[htb]
\centering
\input{hierarchy_intro.tikz}
\caption{The 4-dimensional hypercube representing the model hierarchy.
  Parallel edges correspond to the same relaxation processes, and each vertex signifies a unique model in the hierarchy, assuming instantaneous relaxation in zero or more of the variables $p$ (pressure), $T$ (temperature), $\mu$ (chemical
  potential) and $v$ (velocity).
  The leftmost, red circle denoted by ``0'' represents the Baer--Nunziato model \cite{baer_two-phase_1986}.
  The colored edges represent relaxation processes where a subcharacteristic condition has been established.
  This was done by \citet{flatten_relaxation_2011} and \citet{lund_hierarchy_2012} between the models shown in yellow circles, by \citet{martinez_ferrer_effect_2012} for the models in green circles, and for the model in blue by \citet{morin_two-fluid_2013}.
}
\label{fig:model_hierarchy}
\end{figure}
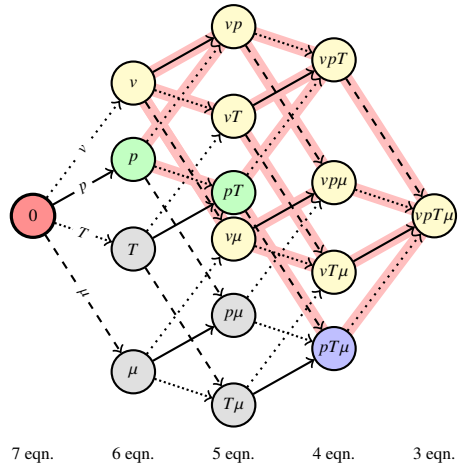
Herein, each model is symbolized by a circle, and corresponds to a ``corner'' of the hypercube.
Parallel edges, in turn, correspond to the same instantaneous relaxation assumption, in the direction of the arrow.
The basic model, denoted by ``0'' and shown in red as the leftmost circle of \cref{fig:model_hierarchy}, is thus reducible to all models in the hierarchy.
Many of the models in the hierarchy have already been derived, explicitly expressed and thoroughly analyzed, and in this respect, the current paper builds heavily on previous work \cite{gallouet2004numerical,allaire2002five,kapila2001two,bdzil1999two,saurel_modelling_2008,coquel2002closure,zein2010modeling,saurel2009simple}.

The models shown in yellow circles in \cref{fig:model_hierarchy} constitute the $v$-branch of the hierarchy, i.e.~the \emph{homogeneous flow} models, wherein the phase velocities are equal.
Such models are subclass of the so-called drift-flux models, where the phasic velocities are related by an algebraic expression.
Herein, the $v$-model was derived by \citet{saurel2009simple}, the $vp$-model is due to \citet{kapila2001two} (see also Refs.~\cite{murrone2005five,allaire2002five}), and the $vpT$-model was studied e.g.~by \cite{flatten_wave_2010}.
The $vpT\mu$-model is known as the \emph{homogeneous equilibrium model} and has been studied by several authors, see e.g.\ Refs.\ \cite{stewart1984two,menikoff1989riemann,clerc2000numerical,voss2005exact,helluy2006relaxation,helluy2011pressure,faccanoni2012modelling}.
Flåtten and Lund \cite{flatten_relaxation_2011} collected results on the $v$-, $vp$-, $vpT$-, and $vpT\mu$-models, derived the $vp\mu$-model, and showed that the subcharacteristic condition was satisfied for all relaxation processes within this branch of the hierarchy.
Lund \cite{lund_hierarchy_2012} completed the $v$-hierarchy by deriving the $vT$-, $v\mu$- and $vT\mu$-models, and established the subcharacteristic condition in the remainder of the $v$-branch, given only physically fundamental assumptions.

With regards to the \emph{two-fluid} models in the hierarchy, several of these models have been derived, employed in simulation \cite{bestion_physical_1990,bendiksen_dynamic_1991}, and analyzed.
Here, the $p$-model was analyzed e.g.\ in Refs.\ \cite{coquel_numerical_1997,stewart1984two}, and the $pT$-model was studied e.g.\ in Refs.~\cite{martinez_ferrer_effect_2012,hammer2014method}.

An important issue with $p$-relaxed (one-pressure) two-fluid models is that they develop complex eigenvalues when the velocity difference between the phases exceeds a critical value, i.e.~they become non-hyperbolic \cite{stewart1984two,cortes1998density,toumi1999approximate,gallouet2010hyperbolic,morin_two-fluid_2013}.  
This may lead to the lack of stable mathematical and numerical solutions.
Nevertheless, these models are extensively used for practical applications; and in numerical simulations they are often mitigated by specifying a regularizing interfacial pressure (see \cite{stuhmiller1977influence,bestion_physical_1990,ndjinga2007influence}).
Further, estimates of fluid-mechanical sound speeds is of practical importance for the construction of efficient numerical schemes \cite{roe1981approximate,hammer2014method}.
For relations between two-fluid models, we find, as in Ref.~\cite{martinez_ferrer_effect_2012}, the need to state a weaker formulation of the subcharacteristic condition.

\begin{dfn}[Weak subcharacteristic condition]\label{dfn:subchar_weak}
  When the subcharacteristic condition of \cref{dfn:subchar} holds with the additonal equilibrium condition of equal phasic velocities, the \emph{weak subcharacteristic condition} is said to be satisfied.
\end{dfn}

The $p$- and $pT$-models were analyzed by Martínez Ferrer et al.\ \cite{martinez_ferrer_effect_2012}, who showed that the subcharacteristic condition, in a weak or strong sense, is satisfied with respect to existing neighbouring models.
Similarly, Morin and Flåtten \cite{morin_two-fluid_2013} derived the $pT\mu$-model, and showed that subcharacteristic conditions were satisfied in relation to existing neighbouring models.
The highlighted edges in \cref{fig:model_hierarchy} summarize the relations between models where a subcharacteristic condition has already been shown to be satisfied.

\subsection{Contributions of this paper}
The objective of the current paper is complete the study of the subcharacteristic condition in the full hierarchy of two-phase flow models, proving the remaining subcharacteristic conditions.
In this respect, a generalization of the work by \citet{flatten_relaxation_2011} and \citet{lund_hierarchy_2012} is provided, extending the hierarchy to encompass also \emph{two-fluid models}, i.e.~models with separate momentum equations for the two phases.
To this end, the derivations of the two-fluid $T$-, $\mu$-, $p\mu$- and $T\mu$-models represent original contributions.
Expressions for the sound speeds in these models are provided.
Moreover, we show that the remaining 15 subcharactistic conditions are satisfied, i.e.~that the subcharacteristic condition is everywhere respected in the hierarchy, either in a strong or in a weak sense.
This is done by comparing the new expressions for the sound speeds to many known results from the literature, and by using techniques involving writing the difference of wave velocities as sums of squares (cf.\ \cite{lund_hierarchy_2012,flatten_relaxation_2011}).
We present each of the models for which we prove at least one subcharacteristic condition. 

\subsection{Outline}
The organization of the current paper is as follows.
In \cref{sec:basic_model} we present the basic model with all possible source terms, derive evolution equations for the primitive variables, and state a parameter set which suffices to satisfy the laws of thermodynamics.
In \cref{sec:v-model,sec:p-model,sec:T-model,sec:mu-model,sec:pmu-model,sec:Tmu-model}, we present in turn the $v$-, $p$-, $T$-, $\mu$-, $p\mu$- and $T\mu$-models.
For each model we give explicit analytic expressions for the sound speeds, and prove the remaining subcharacterisic conditions with respect to related models.
In \cref{sec:comparison} we show plots of the sound speeds in the different models, and briefly discuss physical and mathematical properties of models in the hierarchy.
Finally, in \cref{sec:conclusion}, we draw conclusions and suggest possible future work.


%% file: hierarchy_intro.tikz.tex
\begin{tikzpicture}[scale=1.1, swap, every node/.style={scale=0.65}]
    \foreach \pos/\name/\tag in {{(0,0)/0/$0$}, {(\hs,\vh)/v/$v$}, {(\hs,\ph)/p/$p$}, {(\hs,\Th)/T/$T$}, 
    						{(\hs,\mh)/m/$\mu$}, 
    						{(2*\hs,\vh+\ph)/vp/$vp$}, {(2*\hs,\vh+\Th)/vT/$vT$}, {(2*\hs,\ph+\Th)/pT/$pT$}, 
                            {(2*\hs,\vh+\mh)/vm/$v\mu$}, {(2*\hs,\ph+\mh)/pm/$p\mu$}, {(2*\hs,\Th+\mh)/Tm/$T\mu$},
                            {(3*\hs,\vh+\ph+\Th)/vpT/$vpT$}, {(3*\hs,\vh+\ph+\mh)/vpm/$vp\mu$}, 
                            {(3*\hs,\ph+\Th+\mh)/pTm/$pT\mu$},
                            {(3*\hs,\vh+\Th+\mh)/vTm/$vT\mu$}, {(4*\hs,\vh+\ph+\Th+\mh)/vpTm/$vpT\mu$}}
		\ifthenelse{\equal{\name}{0}}
			{\node[supermain vertex] (\name) at \pos {\tag}}
			{\ifthenelse{\equal{\name}{p} \OR \equal{\name}{pT}}
				{\node[main vertex] (\name) at \pos {\tag}}
				{\ifthenelse{ \equal{\name}{v} \OR \equal{\name}{vp} \OR \equal{\name}{vm} \OR \equal{\name}{vT} \OR \equal{\name}{vpm} \OR \equal{\name}{vpT} \OR \equal{\name}{vpTm} \OR \equal{\name}{vTm}}
					{\node[submain vertex] (\name) at \pos {\tag}}
					{\ifthenelse{ \equal{\name}{pTm}}
						{\node[vertex] (\name) at \pos {\tag}}
						{\node[gray vertex] (\name) at \pos {\tag}}
					}}};
	\foreach \pos/\tag in {{(0,\botl)/7 eqn.},{(\hs,\botl)/6 eqn.}, {(2*\hs,\botl)/5 eqn.}, {(3*\hs,\botl)/4 eqn.},
							{(4*\hs,\botl)/3 eqn.}}
		\node[textthing] (\tag) at \pos {\tag};
    \foreach \source/ \dest /\name /\weight in {0/v/v/v,0/p/p/p,0/T/T/T,0/m/m/\mu}
        \draw[edge \name] (\source) -- node[weight] {$\weight$} (\dest) ;
    \foreach \source/ \dest /\name /\weight in {v/vp/p/p,v/vT/T/T,
    						v/vm/m/\mu,p/vp/v/v,p/pT/T/T,p/pm/m/\mu,
    						T/pT/p/p,T/Tm/m/\mu,T/vT/v/v,m/vm/v/v,m/pm/p/p,m/Tm/T/T,
    						vp/vpT/T/T,vp/vpm/m/\mu,pT/vpT/v/v,pT/pTm/m/\mu,
    						vT/vpT/p/p,vT/vTm/m/\mu,vm/vpm/p/p,vm/vTm/T/T,pm/vpm/v/v,
    						pm/pTm/T/T,Tm/vTm/v/v,Tm/pTm/p/p,vpT/vpTm/m/\mu,
    						vpm/vpTm/T/T,vTm/vpTm/p/p,pTm/vpTm/v/v}
        \draw[edge \name] (\source) -- (\dest) ;
    \begin{pgfonlayer}{background}
	\foreach \source / \dest in {v/vp,vp/vpT,vpT/vpTm,vp/vpm,v/vm,v/vT,vT/vpT,
								v/vm,vT/vTm,vm/vTm,vm/vpm,vTm/vpTm,vpm/vpTm,
								p/pT,p/vp,pT/vpT,pT/pTm,pTm/vpTm}
		\path[selected edge] (\source.center) -- (\dest.center);
    \end{pgfonlayer}
\end{tikzpicture}

%% file: basic_model.sec.tex
\section{Basic model}\label{sec:basic_model}
In this section, we present the basic BN model \cite{baer_two-phase_1986}.
In this model, which is hyperbolic, the two phases have separate pressures, temperatures, chemical potentials and velocities.
We state the model in a form reminiscent of that proposed by Saurel and Abgrall \cite{saurel_multiphase_1999}, but with all four possible relaxation source terms accounting for the interaction between the phases.
From this, we determine the evolution equations of the primitive variables.
Based on the evolution equations, we derive a parameter set which suffices for the model to satisfy fundamental physical laws.

\subsection{Governing equations}
In the following, we present the governing equations in the basic model, supplemented with physically appropriate relaxation terms.
We let $\alpha_k$ denote volume fraction, $v_k$ velocity, $\rho_k$ density, $p_k$ pressure, $T_k$ temperature, $\mu_k$ chemical potential, $e_k$ internal energy per mass, for each phase $k \in \lbrace \Gas, \Liq \rbrace$, where $\Gas$ denotes gas and $\Liq$ denotes liquid.

\subsubsection{Volume advection}
We assume that apart from advection, the interface between the phases can only move due to pressure differences.
This is commonly formulated as
\begin{align}
  \pd {\alpha\gas} t + \vp \pd {\alpha\gas} x = \J (p\gas - p\liq), \label{eq:volumeg}
\end{align}
wherein $\vp$ is an interface velocity and $\J$ is the pressure relaxation parameter. Hence, the volume fraction is advected with the velocity $\vp$.
There are several discussions available on how to choose this interface velocity, see e.g.~\cite{coquel2002closure,saurel2018diffuse}.
In the following, we shall motivate it from a thermodynamic point of view, using the second law of thermodynamics.

The local volume transfer must occur so that the phase with the lowest pressure is compressed, and the phase with the highest pressure is expanded.
This is enforced through $\J \geq 0$.
Moreover, the volume fractions must satisfy $\alpha\gas + \alpha\liq = 1$, where $\alpha_k \in (0,1)$, and hence only one evolution equation for the volume fractions is needed.

\subsubsection{Mass balance}
The evolution of the mass of each phase is contained in the balance equations
\begin{gather}
  \pd {\alpha\gas\rho\gas} t + \pd {\alpha\gas\rho\gas v\gas } x = \K
  (\mu\liq - \mu\gas),\label{eq:massg7} \\
  \pd {\alpha\liq\rho\liq} t + \pd {\alpha\liq\rho\liq v\liq } x = \K (\mu\gas - \mu\liq), \label{eq:massl7}
\end{gather}
wherein $\K$ is the mass relaxation parameter, and the source terms on the right hand sides of \cref{eq:massg7,eq:massl7} account for mass transfer between the phases \cite{gavrilyuk1999new,gavrilyuk2002mathematical}.
The mass transfer occurs from the phase with the highest chemical potential towards the phase with the lowest, which is ensured through the assumption $\K \geq 0$.
We observe that \emph{conservation of total mass} is contained by summing \cref{eq:massg7,eq:massl7}:
\begin{align}
  \pd {\rho} t + \pd {\left(\alpha\gas\rho\gas v\gas + \alpha\liq\rho\liq v\liq\right)} x = 0,
  \label{eq:cons_mass}
\end{align}
wherein we have defined the mixture density $\rho = \alpha\gas \rho\gas + \alpha\liq \rho\liq.$

\subsubsection{Momentum balance}
Similar balance laws apply for the momentum of each phase:
\begin{align}
  \pd {\alpha\gas\rho\gas v\gas} t + \pd {(\alpha\gas\rho\gas v\gas^2 + \alpha\gas p\gas)} x - p\intph \pd {\alpha\gas} x = v\intph \K (\mu\liq - \mu\gas) + {\M (v\liq-v\gas)} , \label{eq:momg7} \\
  \pd {\alpha\liq\rho\liq v\liq} t + \pd {(\alpha\liq\rho\liq v\liq^2 + \alpha\liq p\liq)} x - p\intph \pd {\alpha\liq} x = v\intph \K (\mu\gas - \mu\liq) + {\M (v\gas-v\liq)}. \label{eq:moml7}
\end{align}
Herein, $p\intph$ is an interface pressure and $\M$ is the momentum relaxation parameter.
Note that from the averaging procedure resulting in these models, the interface velocity $v\intph$ in \cref{eq:momg7,eq:moml7} need not be the same as that in \cref{eq:volumeg} (see e.g.~Ref.~\cite{morin2012mathematical}).
However, we have chosen these to be equal to keep the notation to a minimum, as this will not influence the main conclusions of this paper.
The source terms associated with $v\intph$ on the right hand sides of \cref{eq:momg7,eq:moml7} represent the momentum of the condensating or vaporizing fluid, which is transferred to the other phase.
The source terms associated with $\M$ represent interfacial friction, and are assumed to cause momentum transfer from the phase with highest velocity towards the one with lowest velocity, which is ensured by requiring $\M \geq 0$.
We observe that \emph{conservation of total momentum} is ensured by summing \cref{eq:momg7,eq:moml7}:
\begin{align}
  \pd {\left( \alpha\gas\rho\gas v\gas + \alpha\liq\rho\liq v\liq \right)} t + \pd {\left(\alpha\gas\rho\gas v\gas^2 + \alpha\liq\rho\liq v\liq^2 + \alpha\gas p\gas + \alpha\liq p\liq \right)} x = 0.
  \label{eq:cons_mom}
\end{align}

\subsubsection{Energy balance}
The balance laws for the energy of each phase may be stated as
\begin{multline}
  \pd {E\gas} t + \pd {\left( E\gas v\gas + \alpha\gas v\gas p\gas
    \right)} x - p\intph \vp \pd {\alpha\gas} x \\ = -p\intph \J (p\gas-p\liq) + \left( \mu\intph + \tfrac{1}{2}v\intph^2 \right) \K (\mu\liq - \mu\gas) + \vm \M (v\liq-v\gas) + \H (T\liq-T\gas), \label{eq:eng7}
\end{multline}
\begin{multline}
  \pd {E\liq} t + \pd {\left( E\liq v\liq + \alpha\liq v\liq p\liq
    \right)} x - p\intph \vp \pd {\alpha\liq} x \\ = - p\intph \J (p\liq-p\gas) + \left( \mu\intph + \tfrac{1}{2}v\intph^2 \right) \K (\mu\gas - \mu\liq) + \vm \M (v\gas-v\liq) + \H (T\gas-T\liq). \label{eq:enl7}
\end{multline}
Herein, $\mu\intph$ is an interface chemical potential, $\H$ is the temperature relaxation parameter, and we have introduced the total phasic energy per volume $E_k = \Eint + \Ekin$, where the phasic internal and kinetic energies are given by, respectively,
\begin{align}
\Eint &= \alpha_k \rho_k e_k, \label{eq:def_Eint}\\
\Ekin &= \tfrac{1}{2} \alpha_k \rho_k v_k^2 .\label{eq:def_Ekin}
\end{align}
On the right hand side of \cref{eq:eng7,eq:enl7}, the terms associated with $\J$ represent energy transfer due to expansion--compression work, the terms associated with $\K$ represent the energy which the condensating or vaporizing fluid brings into the other phase, the terms associated with $\M$ represent energy transfer due to frictious momentum transfer, and the terms associated with $\H$ represent pure heat flow.
The latter should flow from the hotter to the colder phase, which is ensured through the assumption $\H \geq 0$.
Moreover, we see that \emph{total energy is conserved} by summing \cref{eq:eng7,eq:enl7},
\begin{align}
  \pd {E} t + \pd {\left( E\gas v\gas + E\liq v\liq + \alpha\gas v\gas p\gas + \alpha\liq v\liq p\liq \right)} x = 0,
  \label{eq:cons_en}
\end{align}
where we have introduced the mixed total energy per volume $E = E\gas + E\liq$.

\subsubsection{Phase independent form}
With all possible relaxation terms, the BN model \cite{baer_two-phase_1986}, as presented in \cref{eq:volumeg,eq:massg7,eq:massl7,eq:momg7,eq:moml7,eq:eng7,eq:enl7}, can be stated compactly as
\begin{align}
  \pd {\alpha_k} t + \vp \pd {\alpha_k} x &=I_k, \label{eq:volume7k}\\
  \pd {\alpha_k \rho_k} t + \pd {\alpha_k \rho_k v_k } x &= K_k, \label{eq:mass7k}\\
  \pd {\alpha_k \rho_k
    v_k} t + \pd {(\alpha_k \rho_k v_k^2 + \alpha_k p_k)} x - p\intph
  \pd {\alpha_k} x &= v\intph K_k + M_k, \label{eq:mom7k} \\
  \pd {E_k} t + \pd {\left( E_k v_k + \alpha_k v_k p_k \right)} x - p\intph \vp \pd {\alpha_k} x &= -p\intph I_k + \left( \mu\intph + \tfrac{1}{2}v\intph^2 \right) K_k + \vm M_k + H_k, \label{eq:energy7k}
\end{align}
for each phase $k \in \{ \Gas, \Liq \}$. Herein, the shorthand forms
of the relaxation source terms, $I_k$, $K_k$, $H_k$ and $M_k$, have
been defined such that $I\gas = -I\liq = \J (p\gas - p\liq)$, $K\gas = -K\liq = \K (\mu\liq - \mu\gas)$, $H\gas = -H\liq = \H (T\liq-T\gas)$, and $M\gas = -M\liq = \M (v\liq-v\gas)$.
\subsection{Evolution of primitive variables}
In order to systematically derive other models in the hierarchy, and
to derive a physically valid parameter set for the basic model, we now
seek the evolution equations for \emph{primitive variables}, such as
phasic velocity $v_k$, density $\rho_k$, pressure $p_k$, temperature
$T_k$, entropy $s_k$ and chemical potential $\mu_k$. To simplify the
notation in the forthcoming, we introduce the phasic \emph{material
  derivative}, defined by
\begin{align}
  \D_k \left( \cdot \right) \equiv \pd {\left( \cdot \right)} t + v_k \pd {\left( \cdot \right)} x,
  \label{eq:material_derivative}
\end{align}
for each phase $k \in \lbrace \Gas, \Liq \rbrace$. 

In the forthcoming calculations, the following relation will prove useful.
For an arbitrary quantity $f$, we have from \cref{eq:mass7k,eq:material_derivative} that
\begin{align}
  \alpha_k \rho_k \D_k f 
  &= \pd {\alpha_k \rho_k f } t + \pd {\alpha_k \rho_k v_k f } x - f K_k.
  \label{eq:partialDK}
\end{align}

Note that in order to keep the following analysis as short and concise as possible, most of the details in the derivations are omitted.
Further, in the analysis we assume each phase to be described by a completely general (smooth) equation of state.
The assumption that \emph{any primitive thermodynamic variable can be determined by knowledge of any two other} is important (see also the Supplementary Material for a list of relations between thermodynamic differentials).
Modelling equations of state and closure laws is beyond the scope of this paper.

We now proceed to find the evolution equations for the primitive variables.

\subsubsection{Volume fraction}
For clarity we state the evolution equation for the volume
fraction. Using \cref{eq:volume7k}, we have that
\begin{align}
  \D_k \alpha_k = I_k + (v_k - \vp) \pd{\alpha_k}{x}.
  \label{eq:Dalpha_k}
\end{align}

\subsubsection{Velocity}
We now seek the evolution equation for the phasic velocity. Using
$f = v_k$ in \cref{eq:partialDK}, 
and \cref{eq:mom7k}, we obtain 
\begin{align}
  \D_k v_k &= (\alpha_k \rho_k)^{-1} \left( (p\intph - p_k) \pd {\alpha_k} x - \alpha_k \pd {p_k} x + (v\intph - v_k) K_k + M_k \right).
  \label{eq:Dv_k}
\end{align}

\subsubsection{Density}
The density evolution equation is found by combining
\cref{eq:mass7k,eq:Dalpha_k},
\begin{align}
  \D_k \rho_k &= - \frac{\rho_k}{\alpha_k} (v_k - \vp) \pd{\alpha_k}{x} - \rho_k \pd {v_k} x  - \frac{\rho_k}{\alpha_k} I_k + \frac{1}{\alpha_k} K_k .
  \label{eq:Drho_k}
\end{align}

\subsubsection{Kinetic energy}
In order to obtain the evolution equation for the specific internal
energy, we start by finding the evoluton equations for the kinetic
energy. 
Using $f=v_k^2/2$ in \cref{eq:partialDK}, and
\cref{eq:Dv_k,eq:def_Ekin}, we obtain
\begin{align}
  \pd {\Ekin} t + \pd {\Ekin v_k} x + \alpha_k v_k \pd {p_k} x + v_k (p_k - p\intph) \pd {\alpha_k} x =  \left(v\intph v_k - \tfrac{1}{2} v_k^2 \right) K_k + v_k M_k.
  \label{eq:ekink}
\end{align}

\subsubsection{Internal energy}
We obtain the evolution equation for the
internal energy by subtracting \cref{eq:ekink} from
\cref{eq:energy7k}, expanding and collecting terms:
\begin{equation}
  \pd {\Eint} t + \pd {\Eint v_k} x + \alpha_k p_k \pd {v_k} x + p\intph (v_k - \vp) \pd {\alpha_k} x
  = -p\intph I_k + g_k K_k + (\vm -v_k) M_k + H_k, \label{eq:EintK}
\end{equation}
where we have introduced a shorthand interface energy $g_k = \mu\intph + \tfrac{1}{2} \left( v\intph - v_k \right)^2$.
Now, by using \cref{eq:EintK,eq:def_Eint} and $f=e_k$ in
\cref{eq:partialDK}, we obtain
\begin{equation}
  \D_k e_k = \tfrac{1}{\alpha_k \rho_k}\big(- p\intph ( I_k + (v_k - \vp) \pd {\alpha_k} x ) - \alpha_k p_k \pd {v_k} x + ( g_k - e_k ) K_k + (\vm -v_k ) M_k + H_k \big) .\label{eq:De_k}
\end{equation}

\subsubsection{Entropy}
The fundamental thermodynamic differential---alternatively referred to as the first law of thermodynamics---reads
\begin{align}
  \diff e_k = T_k \diff s_k + p_k \rho_k^{-2} \diff \rho_k,
  \label{eq:1stlaw}
\end{align}
where $s_k$ is the specific entropy of phase $k$.
By writing \cref{eq:1stlaw} in terms of material derivatives, and inserting \cref{eq:De_k,eq:Drho_k}, we obtain the evolution equation for the phasic entropy as
\begin{equation}
  \D_k s_k = (\alpha_k \rho_k T_k)^{-1} \big[ (p_k - p\intph) \left( I_k + (v_k - \vp) \pd {\alpha_k} x \right) + (g_k - h_k ) K_k + (\vm -v_k) M_k + H_k \big].
  \label{eq:Ds_k}
\end{equation}
Herein, the phasic specific enthalpy is defined as $h_k = e_k + p_k /\rho_k$.
By using $f=s_k$ in \cref{eq:partialDK} along with the identity
$\mu_k = h_k - T_k s_k$, \cref{eq:Ds_k} may be written in the balance
form
\begin{equation}
  \pd {S_k} t + \pd{S_k v_k} x = T_k^{-1} \Bigl[ (p_k - p\intph) \left( I_k + (v_k - \vp) \pd {\alpha_k} x \right) + (g_k - \mu_k ) K_k + (\vm -v_k) M_k + H_k \Bigr]
  \label{eq:DS_k}
\end{equation}
where we have defined the phasic entropy per volume $S_k = \alpha_k \rho_k s_k$.

\subsubsection{Pressure}
The pressure differential in terms of the density and entropy differentials may be written as
\begin{align}
  \diff p_k = c_k^2 \diff \rho_k + \Gamma_k \rho_k T_k \diff s_k,
  \label{eq:dp_drho_ds}
\end{align}
where we have introduced the phasic \emph{thermodynamic} speed of sound and the first Gr\"uneisen coefficient, respectively defined by $c_k^2 = \pdcinl{p_k}{\rho_k}{s_k}$ and $\Gamma_k = \rho_k^{-1} \pdcinl {p_k} {e_k}{\rho_k}$.
By writing \cref{eq:dp_drho_ds} in terms of the phasic material derivative, and inserting \cref{eq:Drho_k,eq:De_k}, we arrive at
\begin{equation}
  \D_k p_k 
  = \tfrac{\Gamma_k (p_k - p\intph) - \rho_k c_k^2 }{\alpha_k} \left( I_k + (v_k - \vp) \pd{\alpha_k}{x} \right) - \rho_k c_k^2 \pd {v_k} x + \tfrac{\Gamma_k (g_k - h_k ) + c_k^2}{\alpha_k} K_k + \tfrac{\Gamma_k}{\alpha_k} (\vm -v_k) M_k + \tfrac{\Gamma_k}{\alpha_k} H_k .
  \label{eq:Dp_k}
\end{equation}

\subsubsection{Temperature}
We now seek the equation governing the phasic temperature
evolution. The temperature differential may in terms of the pressure
and entropy differentials be written as
\begin{align}
  \diff T_k = \Gamma_k T_k \rho_k^{-1} c_k^{-2} \diff p_k + T_k \Cpk^{-1} \diff s_k,
  \label{eq:dT_dp_ds}
\end{align}
where the specific isobaric heat capacity is defined by $\Cpk = T_k \pdcinl {s_k}{T_k}{p_k}$.
Now, writing \cref{eq:dT_dp_ds} in terms of phasic material
derivatives, and inserting \cref{eq:Ds_k,eq:Dp_k}, we obtain
\begin{multline}
  \D_k T_k 
  = \left[ \frac{1 + \frac{\Gamma_k^2 \Cpk T_k}{c_k^2} }{\alpha_k \rho_k \Cpk} (p_k - p\intph) - \frac{\Gamma_k T_k}{ \alpha_k } \right] \left( I_k + (v_k - \vp) \pd{\alpha_k}{x} \right) - \Gamma_k T_k \pd {v_k} x \\ + \left[ \frac{\Gamma_k T_k}{\alpha_k \rho_k} +  \frac{1 + \frac{\Gamma_k^2 \Cpk T_k}{c_k^2} }{\alpha_k \rho_k \Cpk} (g_k - h_k ) \right] K_k 
  + \frac{1 + \frac{\Gamma_k^2 \Cpk T_k}{c_k^2} }{\alpha_k \rho_k \Cpk} \left[ (\vm -v_k) M_k + H_k  \right].
  \label{eq:DT_k}
\end{multline}

\subsubsection{Chemical potential}
The natural differential of the phasic chemical potential reads
\begin{align}
  \diff \mu_k = \rho_k^{-1} \diff p_k - s_k \diff T_k.
  \label{eq:dmu_dp_dT}
\end{align}
Therefore, writing \cref{eq:dmu_dp_dT} in terms of phasic material
derivatives, and inserting \cref{eq:Dp_k,eq:DT_k}, we obtain
\begin{multline}
  \D_k \mu_k 
  = \tfrac{1}{\alpha_k} \left[ \left(\Gamma_k - \tfrac{s_k}{\Cpk} - \tfrac{\Gamma_k^2 T_k s_k}{c_k^2} \right) \tfrac{(p_k - p\intph)}{\rho_k} -  c_k^2 + \Gamma_k T_k s_k \right] \left( I_k + (v_k - \vp) \pd{\alpha_k}{x} \right) \\
  - \left( c_k^2 - \Gamma_k T_k s_k \right) \pd {v_k} x 
  + \tfrac{1}{\alpha_k \rho_k }\left[ c_k^2 -\Gamma_k T_k s_k + \left( \Gamma_k - \tfrac{s_k}{\Cpk} - \tfrac{\Gamma_k^2 T_k s_k}{c_k^2} \right) (g_k - h_k ) \right] K_k \\
  + \tfrac{1}{\alpha_k \rho_k} \left(\Gamma_k - \tfrac{s_k}{\Cpk} - \tfrac{\Gamma_k^2 T_k s_k}{c_k^2} \right) \left[ (\vm -v_k) M_k + H_k \right].
  \label{eq:Dmu_k}
\end{multline}

\subsection{Laws of thermodynamics}
For the model to correctly represent physical phenomena, it should be verified that it satisfies fundamental physical principles \cite{flatten_relaxation_2011,flatten_wave_2010}.
We have already verified that it conserves mass, momentum and energy, respectively represented by \cref{eq:cons_mass,eq:cons_mom,eq:cons_en}, where the latter is known as the first law of thermodynamics.
We now consider the second law of thermodynamics, which states that the total entropy should be non-decreasing.
The analysis in the following is reminiscent of that of various previous works \cite{flatten_relaxation_2011,coquel2002closure}.

\subsubsection{Total entropy evolution}
The total entropy per volume is given by $S = S\gas + S\liq$.
The evolution equation for the total entropy is therefore found by summing \cref{eq:DS_k} over $k \in \{ \Gas, \Liq \}$:
\begin{align}
  \pd {S} t + \pd{(S\gas v\gas + S\liq v\liq)} x = \S_p + \S_{\mu} + \S_v + \S_T = \S,
  \label{eq:DS}
\end{align}
where we have defined the \emph{entropy source terms}
\begin{align}
  \S_p &= \left( \tfrac{p\gas - p\intph}{T\gas} - \tfrac{p\liq - p\intph}{T\liq} \right) I\gas + \left[ \tfrac{(p\gas - p\intph)(v\gas - \vp)}{T\gas} - \tfrac{(p\liq - p\intph)(v\liq - \vp)}{T\liq} \right] \pd {\alpha\gas} x ,\label{eq:entropy_source_p}\\
  \S_{\mu} &= \left( \left( \mu\intph - \mu\gas + \tfrac{1}{2} (v\intph-v\gas)^2 \right) T\gas^{-1} - \left( \mu\intph - \mu\liq + \tfrac{1}{2} (v\intph-v\liq)^2 \right) T\liq^{-1} \right) K\gas ,\label{eq:entropy_source_mu}\\
  \S_v &= \left[ \left(\vm -v\gas\right) T\gas^{-1} - \left(\vm -v\liq\right) T\liq^{-1} \right] M\gas ,\label{eq:entropy_source_v}\\
  \S_T &= ( T\gas^{-1} - T\liq^{-1} ) H\gas. \label{eq:entropy_source_T}
\end{align}
\subsubsection{The second law of thermodynamics}
We define the \emph{global entropy} as
\begin{align}
  \Omega (t) = \int_{\mathscr C}  S(x, t) \, \diff x,
  \label{eq:def_global_entropy}
\end{align}
where $\mathscr C \subseteq \mathbb R $ is some closed region.
\begin{dfn}\label{dfn:second_law}
  The \emph{second law of thermodynamics} states that the global entropy is non-decreasing, i.e.,
  \begin{align}
    \d {\Omega} t
      \geq 0 \quad \forall t,
    \label{eq:def_second_law}
  \end{align}
  in our context.
\end{dfn}

\begin{prop}
Sufficient conditions for the relaxation model given by \cref{eq:volumeg,eq:massg7,eq:massl7,eq:momg7,eq:moml7,eq:eng7,eq:enl7} to satisfy the second law of thermodynamics (\cref{dfn:second_law}) are
\begin{gather}
  \J,\K,\M,\H \geq 0,\label{eq:JKMHgeq0}\\
  \min\{ \mu\gas, \mu\liq \} \leq \mu\intph \leq \max\{ \mu\gas, \mu\liq \}, \label{eq:condmui}\\
  p\intph = \tfrac{ \sqrt{T\liq} p\gas + \sqrt{T\gas} p\liq }{\sqrt{T\gas} + \sqrt{T\liq}} , \label{eq:condpi}\\
  \vp = \tfrac{ \sqrt{T\liq} v\gas + \sqrt{T\gas} v\liq }{\sqrt{T\gas} + \sqrt{T\liq}}, \label{eq:condvivp}
\end{gather}
given only the physically fundamental assumption $T_k \geq 0$ for $k \in \{ \Gas, \Liq \}$.
\end{prop}
\begin{proof}
  By temporal differentiation of \cref{eq:def_global_entropy}, in combination with \cref{eq:def_second_law,eq:DS}, we obtain
  \begin{align}
    \int_{\mathscr C}\diff x \, \S (x,t) \geq 0,
    \label{eq:second_law_source}
  \end{align}
  where we have assumed that the entropy flux of \cref{eq:DS},
  $S\gas v\gas + S\liq v\liq$, vanishes at the boundary of
  $\mathscr C$. For \cref{eq:second_law_source} to be
  satisfied, clearly $\S \geq 0$ is a sufficient criterion, for which
  statement to hold the non-negativity of \emph{all} the partial source terms
  $\S_p, \S_{\mu}, \S_T$ and $\S_v$ is in turn sufficient. We now show
  this for each of the terms under the conditions of \cref{eq:JKMHgeq0,eq:condmui,eq:condpi,eq:condvivp}.

  Firstly, the conditions \cref{eq:condpi,eq:condvivp} inserted
  into \cref{eq:entropy_source_p} yields
  \begin{align}
    \S_p = \J \left( p\gas - p\liq \right)^2 (T\gas T\liq)^{-1/2} \geq 0.
  \end{align}
  Now, \cref{eq:condmui} is equivalent to $\mu\intph = \beta_\mu \mu\gas + (1-\beta_\mu) \mu\liq$, with $\beta_\mu \in [0,1]$. Hence, combination of \cref{eq:condvivp,eq:entropy_source_mu} yields
  \begin{align}
    \S_{\mu} = \K (\mu\liq -\mu\gas)^2 \left[ (1-\beta_\mu) T\gas^{-1} + \beta_\mu T\liq^{-1} \right] \geq 0.
  \end{align}
  Next, \Cref{eq:condvivp} inserted into \cref{eq:entropy_source_v} yields
  \begin{align}
    \S_v = \M (v\liq-v\gas)^2 (T\gas T\liq)^{-1/2} \geq 0.
  \end{align}
  Finally, \cref{eq:entropy_source_T} becomes
  \begin{align}
    \S_T = \H (T\liq-T\gas)^2 (T\gas T\liq)^{-1} \geq 0,
  \end{align}
  and hence all the source terms are non-negative.
\end{proof}

\begin{rmk}
  The interface conditons \cref{eq:condpi,eq:condvivp} are sufficient, not necessary, and the square-root-of-temperature weighted average between the phasic values differs from choices in the literature, e.g.~the initial choices by \citet{baer_two-phase_1986}.
  The reason for this particular weighting is that we enforced the interface velicities in \cref{eq:momg7,eq:moml7,eq:volumeg} to be equal, as noted above.
  Allowing these to differ would enable other linear combinations of the phasic quantities, which could possibly be more suitable for numerical simulations \cite{saurel2018diffuse}.
  These differences, however, do not have implications for the main conclusions of this paper.
\end{rmk}

\subsection{Wave velocities}
We now consider the homogeneous limit of the BN model, where the source terms $\J, \K, \M, \H \to 0$.
The resulting model has previously been extensively studied by several authors, see e.g.~\cite{saurel_multiphase_1999,zein2010modeling}.
The model has two fluid-mechanical sound speeds; one for each of the phases.
The seven wave velocities are given by $\boldsymbol{\lambda}_0 = \{ \vp, v\gas, v\liq, v\gas - c\gas, v\gas + c\gas, v\liq - c\liq, v\liq + c\liq \} $ \cite{saurel_multiphase_1999}.

In typical applications, the flow is subsonic, i.e. $|v\gas-v\liq| \ll c\gas, c\liq$ may be a valid approximation.
Evaluated in the velocity equilibrium limit, taking $v \equiv v\gas = v\liq$, the eigenvalues are, sorted in ascending order,
\begin{align}
  \boldsymbol{\lambda}_0^{(0)} = \{ v - c_{0,+} , v - c_{0,-} , v, v, v, v + c_{0,-}, v + c_{0,+} \}
  \label{eq:eigenvalues_0_local}
\end{align}
where we have defined
$ c_{0,+} = \max \{ c\gas, c\liq \}$  
 and
$ c_{0,-} = \min \{ c\gas, c\liq \}$ 
as the higher and lower sound speeds, respectively.


%% file: v_model.sec.tex
\section{The $v$-model}\label{sec:v-model}
We now study the model that arises upon imposing instantaneous equilibrium in velocity, i.e.~letting the velocity relaxation parameter $\M \to \infty$, which we expect corresponds to
\begin{align}
  v\gas = v\liq \equiv v.
  \label{eq:v-equil}
\end{align}
Simultaneously, we require the term $M\gas = \M (v\liq-v\gas)$ to remain finite.
By noting that for a general function $f$, the phasic material derivatives are equal for the two phases, i.e.~
$ \D_k f = \pd f t + v \pd f x \equiv \D f $,
then the system that results from evaluating \cref{eq:Dv_k} for the two phases $k \in \{ \Gas, \Liq \}$ can be solved to yield
\begin{align}
  M\gas 
  = \left( Y\gas p\liq + Y\liq p\gas - p\intph \right) \pd {\alpha\gas} x  + \alpha\gas Y\liq \pd {p\gas} x - \alpha\liq Y\gas \pd {p\liq} x ,
  \label{eq:M_v-equil}
\end{align}
where we have introduced the phasic mass fractions $Y_k = \alpha_k \rho_k / \rho $.
The model that now results from inserting \cref{eq:M_v-equil,eq:v-equil} into the basic model of \cref{sec:basic_model}, was analyzed by Flåtten and Lund \cite{flatten_relaxation_2011,lund_hierarchy_2012}, as it constitutes the basic model of the $v$-branch of the hierarchy. The model is hyperbolic and has previously been studied by many authors \cite{pelanti2014mixture,kapila2001two,saurel2009simple}.

\subsection{Wave velocities}
The wave velocities of the velocity equilibrium model, in the homogeneous limit where $\J,\K,\H\to 0$, are given by \cite{flatten_relaxation_2011}
\begin{align}
  \boldsymbol{\lambda}_v = \left\lbrace v - c_v , v, v, v, v, v + c_v \right\rbrace.
\end{align}
Herein, the sound speed of this model is defined by
\begin{align}
  c_v^2 = Y\gas c\gas^2 + Y\liq c\liq^2.
  \label{eq:c_v}
\end{align}

\begin{prop}\label{prop:SC_v_0}
  The $v$-model satisfies the subcharacteristic condition with respect to the basic model, given only the physically fundamental conditions $\rho_k, c_k^2 > 0$, for $k \in \{ \Gas, \Liq \}$.
\end{prop}
\begin{proof}
  We observe that $Y\gas + Y\liq = 1$, and due to the given positivity conditions, we have that $Y_k \in (0,1)$.
  Therefore, \cref{eq:c_v} implies that $\min \lbrace c\gas, c\liq \rbrace \leq c_v \leq \max \lbrace c\gas, c\liq \rbrace$.
  It then follows trivially that the wave velocities of the $v$-model are \emph{interlaced} in the wave velocities \cref{eq:eigenvalues_0_local} of the basic model evaluated in the velocity equilibrium state \cref{eq:v-equil}.
  Hence, the associated subcharacteristic condition of \cref{dfn:subchar} is satisfied.
\end{proof}


%% file: p_model.sec.tex
\section{The $p$-model}\label{sec:p-model}
In this section, we consider the mechanical equilibrium model, which arises when we assume instantaneous mechanical equilibrium in the basic model of \cref{sec:basic_model}.
We let the pressure relaxation parameter $\J \to \infty$, which we expect to correspond to $p\gas = p\liq \equiv p$.
Simultaneously, the product $I\gas = \J (p\gas-p\liq)$ should remain finite.
The mechanical equilibrium model is found by using \cref{eq:Dp_k} evaluated for each of the two phases.
From this, we may find an expression for $I\gas$ without temporal derivatives, and insert it into the basic model of \cref{sec:basic_model}.
The resulting model has been extensively studied previously \cite{saurel_multiphase_1999,martinez_ferrer_effect_2012}.
As other one-pressure two-fluid models, the model is not hyperbolic.

\subsection{Wave velocities}
We consider now the homogeneous limit, where $\K,\M,\H\to 0$.
The eigenvalues to the lowest order in the small parameter $\varepsilon = v\gas-v\liq$, i.e. evaluated in the equilibrium state defined by \cref{eq:v-equil}, are given by \cite{martinez_ferrer_effect_2012}
\begin{align}
  \boldsymbol{\lambda}_p^{(0)} = \left\lbrace v-c_p , v, v, v, v, v+c_p \right\rbrace,
\end{align}
where the sound speed in the $p$-model is given by
\begin{align}
  c_p^2
  = \left(\tfrac{\alpha\gas}{\rho\gas} + \tfrac{\alpha\liq}{\rho\liq} \right) \left( \tfrac{\alpha\gas}{\rho\gas c\gas^2 } + \tfrac{\alpha\liq}{\rho\liq c\liq^2 } \right)^{-1} . \label{eq:c_p}
\end{align}

\begin{prop}\label{prop:SC_p_0}
  The $p$-model satisfies the weak subcharacteristic condition of \cref{dfn:subchar_weak} with respect
  to the basic model of \cref{sec:basic_model}, subject only to the
  physically fundamental conditions $\rho_k, c_k^2 > 0$, for $k \in \{ \Gas, \Liq \}$, in the equilibrium state defined by
  \cref{eq:v-equil}.
\end{prop}
\begin{proof}
  We see from \cref{eq:c_p} that $c_p^2$ is a convex combination
  \begin{align}
    c_p^2 = \varphi\gas c\gas^2 + \varphi\liq c\liq^2, \qquad \textrm{where} \qquad \varphi_k = \left( \tfrac{\alpha_k}{\rho_k c_k^2} \right) \left( \tfrac{\alpha\gas}{\rho\gas c\gas^2 } + \tfrac{\alpha\liq}{\rho\liq c\liq^2 } \right)^{-1},
  \end{align}
  since $\varphi\gas + \varphi\liq=1$, and $\varphi_k \in (0, 1)$, due
  to the given conditions. This implies that
  \begin{align}
    \min\{ c\gas, c\liq \} \leq c_p \leq \max\{ c\gas, c\liq \},
  \end{align}
  and hence the weak subcharacteristic condition is fullfilled with respect
  to the basic model, whose local eigenvalues evaluated in the same
  state are given by \cref{eq:eigenvalues_0_local}.
\end{proof}


%% file: T_model.sec.tex
\section{The $T$-model}\label{sec:T-model}
In this section, we investigate the thermal-equilibrium model ($T$-model), which emerges from assuming instantaneous thermal equilibrium in the basic model of \cref{sec:basic_model}.
To this end, we let $\H \to \infty$ herein, which we expect corresponds to
\begin{align}
  T\gas = T\liq \equiv T,
  \label{eq:T-equil}
\end{align}
in such a way that $H\gas = \H (T\liq-T\gas)$ remains finite.
In the following we present the governing equations.

\subsection{Governing equations}
The full $T$-model may be stated as the basic model of \cref{sec:basic_model}, in which \cref{eq:eng7,eq:enl7} are replaced by \cref{eq:cons_en} and the thermal equilibrium condition \cref{eq:T-equil}.

In order to establish the impact of instantaneous thermal relaxation on the wave velocities, we need to express the model on a quasi-linear form, and thus obtain the velocities as the eigenvalues of the associated Jacobian.
This is most easily done by exploiting the primitive variables, which is what we now turn to do.

Firstly, we have that the phasic pressure differential in terms of density and temperature may be written as
\begin{align}
  \diff p_k
  = c_k^2 \zeta_k^{-1} \diff \rho_k + \Gamma_k \rho_k \Cpk \zeta_k^{-1} \diff T .
  \label{eq:Dpk_T-equil}
\end{align}
where we have introduced the ratio of specific heats $\zeta_k = 1 + \Gamma_k^2 \Cpk T / c_k^2 $,
and used \cref{eq:T-equil}. With \cref{eq:Dpk_T-equil}, \cref{eq:Dv_k} becomes
\begin{align}
  \D_k v_k 
           &= \tfrac{\Dip_k}{m_k} \pd {\alpha_k} x - \tfrac{ c_k^2}{\rho_k \zeta_k} \pd {\rho_k} x - \tfrac{\Gamma_k \Cpk}{\zeta_k} \pd T x + \tfrac{\Div_k}{m_k} K_k + \tfrac{1}{m_k} M_k ,
             \label{eq:Dv_k_T-equil}
\end{align}
where we have defined the phasic mass per volume $m_k= \alpha_k \rho_k$, the phasic interface pressure jump $\Dip_k= p\intph - p_k$, and the phasic interface velocity difference $\Div_k= v\intph - v_k$. Furthermore, \cref{eq:DT_k} becomes
\begin{multline}
  \D_k T = - \left[ \tfrac{ \zeta_k \Dip_k}{\mCpk} + \tfrac{\Gamma_k T}{ \alpha_k } \right] \left(  I_k - \Dpv_k \pd{\alpha_k}{x} \right)
  - \Gamma_k T \pd {v_k} x + \left[ \tfrac{\Gamma_k T}{m_k} + \tfrac{ \zeta_k }{\mCpk} (g_k - h_k ) \right] K_k
  + \tfrac{ \zeta_k }{\mCpk} \Dpv_k M_k + \tfrac{ \zeta_k }{\mCpk} H_k ,
  \label{eq:DT_1}
\end{multline}
where we have introduced the extensive heat capacity at constant pressure
$\mCpk = m_k \Cpk$.
We now define the weighting factor $\theta_k = \mCpk \zeta_k^{-1} / (\mCpg \zeta\gas^{-1} + \mCpl \zeta\liq^{-1})$,
for which clearly $\theta\gas + \theta\liq = 1$ and
$\theta_k \in (0, 1)$. Multiplying \cref{eq:DT_1} by $\theta_k$, and
summing over the phases yields
\begin{multline}
  \pd T t + \left( \theta\gas v\gas + \theta\liq v\liq \right) \pd T x
 = - \left[  \tfrac{\theta\gas\Gamma\gas T}{\alpha\gas} + \tfrac{\theta\liq\Gamma\liq T}{\alpha\liq}\right] \tfrac{v\gas-v\liq}{2} \pd{\alpha\gas}{x}
  - \theta\gas \Gamma\gas T \pd {v\gas} x
  - \theta\liq \Gamma\liq T \pd {v\liq} x \\
  + \left[ \frac{ p\gas - p\liq }{\frac{\mCpg}{\zeta\gas} + \frac{\mCpl}{\zeta\liq}} - \frac{\theta\gas\Gamma\gas T}{\alpha\gas} + \frac{\theta\liq \Gamma\liq T}{\alpha\liq} \right] I\gas 
  + \left[ \frac{ h\liq - h\gas  }{\frac{\mCpg}{\zeta\gas} + \frac{\mCpl}{\zeta\liq}} +  \frac{\theta\gas \Gamma\gas T}{m\gas} - \frac{\theta\liq\Gamma\liq T}{m\liq} \right] K\gas 
  + \frac{v\liq -v\gas}{\frac{\mCpg}{\zeta\gas} + \frac{\mCpl}{\zeta\liq}} M\gas ,
  \label{eq:DT}
\end{multline}
where have used the interface parameter definitions of \cref{eq:condvivp,eq:condpi} evaluated in thermal equilibrium \cref{eq:T-equil} to simplify.


\subsection{Wave velocities}
We now seek the wave velocities, i.e.~eigenvalues, in the homogeneous limit, where the relaxation source terms $\J,\K,\M\to 0$.
From \cref{eq:Dalpha_k}, it is then clear that $\alpha\gas$ is a characteristic variable of the system, since the volume fraction is advected with the velocity $\vp$ in the absence of relaxation source terms.
By using \cref{eq:Drho_k,eq:Dv_k_T-equil,eq:DT}, the remaining, reduced system may now be expressed on the quasi-linear form $\pd { \tilde{\v u}_T } t + \tilde{\v A}_T \left(\tilde{\v u}_T\right) \pd { \tilde{\v u}_T } x = 0$, where $\tilde{\v u}_T = [ \rho\gas, \rho\liq, v\gas, v\liq, T ]$, and the associated Jacobian is given by
\begin{align}
  \tilde {\v A} _T = \begin{bmatrix}
    v\gas & 0 & \rho\gas & 0 & 0 \\
    0 & v\liq & 0 & \rho\liq & 0 \\
    \frac{c\gas^2}{\rho\gas \zeta\gas} & 0 & v\gas & 0 & \frac{\Gamma\gas \Cpg}{\zeta\gas} \\
    0 & \frac{c\liq^2}{\rho\liq \zeta\liq} & 0 & v\liq & \frac{\Gamma\liq \Cpl}{\zeta\liq} \\
    0 & 0 & \theta\gas \Gamma\gas T & \theta\liq \Gamma\liq T & \theta\gas v\gas + \theta\liq v\liq
  \end{bmatrix},
\end{align}
from which we can find the remaining five eigenvalues. The characteristic polynomial of the latter is a fifth-degree polynomial, for which in general no closed-form solution can be obtained.
We now note that we may write $\tilde {\v A}_T = \tilde {\v A}_T^{(0)} + \varepsilon \tilde {\v A}_T^{(1)} $,
where $\varepsilon=v\gas-v\liq$. The matrices are given by
\begin{align}
  \tilde {\v A} _T^{(0)} = \begin{bmatrix}
    \bar v & 0 & \rho\gas & 0 & 0 \\
    0 & \bar v & 0 & \rho\liq & 0 \\
    \frac{c\gas^2}{\rho\gas \zeta\gas} & 0 & \bar v & 0 & \frac{\Gamma\gas \Cpg}{\zeta\gas} \\
    0 & \frac{c\liq^2}{\rho\liq \zeta\liq} & 0 & \bar v & \frac{\Gamma\liq \Cpl}{\zeta\liq} \\
    0 & 0 & \theta\gas \Gamma\gas T & \theta\liq \Gamma\liq T & \bar v
  \end{bmatrix},
\end{align}
and
$
  \tilde {\v A} _T^{(1)} 
  = \diag\left( \theta\liq, -\theta\gas, \theta\liq, -\theta\gas, 0 \right),
$
where we have taken $\bar v = \theta\gas v\gas + \theta\liq v\liq$.  Hence,
we approximate the eigenvalues by means of a perturbation expansion in
the small parameter $\veps$. To the lowest order in $\veps$, $v\gas = v\liq = \bar v = v$, and the eigenvalues of the $T$-model are given by
\begin{align}
  \boldsymbol{\lambda}_T^{(0)} = \left\lbrace v - c_{T,+}, v - c_{T,-}, v, v, v + c_{T,-}, v + c_{T,+} \right\rbrace
\end{align}
where 
the two distinct sound speeds of the model are given by
\begin{equation}
  c_{T,\pm}^2 = \frac{ \tfrac{c\gas^2 + c\liq^2}{T} \left(
      \tfrac{1}{\mCpg} + \tfrac{1}{\mCpl} \right) +
    \tfrac{\Gamma\gas^2 c\liq^2}{m\gas c\gas^2} + \tfrac{\Gamma\liq^2
      c\gas^2}{m\liq c\liq^2}  \pm \sqrt{ \left[ \tfrac{c\gas^2 -
          c\liq^2}{T} \left( \tfrac{1}{\mCpg} + \tfrac{1}{\mCpl}
        \right) -\tfrac{\Gamma\gas^2 c\liq^2}{m\gas c\gas^2} +
        \tfrac{\Gamma\liq^2 c\gas^2}{m\liq c\liq^2} \right]^2 + 4
      \tfrac{\Gamma\gas^2 \Gamma\liq^2}{m\gas m\liq} } }{
      2 \left[ \tfrac{\Gamma\gas^2}{m\gas c\gas^2} +
        \tfrac{\Gamma\liq^2}{m\liq c\liq^2} + \tfrac{1}{T} \left(
          \tfrac{1}{\mCpg} + \tfrac{1}{\mCpl} \right) \right]  } .
  \label{eq:c_T}
\end{equation}

\begin{prop}\label{prop:SC_T_0}
  The $T$-model satisfies the weak subcharacteristic condition with respect to the basic model of \cref{sec:basic_model}, subject only to the physically fundamental conditions $\rho_k, \Cpk, T > 0$, for $k\in\{\Gas,\Liq\}$, in the equilibrium state defined by \cref{eq:v-equil}.
\end{prop}
\begin{proof}
  We first show that the sound speeds are real. We note that on the given conditions, clearly $c_{T,\pm}^2 \in \mathbb R$, and moreover, $c_{T,+}^2 \geq 0$. The
  product of the sound speeds may be written as
  \begin{align}
    c_{T,+}^{-2} c_{T,-}^{-2} = c_{0,+}^{-2} c_{0,-}^{-2} + \Z_T^0,
    \quad
    \textrm{where}
    \quad
    \Z_T^0 = \frac{T \left(
      \tfrac{\Gamma\gas^2}{m\gas c\gas^2} + \tfrac{\Gamma\liq^2}{m\liq
        c\liq^2} \right)}{c_{0,+}^2 c_{0,-}^2 \left( \tfrac{1}{\mCpg} + \tfrac{1}{\mCpl} \right)}.
  \end{align}
  Based on the given conditions, it is clear that $\Z_T^0 \geq 0$ and therefore
  \begin{align}
    0 \leq c_{T,+}^2c_{T,-}^2 \leq c_{0,+}^2 c_{0,-}^2,
    \label{eq:SC_T_0_ordering_1}
  \end{align}
  and hence also $c_{T,-}^2 \geq 0$, and thus $c_{T,\pm}$ are real,
  and by definition, positive. Now, using the definitions of $c_{0,\pm}$ and \cref{eq:c_T}, it follows that
  \begin{gather}
    (c_{0,+}^2 - c_{T,+}^2)(c_{0,+}^2 - c_{T,-}^2)(c_{0,-}^2 -
    c_{T,+}^2)(c_{0,-}^2 - c_{T,-}^2) = - \Q_T^0,
    \label{eq:SC_T_0_prod}
  \end{gather}
  where
  \begin{gather}
    \Q_T^0 = \left( c\gas^2 - c\liq^2 \right)^2 \tfrac{ \Gamma\gas^2
      \Gamma\liq^2 }{m\gas m\liq} \left[ \tfrac{\Gamma\gas^2}{m\gas
        c\gas^2} + \tfrac{\Gamma\liq^2}{m\liq c\liq^2} + \tfrac{1}{T}
      \left( \tfrac{1}{\mCpg} + \tfrac{1}{\mCpl} \right) \right]^{-2} .
  \end{gather}
  The given conditions ensure that $\Q_T^0 \geq 0$. The only ordering of sound speeds compatible with \cref{eq:SC_T_0_prod,eq:SC_T_0_ordering_1} is
   $ 0 \leq c_{T,-} \leq c_{0,-} \leq c_{T,+} \leq c_{0,+} $,
  and hence the subcharacteristic condition of \cref{dfn:subchar} is
  satisfied.
\end{proof}

\begin{prop}\label{prop:SC_vT_T}
  The $vT$-model of Lund \cite{lund_hierarchy_2012} satisfies the
  subcharacteristic condition with respect to the $T$-model, given the
  physically fundamental assumptions $\rho_k, \Cpk, T > 0, $ for $k\in\{\Gas,\Liq\}$.
\end{prop}
\begin{proof}
  The sound speed of the $vT$-model is given by \cite{lund_hierarchy_2012}
  \begin{align}
    c_{vT}^2 = \frac{1}{\rho}\frac{m\gas c\gas^2 m\liq c\liq^2 \left(
        \frac{\Gamma\gas}{m\gas c\gas^2} + \frac{\Gamma\liq}{m\liq
          c\liq^2} \right)^2 + \frac{1}{T} \left( \frac{1}{\mCpg} +
        \frac{1}{\mCpl} \right) \left( m\gas c\gas^2 + m\liq c\liq^2
      \right)}{\frac{\Gamma\gas^2}{m\gas c\gas^2} +
      \frac{\Gamma\liq^2}{m\liq c\liq^2} + \frac{1}{T} \left(
        \frac{1}{\mCpg} + \frac{1}{\mCpl} \right)} 
    .
  \end{align}
  Now, using \cref{eq:c_T}, we can write the product of the
  differences as
  \begin{align}
    \left( c_{T,+}^2 - c_{vT}^2 \right) \left( c_{T,-}^2 - c_{vT}^2
    \right) 
    = - \Q_{vT}^T,
    \label{eq:SC_vT_T_prod}
  \end{align}
  where
  \begin{align}
    \Q_{vT}^T = Y\gas Y\liq \left[ \frac{ \frac{1}{T} \left( \frac{1}{\mCpg}
        + \frac{1}{\mCpl} \right) \left( c\gas^2 - c\liq^2 \right) -
      \frac{\Gamma\gas^2 c\liq^2}{m\gas c\gas^2} + \frac{\Gamma\liq^2
        c\gas^2}{m\liq c\liq^2} + \left( \frac{1}{m\gas} -
        \frac{1}{m\liq} \right) \Gamma\gas \Gamma\liq }{
      \frac{\Gamma\gas^2}{m\gas c\gas^2} +
      \frac{\Gamma\liq^2}{m\liq c\liq^2} + \frac{1}{T} \left(
        \frac{1}{\mCpg} + \frac{1}{\mCpl} \right) }\right]^2 .
  \end{align}
  With the given conditions, clearly $\Q_{vT}^T \geq 0$.
  Hence exactly one of the factors on the left hand side of \cref{eq:SC_vT_T_prod} is negative, and combined with \cref{prop:SC_T_0} we realize that $ c_{T,-} \leq c_{vT} \leq c_{T,+}$, and hence the subcharacteristic condition is satisfied.
\end{proof}

\begin{prop}\label{prop:SC_pT_T}
  The $pT$-model satisfies the weak subcharacteristic condition with respect to the $T$-model, given the physically fundamental assumptions $\rho_k,\Cpk, T > 0 $ for $k\in\{\Gas,\Liq\}$ in the equilibrium state defined by \cref{eq:v-equil}.
\end{prop}
\begin{proof}
  The sound speed of the $pT$-model is given by \cite{martinez_ferrer_effect_2012}
  \begin{align}
    c_{pT}^2 
    = \frac{ \left( \frac{\alpha\gas}{\rho\gas} +
      \frac{\alpha\liq}{\rho\liq} \right) \left( \frac{1}{\mCpg}
        + \frac{1}{\mCpl} \right) }{ \left(\frac{\alpha\gas}{\rho\gas c\gas^2}
      + \frac{\alpha\liq}{\rho\liq c\liq^2}\right) \left( \frac{1}{\mCpg}
        + \frac{1}{\mCpl} \right) + T \left(
        \frac{\Gamma\gas}{\rho\gas c\gas^2} -
        \frac{\Gamma\liq}{\rho\liq c\liq^2}\right)^2}
  \end{align}
  We may now write
  \begin{gather}
    \left( c_{T,+}^2 - c_{pT}^2 \right) \left( c_{T,-}^2 - c_{pT}^2
    \right) = - \Q_{pT}^T,
    \label{eq:SC_pT_T_prod}
  \end{gather}
  where
  \begin{multline}
    \Q_{pT}^T 
    = \frac{\alpha\gas \alpha\liq}{\rho\gas c\gas^2 \rho\liq
        c\liq^2 T} \left( \frac{1}{\mCpg} + \frac{1}{\mCpl} \right) \left[ \frac{\Gamma\gas^2}{m\gas c\gas^2} +
        \frac{\Gamma\liq^2}{m\liq c\liq^2} + \frac{1}{T} \left(
          \frac{1}{\mCpg} + \frac{1}{\mCpl} \right) \right]^{-1} \\ \times
    \left[
    \frac{ \left(\frac{1}{\mCpg}+\frac{1}{\mCpl}\right) \left(
          c\gas^2-c\liq^2 \right)-T \left( \frac{\Gamma\gas}{\rho\gas
            c\gas^2}-\frac{\Gamma\liq}{\rho\liq
            c\liq^2}\right)\left(\frac{\Gamma\gas
            c\liq^2}{\alpha\gas}+\frac{\Gamma\liq c\gas^2}{\alpha\liq}
        \right) }{
        \left(\frac{\alpha\gas}{\rho\gas c\gas^2} +
          \frac{\alpha\liq}{\rho\liq c\liq^2} \right)
        \left(\frac{1}{\mCpg}+\frac{1}{\mCpl}\right) + T \left(
          \frac{\Gamma\gas}{\rho\gas c\gas^2} -
          \frac{\Gamma\liq}{\rho\liq c\liq^2}\right)^2 } \right]^{2} .
  \end{multline}
  Clearly $\Q_{pT}^T \geq 0$, on the given conditions.
  Hence exactly one factor on the left hand side of \cref{eq:SC_pT_T_prod} is negative, yielding $c_{T,-} \leq c_{pT} \leq c_{T,+}$, and the weak subcharacteristic condition is satisfied.
\end{proof}




%% file: mu_model.sec.tex
\section{The $\mu$-model}\label{sec:mu-model}
We now proceed to investigate the chemical-equilibrium model (the $\mu$-model), which arises when we assume instantaneous chemical equilibrium, i.e.\ let the chemical relaxation parameter $\K \to \infty$, which we expect corresponds to
\begin{align}
  \mu\gas = \mu\liq \equiv \mu.
  \label{eq:mu-equil}
\end{align}
Simultaneously, we require the product $K\gas = \K (\mu\liq-\mu\gas)$ to remain finite, and in the forthcoming we seek to express this without any temporal derivatives.
\begin{rmk}
\label{rmk:mu-equil}
It should be noted that there does not seem to be a general agreement in the literature on how to properly model mass transfer (see e.g.\ \cite[pp.\ 13]{lund_phd_thesis}).
Strictly enforcing \cref{eq:mu-equil} may sometimes lead to unphysical results \cite{barberon2005finite}.
The present choice \cref{eq:mu-equil} is primarily motivated by compliance with the subhierarchy compiled by \citet{lund_hierarchy_2012}, and evaluating the physical relevance of these models is out of the scope of the present work.
\end{rmk}

The chemical potential evolution equation \cref{eq:Dmu_k} may be written as
\begin{equation}
  \D_k \mu = - \left[ \psi_k \Dip_k + \xi_k \alpha_k^{-1} \right] \left(
    I_k - \Dpv_k \pd{\alpha_k}{x} \right) - \xi_k \pd {v_k} x +
  \chi_k K_k + \psi_k \Dpv_k M_k + \psi_k H_k .
  \label{eq:Dmu}
\end{equation}
where we have used \cref{eq:mu-equil}, and defined the shorthands
\begin{align}
  \xi_k = c_k^2 - \Gamma_k T_k s_k, \quad
  \psi_k = \tfrac{\Gamma_k \xi_k}{m_k c_k^2} -
      \tfrac{s_k}{\mCpk} ,\quad
  \chi_k =  \tfrac{\xi_k^2}{m_k c_k^2} + \tfrac{T_k s_k^2}{\mCpk} + \tfrac{1}{2} (\Div_k)^2 \psi_k .
\end{align}
By using \cref{eq:Dmu} evaluated for each of the phases, and subtracting these expressions from each other, we obtain
\begin{equation}
  K\gas 
  =
  \kappa_\mu^{-1} \left( \xi\gas \pd {v\gas} x
  - \xi\liq \pd {v\liq} x
  - (\psi\gas + \psi\liq) H\gas \right)
  - \mcl{A}^\mu \pd {\alpha\gas} x
  + (v\gas-v\liq)\kappa_\mu^{-1} \pd {\mu} x 
  + \mcl K^\mu_p I\gas - \mcl K^\mu_v M\gas
  \label{eq:Kg_mu}
\end{equation}
where we have defined the shorthands
\begin{align}
  \kappa_{\mu} &= \tfrac{T\gas s\gas^2}{\mCpg} + \tfrac{T\liq s\liq^2}{\mCpl} + \tfrac{\xi\gas^2}{m\gas c\gas^2} + \tfrac{\xi\liq^2}{m\liq c\liq^2} + \tfrac{1}{2} \left( \psi\gas (\Div\gas)^2 + \psi\liq (\Div\liq)^2 \right) ,
  \label{eq:def_kappa_mu}\\
  \mcl A^\mu &=  \kappa_{\mu}^{-1} \left[ \left( \psi\gas \Dip\gas + \tfrac{\xi\gas}{\alpha\gas} \right) \Dpv\gas + \left( \psi\liq \Dip\liq + \tfrac{\xi\liq}{\alpha\liq}\right) \Dpv\liq \right] ,\\
  \mcl K^\mu_p &= \kappa_{\mu}^{-1} \left[ \psi\gas \Dip\gas + \psi\liq \Dip\liq + \tfrac{\xi\gas}{\alpha\gas} + \tfrac{\xi\liq}{\alpha\liq} \right] ,\quad
  \mcl K^\mu_v = \kappa_{\mu}^{-1} \left( \psi\gas \Dpv\gas + \psi\liq \Dpv\liq \right) .
\end{align}

\subsection{Governing equations}
By using the expression \cref{eq:Kg_mu} to insert for $K\gas$ in the basic model of \cref{sec:basic_model}, the $\mu$-model can now be summarized with the following set of equations:
\begin{itemize}
\item Volume advection: $\pd {\alpha\gas} t + \vp \pd {\alpha\gas} x = I\gas$,
\item Conservation of mass: $\pd {\rho} t + \pd {\left(\alpha\gas \rho\gas v\gas + \alpha\liq \rho\liq v\liq \right)}{x} = 0$,
\item Momentum balance:
  \begin{multline}
    \pd {\alpha\gas\rho\gas v\gas} t + \pd {(\alpha\gas\rho\gas  v\gas^2 + \alpha\gas p\gas)} x - \left( p\intph - v\intph \mcl A^\mu \right) \pd {\alpha\gas} x - v\intph \xi\gas \kappa_{\mu}^{-1} \pd {v\gas} x 
    + v\intph \xi\liq \kappa_{\mu}^{-1} \pd {v\liq} x \\ - v\intph (v\gas-v\liq) \kappa_{\mu}^{-1} \pd {\mu} x
    = v\intph \mcl K^\mu_p I\gas + \left(1 - v\intph \mcl K^\mu_v  \right) M\gas - v\intph \left( \psi\gas + \psi\liq \right) \kappa_{\mu}^{-1} H\gas , \label{eq:momg_mu}
  \end{multline}
\item Energy balance:
  \begin{multline}
    \pd {E\gas} t + \pd {\left( E\gas v\gas + \alpha\gas v\gas p\gas \right)} x - \left[ p\intph \vp - \left( \mu + \tfrac{1}{2}v\intph^2 \right) \mcl{A}^\mu \right] \pd
    {\alpha\gas} x - \left( \mu + \tfrac{1}{2}v\intph^2 \right) \kappa_\mu^{-1} \left[ \xi\gas \pd {v\gas} x - \xi\liq \pd
      {v\liq} x + (v\gas-v\liq) \pd {\mu} x \right] \\
    = \left[\left( \mu + \tfrac{1}{2}v\intph^2 \right) \mcl K^\mu_p -
      p\intph \right] I\gas + \left[ \vm - \left( \mu +
        \tfrac{1}{2}v\intph^2 \right) \mcl K^\mu_v \right] M\gas +
    \left[ 1 - \left( \mu + \tfrac{1}{2}v\intph^2 \right) (\psi\gas + \psi\liq)\kappa_\mu^{-1} \right]  H\gas, \label{eq:eng_mu}
  \end{multline}
\end{itemize}
Momentum and energy equations for the liquid phase are found by phase symmetry; interchanging indices $\Gas$ and $\Liq$.

\subsection{Evolution of primitive variables}
In order to write the system in a quasi-linear form, and thereby find the wave speeds of the $\mu$-model, we use the evolution equations for the primitive variables. We therefore now seek the evolution of some of the primitive variables under the assumption of instantaneous chemical equilibrium.

We first define the weighting factor $\phi_k = \chi_k^{-1} / (\chi\gas^{-1}+\chi\liq^{-1})$.
Multiplying \cref{eq:Dmu} by $\phi_k$ and summing over the phases, we get for the chemical potential
\begin{multline}
  \pd {\mu} t + \left(\phi\gas v\gas + \phi\liq v\liq \right) \pd
  {\mu} x
  + G^\mu_{\alpha\gas} \pd{\alpha\gas}{x} + \phi\gas \xi\gas \pd {v\gas} x +
  \phi\liq \xi\liq \pd {v\liq} x \\= \left[ - \phi\gas ( \psi\gas
      \Dip\gas + \xi\gas \alpha\gas^{-1} ) + \phi\liq (
      \psi\liq \Dip\liq + \xi\liq \alpha\liq^{-1} ) \right]
  I\gas + \left( \phi\gas \psi\gas \Dpv\gas - \phi\liq \psi\liq
    \Dpv\liq \right) M\gas + \left(\phi\gas \psi\gas - \phi\liq
    \psi\liq \right) H\gas ,
  \label{eq:Dmu_mu}
\end{multline}
where we have defined the shorthand coefficient
\begin{align}
  G^\mu_{\alpha\gas} = - \phi\gas \left( \psi\gas \Dip\gas +
  \xi\gas \alpha\gas^{-1} \right) \Dpv\gas + \phi\liq \left(
  \psi\liq \Dip\liq + \xi\liq \alpha\liq^{-1} \right) \Dpv\liq .
\end{align}
For the phasic velocity $v\gas$, we find from \cref{eq:Dv_k} the evolution equation
\begin{multline}
  \pd {v\gas} t + \left[ v\gas - \tfrac{\xi\gas \Div\gas}{m\gas \kappa_\mu} \right] \pd {v\gas} x + \tfrac{\xi\liq \Div\gas}{m\gas \kappa_\mu} \pd {v\liq} x + \tfrac{\Div\gas \mcl{A}^\mu - \Dip\gas}{m\gas}  \pd {\alpha\gas} x + \tfrac{1}{\rho\gas} \pd {p\gas} x
     -
    \tfrac{\Div\gas (v\gas-v\liq)}{m\gas \kappa_\mu} \pd {\mu} x
\\=
    \tfrac{\Div\gas}{m\gas} \mcl K^\mu_p I\gas
     + \tfrac{1}{m\gas} \left( 1 - \Div\gas \mcl K^\mu_v \right) M\gas - 
\tfrac{\Div\gas}{m\gas} \tfrac{\psi\gas + \psi\liq}{\kappa_\mu}
    H\gas,
    \label{eq:Dv_g_mu}
\end{multline}
and $v\liq$ is found by phase symmetry. 

The phasic pressure evolution is found from \cref{eq:Dp_k}. For the gas phase, it reads
\begin{multline}
  \pd {p\gas} t + v\gas \pd {p\gas} x + P^\mu_{\Gas,\alpha\gas}
  \pd{\alpha\gas}{x} + P^\mu_{\Gas,v\gas} \pd {v\gas} x +
  P^\mu_{\Gas,v\liq} \pd {v\liq} x + P^\mu_{\Gas,\mu} \pd {\mu} x \\=
  \alpha\gas^{-1}\left[ - \left( \Gamma\gas \Dip\gas + \rho\gas
      c\gas^2 \right) + \left( \xi\gas + \tfrac{1}{2} \Gamma\gas
      (\Div\gas)^2 \right) \mcl K^\mu_p \right] I\gas +
  \alpha\gas^{-1} \left[\Gamma\gas \Dpv\gas -\left(\xi\gas +
      \tfrac{1}{2} \Gamma\gas (\Div\gas)^2 \right) \mcl K^\mu_v
  \right] M\gas \\ + \alpha\gas^{-1} \left[ \Gamma\gas - \left(
      \xi\gas + \tfrac{1}{2} \Gamma\gas (\Div\gas)^2 \right)
    (\psi\gas + \psi\liq) \kappa_\mu^{-1} \right] H\gas .
  \label{eq:Dp_g_mu}
\end{multline}
wherein we have defined the coefficients
\begin{align}
  P^\mu_{\Gas,\alpha\gas} &= \alpha\gas^{-1}\left[ \left(\xi\gas + \tfrac{1}{2} \Gamma\gas (\Div\gas)^2 \right) \mcl{A}^\mu - \left(\Gamma\gas \Dip\gas +
                            \rho\gas c\gas^2 \right) \Dpv\gas \right],\\
  P^\mu_{\Gas,v\gas} &= \rho\gas c\gas^2 - \left(\xi\gas + \tfrac{1}{2}
    \Gamma\gas (\Div\gas)^2 \right) \xi\gas \alpha\gas^{-1} \kappa_\mu^{-1} ,\\
  P^\mu_{\Gas,v\liq} &= \left( \xi\gas +
                       \tfrac{1}{2} \Gamma\gas (\Div\gas)^2 \right)
                       \xi\liq \alpha\gas^{-1} \kappa_\mu^{-1},\\
  P^\mu_{\Gas,\mu} &= - \left(\xi\gas +
                     \tfrac{1}{2} \Gamma\gas (\Div\gas)^2 \right)
                     (v\gas-v\liq) \alpha\gas^{-1}\kappa_\mu^{-1} .
\end{align}
The corresponding expressions related to the evolution of $p\liq$ are found by phase symmetry.

\subsection{Wave velocities}
We now wish to derive the wave velocities of the $\mu$-model in the homogeneous limit, where $\J,\H,\M \to 0$.
In this limit, the volume fraction $\alpha\gas$ is a characteristic variable with the associated eigenvalue $\vp$.
The remaining, reduced model, i.e.~\cref{eq:Dp_g_mu,eq:Dmu_mu,eq:Dv_g_mu} for both phases, may then be expressed in the quasi-linear form
$\pd {\tilde{\v u}_\mu} t + \tilde {\v A}_\mu (\tilde{\v u}_\mu) \pd {\tilde{\v u}_\mu} x = 0$, where the reduced vector of unknowns is $\tilde{\v u}_\mu = [ \mu, v\gas, v\liq, p\gas, p\liq ]$, and the reduced Jacobian reads
\begin{align}
  \tilde{\v A}_{\mu} = \begin{bmatrix}
    \phi\gas v\gas + \phi\liq v\liq & \phi\gas \xi\gas & \phi\liq \xi\liq & 0 & 0 \\
    -\frac{\Div\gas (v\gas-v\liq)}{m\gas \kappa_\mu} & v\gas - \frac{\xi\gas \Div\gas}{m\gas \kappa_\mu} & \frac{\xi\liq \Div\gas}{m\gas \kappa_\mu} & \rho\gas^{-1} & 0 \\
    \frac{\Div\liq (v\gas-v\liq)}{m\liq \kappa_\mu} & \frac{\xi\gas \Div\liq}{m\liq \kappa_\mu} & v\liq - \frac{\xi\liq \Div\liq}{m\liq \kappa_\mu} & 0 & \rho\liq^{-1} \\
    P^\mu_{\Gas,\mu} & P^\mu_{\Gas,v\gas} & P^\mu_{\Gas,v\liq} & v\gas & 0 \\
    P^\mu_{\Liq,\mu} & P^\mu_{\Liq,v\gas} & P^\mu_{\Liq,v\liq} & 0 & v\liq
  \end{bmatrix} .
\end{align}
Again the eigenvalues $\lambda$ are given the roots of a fifth degree polynomial, for which in general no closed-form solution exists. We therefore expand in the small parameter $\veps=v\gas-v\liq$, i.e.
$ \tilde{\v A}_{\mu} = \tilde{\v A}_{\mu}^{(0)} + \veps \tilde{\v A}_{\mu}^{(1)} + \ldots$,
 and
$\lambda = \lambda^{(0)} + \veps \lambda^{(1)} + \ldots $.
Herein, the lowest-order system matrix reads, taking $\bar v = \phi\gas v\gas + \phi\liq v\liq$,
\begin{align}
  \tilde{\v A}_\mu^{(0)} = \begin{bmatrix}
    \bar v & \phi\gas \xi\gas & \phi\liq \xi\liq & 0 & 0 \\
    0 & \bar v & 0 & \rho\gas^{-1} & 0 \\
    0 & 0 & \bar v & 0 & \rho\liq^{-1} \\
    0 & \rho\gas c\gas^2 - \xi\gas^2 /(\alpha\gas \kappa_\mu^{(0)}) & \xi\gas \xi\liq / (\alpha\gas \kappa_\mu^{(0)}) & \bar v & 0 \\
    0 & \xi\gas \xi\liq / (\alpha\liq \kappa_\mu^{(0)}) & \rho\liq c\liq^2 - \xi\liq^2 /(\alpha\liq \kappa_\mu^{(0)}) & 0 & \bar v
  \end{bmatrix} ,
\end{align}
where we have used the lowest-order term of $\kappa_\mu$, as defined in \cref{eq:def_kappa_mu}:
\begin{align}
  \kappa_\mu^{(0)} = \tfrac{T\gas s\gas^2}{\mCpg} + \tfrac{T\liq s\liq^2}{\mCpl} + \tfrac{\xi\gas^2}{m\gas c\gas^2} + \tfrac{\xi\liq^2}{m\liq c\liq^2}.
\end{align}
To the lowest order in $\veps$, $v\gas = v\liq = \bar v = v$, and thus the eigenvalue problem consists in finding the roots of $\det(\tilde{\v A}_\mu^{(0)}-\lambda^{(0)} \v I )=0$. Hence, the full vector of eigenvalues is given by
\begin{align}
  \boldsymbol{\lambda}_\mu^{(0)} = \{ v - c_{\mu,+}, v - c_{\mu,-}, v, v, v + c_{\mu,-}, v + c_{\mu,+}\}
\end{align}
where the two sound speeds in the $\mu$-model are given by
\begin{equation}
  c^2_{\mu,\pm} = \frac{ \left( \tfrac{T\gas
        s\gas^2}{\mCpg} + \tfrac{T\liq
        s\liq^2}{\mCpl} \right) ( c\gas^2 +
      c\liq^2 ) + \tfrac{\xi\liq^2 c\gas^2}{m\liq
      c\liq^2} + \tfrac{\xi\gas^2 c\liq^2}{m\gas c\gas^2} 
    \pm \sqrt{ \left[ \left( \tfrac{T\gas s\gas^2}{\mCpg} + \tfrac{T\liq s\liq^2}{\mCpl}
        \right) ( c\gas^2 - c\liq^2 ) + \tfrac{\xi\liq^2
          c\gas^2}{m\liq c\liq^2} - \tfrac{\xi\gas^2
          c\liq^2}{m\gas c\gas^2} \right]^2 + 4
      \tfrac{\xi\gas^2 \xi\liq^2}{m\gas m\liq} } }{ 2
  \left[ \tfrac{T\gas s\gas^2}{\mCpg} + \tfrac{T\liq s\liq^2}{\mCpl} +
    \tfrac{\xi\gas^2}{m\gas c\gas^2} + \tfrac{\xi\liq^2}{m\liq
      c\liq^2}\right] }
.
  \label{eq:c_mupm}
\end{equation}



\begin{prop}\label{prop:SC_mu_0}
  The $\mu$-model satisfies the weak subcharacteristic condition with
  respect to the basic model of \cref{sec:basic_model}, given only the
  physically fundamental conditions $\rho_k,\Cpk,T_k > 0$ for $k\in\{\Gas,\Liq\}$, in the equilibrium state defined by \cref{eq:v-equil}.
\end{prop}
\begin{proof}
  We first note that $c_{\mu,\pm}^2 \in \mathbb R$ on the given conditions, and that $c_{\mu,+}^2 \geq 0$.
  The product of the sound speeds may be written as
  \begin{align}
    c_{\mu,+}^{-2} c_{\mu,-}^{-2} = c_{0,+}^{-2} c_{0,-}^{-2} + \Z_\mu^0, \quad \textrm{where} \quad
    \Z_\mu^0 = c\gas^{-2} c\liq^{-2} \frac{ \tfrac{\xi\gas^2}{m\gas c\gas^2} + \tfrac{\xi\liq^2}{m\liq c\liq^2} }{ \tfrac{T\gas s\gas^2}{\mCpg} + \tfrac{T\liq s\liq^2}{\mCpl} }.
  \end{align}
  Given the conditions we have that $\Z_\mu^0 \geq 0$, and hence
  \begin{align}
    0 \leq c_{\mu,+}^2 c_{\mu,-}^2 \leq c_{0,+}^2 c_{0,-}^2.
    \label{eq:SC_mu_0_order_1}
  \end{align}
  Therefore also $c_{0,-}^2$ is positive, and thus we have that $c_{0,\pm}$
  are real and, by choice, positive.

  Now, the product of the differences of the sound speeds may be
  written as
  \begin{gather}
    (c_{0,+}^2 - c_{\mu,+}^2)(c_{0,+}^2 - c_{\mu,-}^2)(c_{0,-}^2 - c_{\mu,+}^2)(c_{0,-}^2 - c_{\mu,-}^2) = - \Q_\mu^0,
    \label{eq:SC_mu_0_prod}
  \end{gather}
  where
  \begin{gather}
    \Q_\mu^0 = \left(c\gas^2-c\liq^2\right)^2 \tfrac{\xi\gas^2 \xi\liq^2}{m\gas m\liq}
  \left[ \tfrac{T\gas s\gas^2}{\mCpg} + \tfrac{T\liq s\liq^2}{\mCpl} + \tfrac{\xi\gas^2}{m\gas c\gas^2} + \tfrac{\xi\liq^2}{m\liq c\liq^2}\right]^{-2}.
  \end{gather}
  Clearly, with the given conditions, $\Q_\mu^0 \geq 0$, and hence the only ordering of sound speeds compatible with \cref{eq:SC_mu_0_prod,eq:SC_mu_0_order_1} is $ 0 \leq c_{\mu,-} \leq c_{0,-} \leq c_{\mu,+} \leq c_{0,+}$, which means that the weak subcharacteristic condition is satisfied.
\end{proof}

\begin{prop}\label{prop:SC_vmu_mu}
  The $v\mu$-model satisfies the subcharacteristic condition with respect to the $\mu$-model, subject only to the physically fundamental conditions $\rho_k, \Cpk, T_k > 0$, for $k \in \{\Gas,\Liq\}$.
\end{prop}
\begin{proof}
The sound speed in the $v\mu$-model is given by \cite{lund_hierarchy_2012}
\begin{align}
  c_{v\mu}^2 = \frac{1}{\rho}\frac{ m\gas c\gas^2 m\liq c\liq^2 \left( \frac{\xi\gas}{m\gas c\gas^2} + \frac{\xi\liq}{m\liq c\liq^2} \right)^2  + \left( \frac{T\gas s\gas^2}{\mCpg} + \frac{T\liq s\liq^2}{\mCpl} \right)\left(m\gas c\gas^2 + m\liq c\liq^2 \right) }{\frac{T\gas s\gas^2}{\mCpg} + \frac{T\liq s\liq^2}{\mCpl} + \frac{\xi\gas^2}{m\gas c\gas^2} + \frac{\xi\liq^2}{m\liq c\liq^2}}.
\end{align}
We now consider the product of the differences in the sound speeds of the two models, which may be written as
\begin{gather}
  (c^2_{\mu,+}-c_{v\mu}^2)(c^2_{\mu,-}-c_{v\mu}^2) 
  = - \Q_{v\mu}^\mu ,
  \label{eq:SC_vmu_mu_prod}
\end{gather}
where
\begin{gather}
  \Q_{v\mu}^\mu = Y\gas Y\liq \left[ \frac{\left( \frac{T\gas s\gas^2}{m\gas \Cpg} + \frac{T\liq s\liq^2}{m\liq \Cpl} \right)( c\gas^2-c\liq^2 ) - \frac{\xi\gas^2 c\liq^2}{m\gas c\gas^2} + \frac{\xi\liq c\gas^2}{m\liq c\liq^2} + \left(\frac{1}{m\gas}-\frac{1}{m\liq} \right) \xi\gas \xi\liq }{\frac{T\gas s\gas^2}{\mCpg} + \frac{T\liq s\liq^2}{\mCpl} + \frac{\xi\gas^2}{m\gas c\gas^2} + \frac{\xi\liq^2}{m\liq c\liq^2}}\right]^2.
\end{gather}
Clearly $\Q_{v\mu}^\mu \geq 0$. Hence, exactly one of the factors on the left hand side of \cref{eq:SC_vmu_mu_prod} must be negative, which gives $c_{\mu,-} \leq c_{v\mu} \leq c_{\mu,+}$, i.e.~the subcharacteristic condition is satisfied.
\end{proof}

%% file: pmu_model.sec.tex
\section{The $p\mu$-model}\label{sec:pmu-model}
We now consider the model which arises when we impose instantaneous mechanical-chemical equilibrium, i.e.~we let the relaxation parameters $\J, \K \to \infty$, which we expect corresponds to
\begin{align}
  p\gas = p\liq \equiv p \quad \textrm{and} \quad \mu\gas = \mu\liq \equiv \mu.
\end{align}
Simultaneously, $I\gas = \J (p\gas-p\liq)$ and $K\gas = \K (\mu\liq-\mu\gas)$ should remain finite.
We now seek explicit expressions for these terms in order to find the governing equations of the model.

In the following analysis we use the parameter set stated in \cref{sec:basic_model} and therefore let the interfacial pressure jump $\Dip = p\intph-p = 0$. From \cref{eq:Dp_k,eq:Dmu} we have
\begin{align}
  \D_k p 
   &= - \tfrac{\rho_k c_k^2}{\alpha_k} \left( \tilde I_k + \pd {\alpha_k v_k} x \right) + \tfrac{\xi_k +  \tfrac{1}{2} \Gamma_k (\Div_k)^2}{\alpha_k} K_k + \tfrac{\Gamma_k}{\alpha_k} \Dpv_k M_k + \tfrac{\Gamma_k}{\alpha_k} H_k ,
  \label{eq:Dp_pmu} \\
  \D_k \mu 
  &= - \tfrac{\xi_k}{\alpha_k} \left( \tilde I_k + \pd {\alpha_k v_k} x \right) +
  \left[ \tfrac{\xi_k^2}{m_k c_k^2} + \tfrac{T_k s_k^2}{\mCpk} + \tfrac{1}{2} (\Div_k)^2 \psi_k \right] K_k + \psi_k \Dpv_k M_k + \psi_k H_k .
  \label{eq:Dmu_pmu}
\end{align}
where we have defined $\tilde I_k = I_k - \vp \pd{\alpha_k}{x} = \pd{\alpha_k}{t}$.

Eqs.~\cref{eq:Dp_pmu,eq:Dmu_pmu} evaluated for each phase now constitute a $4 \times 4$ system which is straightforward to solve for the four unknowns $\pdinl p t$, $\pdinl \mu t$, $\tilde I\gas$, and $K\gas$, in terms of spatial derivatives and the remaining source terms. The final expressions for the latter two are
\begin{align}
  \tilde I\gas &= - \mcl P_p^{p\mu} (v\gas-v\liq) \pd p x - \mcl
                 G_p^{p\mu} (v\gas-v\liq) \pd \mu x - \Phi\gas \pd
                 {\alpha\gas v\gas} x + \Phi\liq \pd {\alpha\liq
                 v\liq} x + \mcl I_v^{p\mu} M\gas + \mcl I_T^{p\mu}
                 H\gas, \label{eq:I_pmu} \\
  K\gas &= - \mcl P^{p\mu}_{\mu} (v\gas-v\liq) \pd p x - \mcl
          G^{p\mu}_\mu  (v\gas-v\liq) \pd \mu x - \mcl
          V^{p\mu}_{\mu,\Gas} \pd {\bar v} x + \mcl K_v^{p\mu} M\gas +
          \mcl K_T^{p\mu}  H\gas, \label{eq:K_pmu}
\end{align}
where the coefficients are given in \cref{sec:coeff_pmu}.

\subsection{Governing equations}
Inserting the expressions \cref{eq:I_pmu,eq:K_pmu} into the basic model of \cref{sec:basic_model}, we are now in a position to state the full model. The mechanical--chemical equilibrium model may thus be formulated as follows.

\begin{itemize}
\item Conservation of mass: $\pd {\rho} t + \pd {\left(m\gas v\gas + m\liq v\liq \right)}{x} = 0$,
\item Momentum balance:
  \begin{multline}
    \pd {m\gas  v\gas} t + \pd {m\gas v\gas^2} x + \left( \alpha\gas + v\intph \mcl P^{p\mu}_{\mu} (v\gas-v\liq) \right) \pd p x + v\intph \mcl G^{p\mu}_\mu  (v\gas-v\liq) \pd \mu x + v\intph \mcl V^{p\mu}_{\mu,\Gas} \pd {\bar v} x = \left( 1 + v\intph \mcl K_v^{p\mu} \right) M\gas + v\intph \mcl K_T^{p\mu}  H\gas , \label{eq:momg_pmu} 
  \end{multline}
\item Energy balance:
  \begin{multline}
    \pd {E\gas} t + \pd { E\gas v\gas}{x} + \left[
      \alpha\gas v\gas + \left(\left( \mu + \tfrac{1}{2}v\intph^2
        \right) \mcl P^{p\mu}_{\mu} - p \mcl P_p^{p\mu} \right)
      (v\gas-v\liq) \right] \pd {p} x  + \left[ \left( \mu +
        \tfrac{1}{2}v\intph^2 \right) \mcl G^{p\mu}_\mu - p \mcl
      G_p^{p\mu} \right] (v\gas-v\liq) \pd \mu x \\
    + \left[ \left( \mu + \tfrac{1}{2}v\intph^2 \right) \mcl V^{p\mu}_{\mu,\Gas} + p \Phi\liq \right] \pd {\bar v} x
    = \left[ \vm + \left( \mu + \tfrac{1}{2}v\intph^2 \right) \mcl
      K_v^{p\mu} - p \mcl I_v^{p\mu} \right] M\gas + \left[ 1 + \left(
        \mu + \tfrac{1}{2}v\intph^2 \right) \mcl K_T^{p\mu} - p \mcl
      I_T^{p\mu}   \right] H\gas. \label{eq:eng_pmu} 
  \end{multline}
\end{itemize}
The momentum and energy equations for the liquid phase are found by phase symmetry.

\subsection{Wave velocities}
We now wish to write the system in a quasilinear form, in order to find the wave speeds of the system, in the homogeneous limit where we let the relaxation terms $\M,\H\to 0$. To this end, we will express the model in the vector of unknowns $\v u_{p\mu} = [p, \mu, \bar v, v\gas, v\liq ]$. We therefore seek the evolution equations for the elements of~$\v u_{p\mu}$.

For the volume evolution, we find, using \cref{eq:Dalpha_k,eq:I_pmu}, that
\begin{align}
  \pd {\alpha\gas} t + \mcl P_p^{p\mu} (v\gas-v\liq) \pd p x + \mcl G_p^{p\mu} (v\gas-v\liq) \pd \mu x + \Phi\gas \pd {\alpha\gas v\gas} x - \Phi\liq \pd {\alpha\liq v\liq} x= 0,
  \label{eq:Dalpha_pmu}
\end{align}
For the volume-averaged velocity $\bar v$ we find, using \cref{eq:Dv_k,eq:K_pmu,eq:I_pmu,eq:Dalpha_pmu}, that
\begin{multline}
  \pd {\bar v} t
  + ( \alpha\gas\rho\gas^{-1} + \alpha\liq \rho\liq^{-1} +
    P^{p\mu}_{\bar v} \veps^2 ) \pd p x 
  + G^{p\mu}_{\bar v} \veps^2 \pd \mu x
  + \alpha\gas \veps \pd {v\gas} x
  - \alpha\liq \veps \pd {v\liq} x
  + \left( \Phi\gas  v\gas + \Phi\liq v\liq - V^{p\mu}_{\bar v,\Gas} \veps \right)  \pd {\bar v} x
  = 0 ,
\end{multline}
where we have defined the shorthand coefficients $P^{p\mu}_{\bar v}, G^{p\mu}_{\bar v}, V^{p\mu}_{\bar v,\Gas}$ (for which expressions are given in \cref{sec:coeff_pmu}), used $\veps=v\gas-v\liq$, and inserted $\beta\gas = 1-\beta\liq = \sqrt{T\liq}/(\sqrt{T\gas}+\sqrt{T\liq})$.
Now, for the pressure and chemical potentials, we get from \cref{eq:Dp_pmu,eq:Dmu_pmu} that
\begin{gather}
  \pd p t + \left( \Psi^p\gas v\gas + \Psi^p\liq v\liq \right) \pd p x + G^{p\mu}_p \veps \pd \mu x + V^{p\mu}_{p} \pd {\bar v} x = 0,\\
  \pd \mu t + P^{p\mu}_\mu \veps \pd p x + \left( \Psi^\mu\gas v\gas + \Psi^\mu\liq v\liq \right) \pd \mu x + V^{p\mu}_\mu \pd {\bar v} x = 0.
\end{gather}
Again, the coefficients are given in \cref{sec:coeff_pmu}.

The homogeneous system in a quasilinear form thus reads $\pd {\v u_{p\mu}} t + \v A_{p\mu}\left( \v u_{p\mu} \right) \pd {\v
    u_{p\mu}} x = 0$, where the system Jacobian is given by
\begin{align}
  \v A _{p\mu} = \begin{bmatrix}
    \Psi^p\gas v\gas + \Psi^p\liq v\liq & G_p^{p\mu} \veps & V_p^{p\mu} & 0 & 0 \\
    P_\mu^{p\mu} \veps & \Psi^\mu\gas v\gas + \Psi^\mu\liq v\liq &
    V_\mu^{p\mu} & 0 & 0 \\
    \frac{\alpha\gas}{\rho\gas} + \frac{\alpha\liq}{\rho\liq} + P^{p\mu}_{\bar v} \veps^2 &
    G^{p\mu}_{\bar v} \veps^2 & \Phi\gas  v\gas + \Phi\liq v\liq - V^{p\mu}_{\bar v,\Gas} \veps & \alpha\gas \veps & - \alpha\liq \veps \\
    \frac{1}{\rho\gas} - \frac{\beta\liq \mcl P^{p\mu}_{\mu}}{m\gas}
    \veps^2 & - \frac{\beta\liq \mcl G^{p\mu}_\mu }{m\gas} \veps^2 & -
    \frac{\beta\liq \mcl V^{p\mu}_{\mu,\Gas}}{m\gas} \veps & v\gas & 0
    \\
    \frac{1}{\rho\liq} - \frac{\beta\gas \mcl P^{p\mu}_{\mu}}{m\liq}
    \veps^2 & - \frac{\beta\gas \mcl G^{p\mu}_\mu }{m\liq} \veps^2 &
    \frac{\beta\gas \mcl V^{p\mu}_{\mu,\Gas}}{m\liq} \veps & 0 & v\liq 
  \end{bmatrix}.
\end{align}
Obtaining the assocated eigenvalues exactly by analytic means is again
unfeasible, as the problem consists in finding the roots of a
fifth-degree polynomial. We therefore expand in $\veps$: $\v A_{p\mu} = \v A_{p\mu}^{(0)} + \veps \v A_{p\mu}^{(1)} + \veps^2 \v A_{p\mu}^{(2)} + \ldots$,
where it is assumed that the matrices $\v A_{p\mu}^{(i)}$ are independent of $\veps$. To the lowest order, where $\veps \to 0$, taking $v = v\gas = v\liq$, we get the matrix
\begin{align}
  \v A _{p\mu}^{(0)} = \begin{bmatrix}
    v & 0 & V_p^{p\mu,(0)} & 0 & 0 \\
    0 & v &
    V_\mu^{p\mu,(0)} & 0 & 0 \\
    \frac{\alpha\gas}{\rho\gas} + \frac{\alpha\liq}{\rho\liq} &
    0 & v & 0 & 0 \\
    \frac{1}{\rho\gas} & 0 & 0 & v & 0 \\
    \frac{1}{\rho\liq} & 0 & 0 & 0 & v
  \end{bmatrix},
  \label{eq:jacobi_pmu_0}
\end{align}
where the superscript ``(0)'' on the coefficients signifies the zeroth-order expansion in $\veps$, such that 
\begin{align}
   V_p^{p\mu, (0)} 
  = \frac{\frac{T\gas s\gas^2}{\mCpg} + \frac{T\liq s\liq^2}{\mCpl}}{\left( \frac{\alpha\gas}{\rho\gas c\gas^2 }  + \frac{\alpha\liq}{\rho\liq c\liq^2 } \right) \left( \frac{T\gas s\gas^2}{\mCpg} + \frac{T\liq s\liq^2}{\mCpl}\right)
   + \left( \frac{\xi\gas}{\rho\gas c\gas^2} - \frac{\xi\liq}{\rho\liq c\liq^2}
   \right)^2}.
 \end{align}
The eigenvalues in the $p\mu$-model are, to the lowest order in $\veps$,
\begin{align}
  \boldsymbol{\lambda}_{p\mu}^{(0)} = \left\lbrace v - c_{p\mu}, v, v, v, v + c_{p\mu} \right\rbrace,
\end{align}
where we have identified the sound speed $c_{p\mu}$ of the model, given by
\begin{align}
  c_{p\mu}^2 = \frac{\left(\frac{\alpha\gas}{\rho\gas} + \frac{\alpha\liq}{\rho\liq}\right)\left(\frac{T\gas s\gas^2}{\mCpg} + \frac{T\liq s\liq^2}{\mCpl} \right)}{\left( \frac{\alpha\gas}{\rho\gas c\gas^2 }  + \frac{\alpha\liq}{\rho\liq c\liq^2 } \right) \left( \frac{T\gas s\gas^2}{\mCpg} + \frac{T\liq s\liq^2}{\mCpl}\right)
   + \left( \frac{\xi\gas}{\rho\gas c\gas^2} - \frac{\xi\liq}{\rho\liq c\liq^2}
   \right)^2}.
  \label{eq:c_pmu}
\end{align}

\begin{prop}\label{prop:SC_pmu_p}
  The $p\mu$-model satisfies the weak subcharacteristic condition with respect to the $p$-model, given only the physically fundamental conditions $\rho_k,\Cpk, T_k > 0$, for $k\in\{\Gas,\Liq\}$, in the equilibrium state defined by \cref{eq:v-equil}.
\end{prop}
\begin{proof}
  From \cref{eq:c_pmu,eq:c_p}, we observe that we may write
\begin{align}
  c_{p\mu}^{-2} 
  = c_p^{-2} + \Z_{p\mu}^p, \quad \textrm{where} \quad
  \Z_{p\mu}^p = \frac{\left( \tfrac{\xi\gas}{\rho\gas c\gas^2} - \tfrac{\xi\liq}{\rho\liq c\liq^2} \right)^2}{ \left( \tfrac{\alpha\gas}{\rho\gas} + \tfrac{\alpha\liq}{\rho\liq} \right) \left(\tfrac{T\gas s\gas^2}{\mCpg} + \tfrac{T\liq s\liq^2}{\mCpl}\right) } .
\end{align}
Due to the given physical conditions, $\Z_{p\mu}^p \geq 0$, and hence $ 0 \leq c_{p\mu} \leq c_{p} $, i.e.~the subcharacteristic condition is satisfied.
\end{proof}

\begin{prop}\label{prop:SC_pmu_mu}
  The $p\mu$-model satisfies the weak subcharacteristic condition with
  respect to the $\mu$-model, under the physically fundamental
  conditions $\rho_k, \Cpk, T_k > 0$, for $k\in\{\Gas,\Liq\}$, in the equilibrium state defined by \cref{eq:v-equil}.
\end{prop}
\begin{proof}
  Using the expressions \cref{eq:c_pmu,eq:c_mupm} for the sound speeds in the two models, we may write
  \begin{gather}
  (c^2_{\mu,+}-c_{p\mu}^2)(c^2_{\mu,-}-c_{p\mu}^2) = - \Q_{p\mu}^\mu ,
    \label{eq:SC_pmu_mu_prod}\end{gather}
  where
  \begin{equation}
  \Q_{p\mu}^\mu  
    = \frac{ \frac{\alpha\gas \alpha\liq}{\rho\gas c\gas^2 \rho\liq
        c\liq^2} \left( \frac{T\gas s\gas^2}{\mCpg} + \frac{T\liq s\liq^2}{\mCpl} \right)
      \left[ \left(\frac{T\gas s\gas^2}{\mCpg}+\frac{T\liq s\liq^2}{\mCpl}\right) \left(
          c\gas^2-c\liq^2 \right)-\left( \frac{\xi\gas}{\rho\gas
            c\gas^2}-\frac{\xi\liq}{\rho\liq
            c\liq^2}\right)\left(\frac{\xi\gas
            c\liq^2}{\alpha\gas}+\frac{\xi\liq c\gas^2}{\alpha\liq}
        \right) \right]^2}{\left[ \frac{T\gas s\gas^2}{\mCpg} + \frac{T\liq s\liq^2}{\mCpl} + \frac{\xi\gas^2}{m\gas c\gas^2} + \frac{\xi\liq^2}{m\liq c\liq^2} \right] \left[ \left( \frac{\alpha\gas}{\rho\gas c\gas^2 }  + \frac{\alpha\liq}{\rho\liq c\liq^2 } \right) \left( \frac{T\gas s\gas^2}{\mCpg} + \frac{T\liq s\liq^2}{\mCpl}\right)
   + \left( \frac{\xi\gas}{\rho\gas c\gas^2} - \frac{\xi\liq}{\rho\liq c\liq^2}
   \right)^2 \right]^{2}} .
  \end{equation}
Clearly, on the given conditions, $\Q_{p\mu}^\mu \geq 0$. Therefore, exactly one factor on the left hand side of \cref{eq:SC_pmu_mu_prod} is negative, and hence
$c_{\mu,-} \leq c_{p\mu} \leq c_{\mu,+}$,
so the subcharacteristic condition is satisfied.
\end{proof}

\begin{prop}\label{prop:SC_vpmu_pmu}
  The $vp\mu$-model satisfies the subcharacteristic condition with
  respect to the $p\mu$-model, given the physically fundamental conditions $\rho_k, \Cpk, T_k > 0$.
\end{prop}
\begin{proof}
The sound speed in the $vp\mu$-model is given by \cite{flatten_relaxation_2011,lund_hierarchy_2012}
\begin{align}
  c_{vp\mu}^2 = \frac{1}{\rho}\frac{\frac{T\gas s\gas^2}{\mCpg} + \frac{T\liq s\liq^2}{\mCpl}}{\left( \frac{\alpha\gas}{\rho\gas c\gas^2 }  + \frac{\alpha\liq}{\rho\liq c\liq^2 } \right) \left( \frac{T\gas s\gas^2}{\mCpg} + \frac{T\liq s\liq^2}{\mCpl}\right)
   + \left( \frac{\xi\gas}{\rho\gas c\gas^2} - \frac{\xi\liq}{\rho\liq c\liq^2}
   \right)^2}.
  \label{eq:c_vpmu}
\end{align}
Now, we may write
\begin{align}
  c_{vp\mu}^{-2} 
  &= c_{p\mu}^{-2} + \Z_{vp\mu}^{p\mu}, \qquad \textrm{where} \qquad
  \Z_{vp\mu}^{p\mu} 
  = \tfrac{\alpha\gas \alpha\liq}{\rho\gas \rho\liq} \left( \rho\liq - \rho\gas \right)^2 c_{p\mu}^{-2},
\end{align}
which is clearly positive, due to the given conditions. Thus, $0 \leq
c_{vp\mu} \leq c_{p\mu} $, i.e.~the subcharacteristic condition is satisfied.
\end{proof}

\begin{rmk}\label{rmk:relation}
By direct comparison of \cref{eq:c_vpmu,eq:c_pmu}, we find the ratio
\begin{align}
  \frac{c_{p\mu}}{c_{vp\mu}} = \sqrt{ \rho \left( \frac{\alpha\gas}{\rho\gas} + \frac{\alpha\liq}{\rho\liq} \right) }.
  \label{eq:ratio_c_vpmu_pmu}
\end{align}
This is exactly the same ratio as has been shown to hold for other models associated with $v$-relaxation in the $p$-branch of the hierarchy \cite{morin_two-fluid_2013,martinez_ferrer_effect_2012}. We can thus extend the relation
\begin{align}
  \frac{c_p}{c_{vp}} = \frac{c_{pT}}{c_{vpT}} = \frac{c_{pT\mu}}{c_{vpT\mu}} =\frac{c_{p\mu}}{c_{vp\mu}},
  \label{eq:relation_extended}
\end{align}
by the newly obtained ratio \eqref{eq:ratio_c_vpmu_pmu} between the sound speeds of the $vp\mu$- and $p\mu$-models.
\end{rmk}

\begin{prop}\label{prop:SC_pTmu_pmu}
  The $pT\mu$-model satisfies the weak subcharacteristic condition
  with respect to the $p\mu$-model, given the physically fundamental conditions $\rho_k, \Cpk, T > 0$, in the equilibrium state defined by \cref{eq:v-equil}.
\end{prop}
\begin{proof}
  In the equilibrium state defined by the $pT\mu$-model, we have
  $T\gas=T\liq \equiv T$. The sound velocity in the $pT\mu$-model is
  given in~\cite{morin_two-fluid_2013}, and may be rewritten as
  \begin{equation}
    c_{pT\mu}^2 = \frac{ \frac{\alpha\gas}{\rho\gas} + \frac{\alpha\liq}{\rho\liq} }{ \frac{\alpha\gas}{\rho\gas c\gas^2} +
      \frac{\alpha\liq}{\rho\liq c\liq^2} +
      \mCpg T \left[ \frac{1}{\Delta h}\left(\frac{1}{\rho\liq} - \frac{1}{\rho\gas} \right) + \frac{\Gamma\gas}{\rho\gas c\gas^2}  \right]^2
      + \mCpl T \left[ \frac{1}{\Delta h} \left( \frac{1}{\rho\gas} - \frac{1}{\rho\liq} \right) - \frac{\Gamma\liq}{\rho\liq c\liq^2} \right]^2
    }  ,
    \label{eq:c_pTmu}
  \end{equation}
  where we have introduced the enthalpy difference $\Delta h = h\gas-h\liq$.

  We may reorganize the last equality in \cref{eq:relation_extended}
  to yield
  \begin{align}
    \frac{c_{p\mu}}{c_{pT\mu}} = \frac{c_{vp\mu}}{c_{vpT\mu}}.
    \label{eq:relation_reorg}
  \end{align}
  Flåtten and Lund \cite{flatten_relaxation_2011} showed that the
  subcharacteristic condition is satisfied between the models on the
  right hand side, i.e.~that $0 \leq c_{vpT\mu} \leq c_{vp\mu}$.
  The same must hold for the models on the left hand side of
  \cref{eq:relation_reorg}, i.e.~$0 \leq c_{pT\mu} \leq c_{p\mu}$, and hence the subcharacteristic condition is satisfied.
  In particular, we may write the sound speed as
  \begin{gather}
    c_{pT\mu}^{-2} = c_{p\mu}^{-2} + \Z_{pT\mu}^{p\mu},\end{gather}
  where
  \begin{gather}
    \Z_{pT\mu}^{p\mu} = \mCpg \mCpl T \frac{\left[ \frac{1}{\Delta h} \left( \frac{1}{\rho\liq}-\frac{1}{\rho\gas} \right) \left(\frac{s\gas}{\mCpg}+\frac{s\liq}{\mCpl}\right) + \frac{\Gamma\gas}{\rho\gas c\gas^2} \frac{s\liq}{\mCpl} + \frac{\Gamma\liq}{\rho\liq c\liq^2} \frac{s\gas}{\mCpg} \right]^2}{\left( \frac{\alpha\gas}{\rho\gas} + \frac{\alpha\liq}{\rho\liq} \right)\left(\frac{s\gas^2}{\mCpg}+\frac{s\liq^2}{\mCpl}\right)} .
  \end{gather}
  Clearly, $\Z_{pT\mu}^{p\mu} \geq 0$ based on the given conditions.
\end{proof}

%% file: Tmu_model.sec.tex
\section{The $T\mu$-model}\label{sec:Tmu-model}
We now investigate the model which arises when we assume instantaneous thermal-chemical equilibrium, i.e.~let the relaxation parameters $\K,\H \to \infty$, which expectedly corresponds to
\begin{align}
  T\gas = T\liq \equiv T \quad \textrm{and} \quad \mu\gas = \mu\liq \equiv \mu .
\end{align}
The products $H\gas = \H(T\liq-T\gas)$ and $K\gas = \K(\mu\liq-\mu\gas)$ remain finite, and may be expressed in terms of spatial derivatives and remaining source terms. In the forthcoming, we seek explicit expressions for these terms to insert into the basic model of \cref{sec:basic_model}.

The equilibrium conditions are contained in \cref{eq:Dmu,eq:DT_1}.
These may be combined to yield
\begin{equation}
  K\gas = - \mcl A^{T\mu}_{\mu} \pd {\alpha\gas} x - \mcl G^{T\mu}_\mu \veps \pd \mu x - \mcl T^{T\mu}_\mu \veps \pd T x - \mcl V^{T\mu}_{\mu,\Gas} \pd {\alpha\gas v\gas} x + \mcl V^{T\mu}_{\mu,\Liq} \pd {\alpha\liq v\liq} x + \mcl K^{T\mu}_p \tilde I\gas + \mcl K^{T\mu}_v \veps M\gas
  \label{eq:K_Tmu}
\end{equation}
where the coefficients are given in \cref{sec:coeff_Tmu}.

\subsection{Governing equations}
We are now in a position to state the $T\mu$-model in its entirety, by inserting \cref{eq:K_Tmu} into the basic model of \cref{sec:basic_model}. The model can be expressed by the following equation set:
\begin{itemize}
\item Volume advection: $\pd {\alpha\gas} t + \vp \pd {\alpha\gas} x = I\gas$,
\item Conservation of mass: $\pd {\rho} t + \pd {\left(\alpha\gas \rho\gas v\gas + \alpha\liq \rho\liq v\liq \right)}{x} = 0$,
\item Conservation of momentum:
  \begin{multline}
    \pd {\alpha\gas\rho\gas v\gas} t + \pd {(\alpha\gas\rho\gas
      v\gas^2 + \alpha\gas p\gas)} x + v\intph \big[ \mcl G^{T\mu}_\mu (v\gas-v\liq) \pd \mu x + \mcl
      T^{T\mu}_\mu (v\gas-v\liq) \pd T x + \mcl V^{T\mu}_{\mu,\Gas} \pd
      {\alpha\gas v\gas} x - \mcl V^{T\mu}_{\mu,\Liq} \pd {\alpha\liq
        v\liq} x \big] \\+ \left( v\intph^2 \left(\mcl V^{T\mu}_{\mu,\Gas} + \mcl V^{T\mu}_{\mu,\Liq}\right) - p\intph \right) \pd
    {\alpha\gas} x  =  v\intph \mcl K^{T\mu}_p I\gas
      + \left(1 + v\intph \mcl K^{T\mu}_v (v\gas-v\liq) \right) M\gas ,
  \end{multline}
\item Conservation of energy: $\pd {E} t + \pd {\left( E\gas v\gas + E\liq v\liq + \alpha\gas
        v\gas p\gas + \alpha\liq v\liq p\liq \right)} x = 0$.
\end{itemize}

\subsection{Wave velocities}
We now seek the wave velocities of the model in the homogeneous limit,
where $\J,\M\to 0$.
As usual, we are interested in the zeroth-order expansion in
$\varepsilon = v\gas-v\liq$.\footnote{Strictly speaking, exact eigenvalues may be found analytically in this model, since noting that $\alpha\gas$ is a characteristic variable reduces the eigenvalue problem to finding the solutions of a fourth-degree polynomial, which is analytically tractable.} We may therefore directly evaluate the evolution equations in this limit, and take $v\gas = v\liq = v$ if they are outside the differential operator.

After some tedious, but fairly straightforward algebra, we find that to the lowest order in $\veps$, the wave velocities of the $T\mu$-model are given by
\begin{align}\label{eq:eigenvals_Tmu}
  \boldsymbol{\lambda}_{T\mu}^{(0)} = \left\lbrace v - c_{T\mu,+}, v - c_{T\mu,-}, v, v + c_{T\mu,-}, v + c_{T\mu,+} \right\rbrace.
\end{align}
Herein, $c_{T\mu,\pm}$ are the sound speeds of this model, which may be expressed by
\begin{multline}
  c_{T\mu,\pm}^2
  = \tfrac{1}{2}
  \Biggl\lbrace \tfrac{\Delta h^2 m\gas m\liq \left( c\gas^{2} + c\liq^{2} \right)}{\mCpg \mCpl T^2 c\gas^2 c\liq^2} + \tfrac{m\liq + m\gas \left( 1 + \frac{\Gamma\liq}{c\liq^2} \Delta h \right)^2 }{\mCpg T} 
    +  \tfrac{m\gas + m\liq \left( 1 - \frac{\Gamma\gas}{c\gas^2} \Delta h \right)^2 }{\mCpl T} \pm \\
    \Biggl[ \Biggl(
      \tfrac{\Delta h^2 m\gas m\liq \left( c\gas^{2} - c\liq^{2} \right)}{\mCpg \mCpl T^2 c\gas^2 c\liq^2} 
      - \tfrac{m\liq - m\gas \left( 1 + \frac{\Gamma\liq}{c\liq^2} \Delta h \right)^2}{\mCpg T}
      + \tfrac{m\gas - m\liq \left( 1 - \frac{\Gamma\gas}{c\gas^2} \Delta h \right)^2}{\mCpl T}
      \Biggr)^2
  + 4 m\gas m\liq \Biggl( \tfrac{1 + \frac{\Gamma\liq}{c\liq^2} \Delta h }{\mCpg T}
      + \tfrac{1 - \frac{\Gamma\gas}{c\gas^2} \Delta h }{\mCpl T}
      \Biggr)^2 \Biggr]^{\frac{1}{2}}
    \Biggr\rbrace  \\ \times \Biggl[
 \tfrac{\Delta h^2 m\gas m\liq}{\mCpg \mCpl T^2 c\gas^2 c\liq^2}
  + \tfrac{\tfrac{m\liq}{c\liq^2} + \tfrac{m\gas}{c\gas^2} \left( 1 + \tfrac{\Gamma\liq}{c\liq^2} \Delta h \right)^2}{\mCpg T}
  + \tfrac{\tfrac{m\gas}{c\gas^2} + \tfrac{m\liq}{c\liq^2} \left( 1 - \tfrac{\Gamma\gas}{c\gas^2} \Delta h \right)^2 }{\mCpl T }  
 + \left(\tfrac{\Gamma\liq}{c\liq^2} - \tfrac{\Gamma\gas}{c\gas^2} - \tfrac{\Gamma\gas\Gamma\liq}{c\gas^2 c\liq^2} \Delta h \right)^2 \Biggr]^{-1}.
 \label{eq:soundvel_Tmu}
\end{multline}

\begin{prop}\label{prop:SC_Tmu_T}
  The $T\mu$-model satisfies the weak subcharacteristic condition with
  respect to the $T$-model, given the physically fundamental
  conditions $\rho_k, \Cpk, T > 0$, in the equilibrium state defined by \cref{eq:v-equil}.
\end{prop}

\begin{proof}
  We may write
  \begin{align}
    \left( c_{T,+}^2 - c_{T\mu,+}^2 \right) \left( c_{T,+}^2 - c_{T\mu,-}^2 \right) \left( c_{T,-}^2 - c_{T\mu,+}^2 \right) \left( c_{T,-}^2 - c_{T\mu,-}^2 \right) = - \Q_{T\mu}^{T},
  \end{align}
  where
  \begin{multline}
    \Q_{T\mu}^{T} = m\gas m\liq \Biggl[ \left( \tfrac{1}{\mCpl T} + \tfrac{1+\Gamma\liq c\liq^{-2}\Delta h}{\mCpg T} \right) \left( \tfrac{1}{\mCpg T} + \tfrac{1-\Gamma\gas c\gas^{-2}\Delta h}{\mCpl T} + \tfrac{\Gamma\liq}{m\liq} \left( \tfrac{\Gamma\liq}{c\liq^2} - \tfrac{\Gamma\gas}{c\gas^2} - \tfrac{\Gamma\gas \Gamma\liq}{c\gas^2 c\liq^2} \Delta h \right) \right) c\gas^2 \\- \left( \tfrac{1}{\mCpg T} + \tfrac{1-\Gamma\gas c\gas^{-2}\Delta h}{\mCpl T} \right)\left( \tfrac{1}{\mCpl T} + \tfrac{1+\Gamma\liq c\liq^{-2}\Delta h}{\mCpg T} - \tfrac{\Gamma\gas}{m\gas} \left( \tfrac{\Gamma\liq}{c\liq^2} - \tfrac{\Gamma\gas}{c\gas^2} - \tfrac{\Gamma\gas \Gamma\liq}{c\gas^2 c\liq^2} \Delta h \right)\right) c\liq^2  \Biggr]^2 \\ \times \Biggl[ \Biggl(
 \tfrac{\Delta h^2 m\gas m\liq}{\mCpg \mCpl T^2 c\gas^2 c\liq^2}
  + \tfrac{1}{\mCpg T} \left( \tfrac{m\liq}{c\liq^2} + \tfrac{m\gas}{c\gas^2} \left( 1 + \tfrac{\Gamma\liq}{c\liq^2} \Delta h \right)^2
 \right)
  + \tfrac{1}{\mCpl T } \left( \tfrac{m\gas}{c\gas^2} + \tfrac{m\liq}{c\liq^2} \left( 1 - \tfrac{\Gamma\gas}{c\gas^2} \Delta h \right)^2
 \right) \\
  + \left(\tfrac{\Gamma\liq}{c\liq^2} - \tfrac{\Gamma\gas}{c\gas^2} - \tfrac{\Gamma\gas\Gamma\liq}{c\gas^2 c\liq^2} \Delta h \right)^2 \Biggr) \left( \tfrac{1}{\mCpg T} + \tfrac{1}{\mCpl T} + \tfrac{\Gamma\liq^2}{m\liq c\liq^2} + \tfrac{\Gamma\gas^2}{m\gas c\gas^2} \right) \Biggr]^{-2}.
  \end{multline}
  Moreover, we may write
  \begin{align}
    \left( c_{T,+}^2 + c_{T,-}^2 \right) - \left( c_{T\mu,+}^2 + c_{T\mu,-}^2 \right) = \Z_{T\mu}^{T}
  \end{align}
  where
  \begin{multline}
    \Z_{T\mu}^T = \Biggl[ \Bigl( \tfrac{1}{\mCpg T} +
    \tfrac{1-\frac{\Gamma\gas} {c\gas^2}\Delta h}{\mCpl T} +
    \tfrac{\Gamma\liq}{m\liq} \left( \tfrac{\Gamma\liq}{c\liq^2} -
      \tfrac{\Gamma\gas}{c\gas^2} - \tfrac{\Gamma\gas
        \Gamma\liq}{c\gas^2 c\liq^2} \Delta h \right) \Bigr) \tfrac{
      m\liq c\gas^2}{c\liq^2} \\+ \Bigl( \tfrac{1}{\mCpl T} +
    \tfrac{1+\tfrac{\Gamma\liq}{c\liq^2}\Delta h}{\mCpg T} -
    \tfrac{\Gamma\gas}{m\gas} \left( \tfrac{\Gamma\liq}{c\liq^2} -
      \tfrac{\Gamma\gas}{c\gas^2} - \tfrac{\Gamma\gas
        \Gamma\liq}{c\gas^2 c\liq^2} \Delta h \right)\Bigr) \tfrac{
      m\gas c\liq^2}{c\gas^2} \Biggr] \\
 \times \Biggl[ \Biggl(
 \tfrac{\Delta h^2 m\gas m\liq}{\mCpg \mCpl T^2 c\gas^2 c\liq^2}
  + \tfrac{ \frac{m\liq}{c\liq^2} + \frac{m\gas}{c\gas^2} \left( 1 + \frac{\Gamma\liq}{c\liq^2} \Delta h \right)^2
 }{\mCpg T}
  + \tfrac{\frac{m\gas}{c\gas^2} + \frac{m\liq}{c\liq^2} \left( 1 - \frac{\Gamma\gas}{c\gas^2} \Delta h \right)^2}{\mCpl T }
  + \left(\tfrac{\Gamma\liq}{c\liq^2} - \tfrac{\Gamma\gas}{c\gas^2} - \tfrac{\Gamma\gas\Gamma\liq}{c\gas^2 c\liq^2} \Delta h \right)^2 \Biggr) \\ \times \left( \tfrac{1}{\mCpg T} + \tfrac{1}{\mCpl T} + \tfrac{\Gamma\liq^2}{m\liq c\liq^2} + \tfrac{\Gamma\gas^2}{m\gas c\gas^2} \right) \Biggr]^{-1}.
  \end{multline}
  Clearly, $\Z_{T\mu}^T \geq 0$ and $\Q_{T\mu}^T \geq 0$ on grounds of
  the given
  conditions. 
  The only possible ordering of sound speeds is thus $ 0 \leq c_{T\mu,-} \leq c_{T,-} \leq c_{T\mu,+} \leq c_{T,+}$, i.e.~the weak subcharacteristic condition is satisfied.
\end{proof}

\begin{prop}\label{prop:SC_Tmu_mu}
  The $T\mu$-model satisfies the weak subcharacteristic condition with
  respect to the $\mu$-model, given the physically fundamental conditions $\rho_k,\Cpk, T > 0$, in the equilibrium state defined by \cref{eq:v-equil}.
\end{prop}
\begin{proof}
  We note that we may write
  \begin{align}
    \left( c_{\mu,+}^2 - c_{T\mu,+}^2 \right) \left( c_{\mu,+}^2 - c_{T\mu,-}^2 \right) \left( c_{\mu,-}^2 - c_{T\mu,+}^2 \right) \left( c_{\mu,-}^2 - c_{T\mu,-}^2 \right) = - \Q_{T\mu}^{\mu},
  \end{align}
  where
  \begin{multline}
    \Q_{T\mu}^{\mu} = m\gas m\liq \Biggl[ \Bigl( \tfrac{s\liq}{\mCpl} + \tfrac{s\gas \left( 1 + \tfrac{\Gamma\liq}{c\liq^2} \Delta h \right)}{\mCpg} \Bigr) \Bigl( \tfrac{s\gas}{\mCpg} + \tfrac{s\liq \left( 1- \tfrac{\Gamma\gas}{c\gas^2} \Delta h \right)}{\mCpl}  - \left( \tfrac{\Gamma\liq}{c\liq^2} - \tfrac{\Gamma\gas}{c\gas^2} - \tfrac{\Gamma\gas \Gamma\liq}{c\gas^2 c\liq^2} \Delta h \right) \tfrac{\xi\liq}{m\liq} \Bigr) c\gas^2  \\ - \Biggl(  \tfrac{s\gas}{\mCpg} + \tfrac{s\liq \left( 1 - \tfrac{\Gamma\gas}{c\gas^2} \Delta h \right)}{\mCpl} \Biggr)\Biggl( \tfrac{s\liq}{\mCpl} + \tfrac{s\gas \left( 1 + \tfrac{\Gamma\liq}{c\liq^2} \Delta h \right)}{\mCpg} + \left( \tfrac{\Gamma\liq}{c\liq^2} - \tfrac{\Gamma\gas}{c\gas^2} - \tfrac{\Gamma\gas \Gamma\liq}{c\gas^2 c\liq^2} \Delta h \right) \tfrac{\xi\gas}{m\gas} \Biggr) c\liq^2 \Biggr]^2 \\\times \Biggl[ \Biggl(
 \tfrac{\Delta h^2 m\gas m\liq}{\mCpg \mCpl T^2 c\gas^2 c\liq^2}
  + \tfrac{1}{\mCpg T} \left( \tfrac{m\liq}{c\liq^2} + \frac{m\gas}{c\gas^2} \left( 1 + \tfrac{\Gamma\liq}{c\liq^2} \Delta h \right)^2
 \right)
  + \tfrac{1}{\mCpl T } \left( \tfrac{m\gas}{c\gas^2} + \tfrac{m\liq}{c\liq^2} \left( 1 - \tfrac{\Gamma\gas}{c\gas^2} \Delta h \right)^2
 \right) \\
  + \left(\tfrac{\Gamma\liq}{c\liq^2} - \tfrac{\Gamma\gas}{c\gas^2} - \tfrac{\Gamma\gas\Gamma\liq}{c\gas^2 c\liq^2} \Delta h \right)^2 \Biggr) \left( \tfrac{T s\gas^2}{\mCpg} + \tfrac{T s\liq^2}{\mCpl} + \tfrac{\xi\gas^2}{m\gas c\gas^2} + \tfrac{\xi\liq^2}{m\liq c\liq^2} \right) \Biggr]^{-2} .
  \end{multline}
  Now, we may also write
$    \left( c_{\mu,+}^2 + c_{\mu,-}^2 \right) - \left( c_{T\mu,+}^2 + c_{T\mu,-}^2 \right) = \Z_{T\mu}^{\mu}$,
where
  \begin{multline}
    \Z_{T\mu}^{\mu} = \Biggl\{ \Biggl[ \tfrac{s\gas}{\mCpg} + \tfrac{s\liq\left( 1- \tfrac{\Gamma\gas}{c\gas^2} \Delta h \right)}{\mCpl}  - \left( \tfrac{\Gamma\liq}{c\liq^2} - \tfrac{\Gamma\gas}{c\gas^2} - \tfrac{\Gamma\gas \Gamma\liq}{c\gas^2 c\liq^2} \Delta h \right) \tfrac{\xi\liq}{m\liq} \Biggr]^2 \tfrac{m\liq c\gas^2}{c\liq^2} \\ + \Biggl[ \tfrac{s\liq}{\mCpl} + \tfrac{s\gas \left( 1 + \tfrac{\Gamma\liq}{c\liq^2} \Delta h \right)}{\mCpg} + \left( \tfrac{\Gamma\liq}{c\liq^2} - \tfrac{\Gamma\gas}{c\gas^2} - \tfrac{\Gamma\gas \Gamma\liq}{c\gas^2 c\liq^2} \Delta h \right) \tfrac{\xi\gas}{m\gas} \Biggr]^2 \tfrac{m\gas c\liq^2}{c\gas^2} \Biggr\} \\ 
    \times \Biggl[ \Biggl(
 \tfrac{\Delta h^2 m\gas m\liq}{\mCpg \mCpl T^2 c\gas^2 c\liq^2}
  + \tfrac{1}{\mCpg T} \left( \tfrac{m\liq}{c\liq^2} + \tfrac{m\gas}{c\gas^2} \left( 1 + \tfrac{\Gamma\liq}{c\liq^2} \Delta h \right)^2
 \right) \\ + \tfrac{1}{\mCpl T } \left( \tfrac{m\gas}{c\gas^2} + \tfrac{m\liq}{c\liq^2} \left( 1 - \tfrac{\Gamma\gas}{c\gas^2} \Delta h \right)^2
 \right) 
  + \left(\tfrac{\Gamma\liq}{c\liq^2} - \tfrac{\Gamma\gas}{c\gas^2} - \tfrac{\Gamma\gas\Gamma\liq}{c\gas^2 c\liq^2} \Delta h \right)^2 \Biggr)  \left( \tfrac{T s\gas^2}{\mCpg} + \tfrac{T s\liq^2}{\mCpl} + \tfrac{\xi\gas^2}{m\gas c\gas^2} + \tfrac{\xi\liq^2}{m\liq c\liq^2} \right) \Biggr]^{-1}.
  \end{multline}
  Clearly, $\Q_{T\mu}^{\mu} \geq 0$ and $\Z_{T\mu}^{\mu} \geq 0$ for the given conditions. The only possible ordering of the sound speeds is therefore
  $    0 \leq c_{T\mu,-} \leq c_{\mu,-} \leq c_{T\mu,+} \leq c_{\mu,+}$,
  i.e.~the eigenvalues of the relaxed model are interlaced between the eigenvalues of the parent model, and the weak subcharacteristic condition is satisfied.
\end{proof}

\begin{prop}\label{prop:SC_vTmu_Tmu}
The $vT\mu$-model satisfies the subcharacteristic condition with
respect to the $T\mu$-model, given the physically fundamental conditions $\rho_k, \Cpk, T > 0$.
\end{prop}
\begin{proof}
  The sound speed in the $vT\mu$-model is given by \cite{lund_hierarchy_2012}
  \begin{equation}
    c_{vT\mu}^2 = \frac{
       \tfrac{ \Delta h^2 m\gas m\liq c_v^2 }
             { \mCpg \mCpl T^2 c\gas^2 c\liq^2 \rho } 
       + \tfrac{ \left( \rho + m\gas \frac{\Gamma\liq}{c\liq^2} \Delta h \right)^2 }
               {\mCpg T}
       + \tfrac{\left( \rho - m\liq \frac{\Gamma\gas }{c\gas^2} \Delta h \right)^2 }
               { \mCpl T} 
    }{
        \tfrac{ \Delta h^2 m\gas m\liq }
              { \mCpg \mCpl T^2 c\gas^2 c\liq^2 } 
     + \tfrac{ \frac{m\liq}{c\liq^2} + \frac{m\gas}{c\gas^2} \left(1 + \frac{\Gamma\liq}{c\liq^2} \Delta h \right)^2}
                {\mCpg T}
        + \tfrac{ \frac{m\gas}{c\gas^2} + \frac{m\liq}{c\liq^2} \left(1 - \frac{\Gamma\gas}{c\gas^2} \Delta h \right)^2}
                {\mCpl T}
        + \left(
                \tfrac{\Gamma\liq}{c\liq^2} - \tfrac{\Gamma\gas}{c\gas^2} - \tfrac{\Gamma\gas\Gamma\liq}{c\gas^2 c\liq^2} \Delta h 
          \right)^2
    }
  \end{equation}
  We may now write the product of the differences between the sound
  speeds as
  \begin{align}\label{eq:SC_prod_vTmu_Tmu}
    \left( c_{T\mu,+}^2 - c_{vT\mu}^2 \right) \left( c_{T\mu,-}^2 - c_{vT\mu}^2 \right) = -\Q_{vT\mu}^{T\mu} ,
  \end{align}
where
\begin{equation}
  \Q_{vT\mu}^{T\mu} = Y\gas Y\liq 
  \left[ \frac{  \tfrac{\Delta
      h^2 m\gas m\liq \left(
    c\gas^{2} - c\liq^{2} \right)}{\mCpg \mCpl T^2 c\gas^2 c\liq^2} + \tfrac{\Gamma\liq \Delta h}{\mCpg T
      c\liq^2} \left( m\liq + m\gas \left( 1 +
    \tfrac{\Gamma\liq}{c\liq^2} \Delta h \right) \right) +
    \tfrac{\Gamma\gas \Delta h}{\mCpl T c\gas^2} \left( m\gas + m\liq
    \left( 1 - \tfrac{\Gamma\gas}{c\gas^2} \Delta h \right) \right)
     }{
      \tfrac{\Delta h^2 m\gas m\liq}{\mCpg \mCpl
      T^2 c\gas^2 c\liq^2} + \tfrac{\tfrac{m\liq}{c\liq^2} + \tfrac{m\gas}{c\gas^2} \left( 1 +
    \tfrac{\Gamma\liq}{c\liq^2} \Delta h \right)^2}{\mCpg T} 
     +
    \tfrac{\tfrac{m\gas}{c\gas^2} +
    \tfrac{m\liq}{c\liq^2} \left( 1 - \tfrac{\Gamma\gas}{c\gas^2} \Delta
    h \right)^2 }{\mCpl T } + \left(\tfrac{\Gamma\liq}{c\liq^2} -
    \tfrac{\Gamma\gas}{c\gas^2} - \tfrac{\Gamma\gas\Gamma\liq}{c\gas^2
      c\liq^2} \Delta h \right)^2 } \right]^{2} .
\end{equation}
Due to the given conditions, it is clear that $\Q_{vT\mu}^{T\mu} \geq 0$, and thus exactly one of the factors on the left hand side of \cref{eq:SC_prod_vTmu_Tmu} is negative.
Hence, the sound speeds must be ordered as
 $ c_{T\mu,-} \leq c_{vT\mu} \leq c_{T\mu,+} $,
i.e.~the subcharacteristic condition is satisfied.
\end{proof}

\begin{prop}\label{prop:SC_pTmu_Tmu}
  The $pT\mu$-model satisfies the weak subcharacteristic condition with
  respect to the $T\mu$-model, subject to the physically fundamental
  conditions
  $ \rho_k,
    \Cpk,
    T > 0$,
  in the equilibrium state defined by \cref{eq:v-equil}.
\end{prop}
\begin{proof}
  The sound speed in the $pT\mu$-model is 
  stated in \cref{eq:c_pTmu}.
  We note that we may write
  \begin{align}\label{eq:prod_pTmu_Tmu}
    \left( c_{T\mu,+}^2 - c_{pT\mu}^2 \right) \left( c_{T\mu,-}^2 -
    c_{pT\mu}^2 \right) = -\Q_{pT\mu}^{T\mu},
  \end{align}
  where
  \begin{multline}
    \Q_{pT\mu}^{T\mu} 
    = \tfrac{\mCpg \mCpl T^2}{\Delta h^2 \rho\gas^2 \rho\liq^2} \\ \times \Biggl[ \tfrac{\Delta h^2 m\gas m\liq}{\mCpg \mCpl T^2 c\gas^2 c\liq^2}
      ( c\gas^2 - c\liq^2 )
      + \tfrac{\rho\gas \rho\liq}{\mCpl T} \left(\tfrac{1}{\rho\liq} - \tfrac{1 - \frac{\Gamma\gas}{c\gas^2} \Delta h }{\rho\gas} \right)   \left( 1 - \alpha\liq \tfrac{\Gamma\gas}{c\gas^2} \Delta h \right)
      - \tfrac{\rho\gas \rho\liq}{\mCpg T} \left( \tfrac{1}{\rho\gas} - \tfrac{ 1 +
      \tfrac{\Gamma\liq}{c\liq^2} \Delta h }{\rho\liq} \right) \left(
      1 + \alpha\gas \tfrac{\Gamma\liq}{c\liq^2} \Delta h \right)
      \Biggr]^2 \\ \times
     \Biggl[ \tfrac{\Delta h^2 m\gas m\liq}{\mCpg
        \mCpl T^2 c\gas^2 c\liq^2} + \tfrac{
      \tfrac{m\liq}{c\liq^2} + \tfrac{m\gas}{c\gas^2} \left( 1 +
      \tfrac{\Gamma\liq}{c\liq^2} \Delta h \right)^2}{\mCpg T} +
      \tfrac{\frac{m\gas}{c\gas^2} +
      \tfrac{m\liq}{c\liq^2} \left( 1 - \tfrac{\Gamma\gas}{c\gas^2}
      \Delta h \right)^2}{\mCpl T }  + \left(\tfrac{\Gamma\liq}{c\liq^2} -
      \tfrac{\Gamma\gas}{c\gas^2} - \tfrac{\Gamma\gas\Gamma\liq}{c\gas^2
        c\liq^2} \Delta h \right)^2 \Biggr]^{-1} \\ \times \Biggl[
        \tfrac{\alpha\gas}{\rho\gas c\gas^2} +
        \tfrac{\alpha\liq}{\rho\liq c\liq^2} + \tfrac{\mCpg T}{\Delta
          h^2} \Biggl( \tfrac{1}{\rho\liq} - \tfrac{1 -
          \tfrac{\Gamma\gas}{c\gas^2} \Delta h}{\rho\gas} \Biggr)^2 +
        \tfrac{\mCpl T}{\Delta h^2} \Biggl( \tfrac{1}{\rho\gas} - \tfrac{
          1 + \tfrac{\Gamma\liq}{c\liq^2} \Delta h }{\rho\liq}
        \Biggr)^2 \Biggr]^{-2}.
  \end{multline}
  Due to the given conditions, it is clear that
  $\Q_{pT\mu}^{T\mu} \geq 0$. Hence, exactly one of the factors on the
  left hand side of \cref{eq:prod_pTmu_Tmu} must be negative, and
  therefore
 $   c_{T\mu,-} \leq c_{pT\mu} \leq c_{T\mu,+} $,
  i.e.~the weak subcharacteristic condition is satisfied.
\end{proof}

%% file: comparison.sec.tex
\section{Comparison and discussion of models}\label{sec:comparison}
In this section, we compare the models in the hierarchy.
We first show plots for relevant cases, and then briefly discuss physical and numerical aspects of the different models.

  

\subsection{Comparison of sound speeds}
We now present plots of the sound speeds in all the models in the hierarchy, for different physically relevant conditions, in order to illustrate the effect of different equilibrium assumptions.
Plots with the same physical parameters were presented by Lund \cite{lund_hierarchy_2012} for the $v$-branch of the hierarchy, building on plots by \citet{flatten_relaxation_2011}.
\citet{martinez_ferrer_effect_2012} and \citet{morin_two-fluid_2013} presented similar plots for the $p$-branch of the hierarchy.
In the following, results for the complete hierarchy are shown.

The main panel of \cref{fig:sound_vel_water} shows the fluid-mechanical sound speeds of all the models in the hierarchy for a water-steam mixture at atmospheric conditions.
The thermophysical parameters are shown in \cref{tab:water}.
Apart from the fact that the subcharacteristic condition is always respected, we notice that there are mainly two equilibrium assumptions that affect the propagation speeds, namely those of $p$- and $v$-relaxation, respectively.
First, imposing instantaneous equilibrium in pressure attracts the sound velocities to the lower bound of the parent models, which can be seen in the insets of \cref{fig:sound_vel_water}.
Further imposing velocity equilibrium, the sound velocity is reduced by a factor (see \cref{rmk:relation})
\begin{align}
  \sqrt{ \rho \left( \frac{\alpha\gas}{\rho\gas} + \frac{\alpha\liq}{\rho\liq} \right) } \simeq \sqrt{ \alpha\gas \alpha\liq \frac{\rho\liq}{\rho\gas}}.
\end{align}
The approximation made is valid when $\rho\gas \ll \rho\liq$, which is the case here.
Hence, these equilibrium assumptions seem to have severe impact on the wave velocities, in particular when the density difference between the phases is large.
\begin{figure}[htb]
  \centering
  \includegraphics[width=0.7\textwidth]{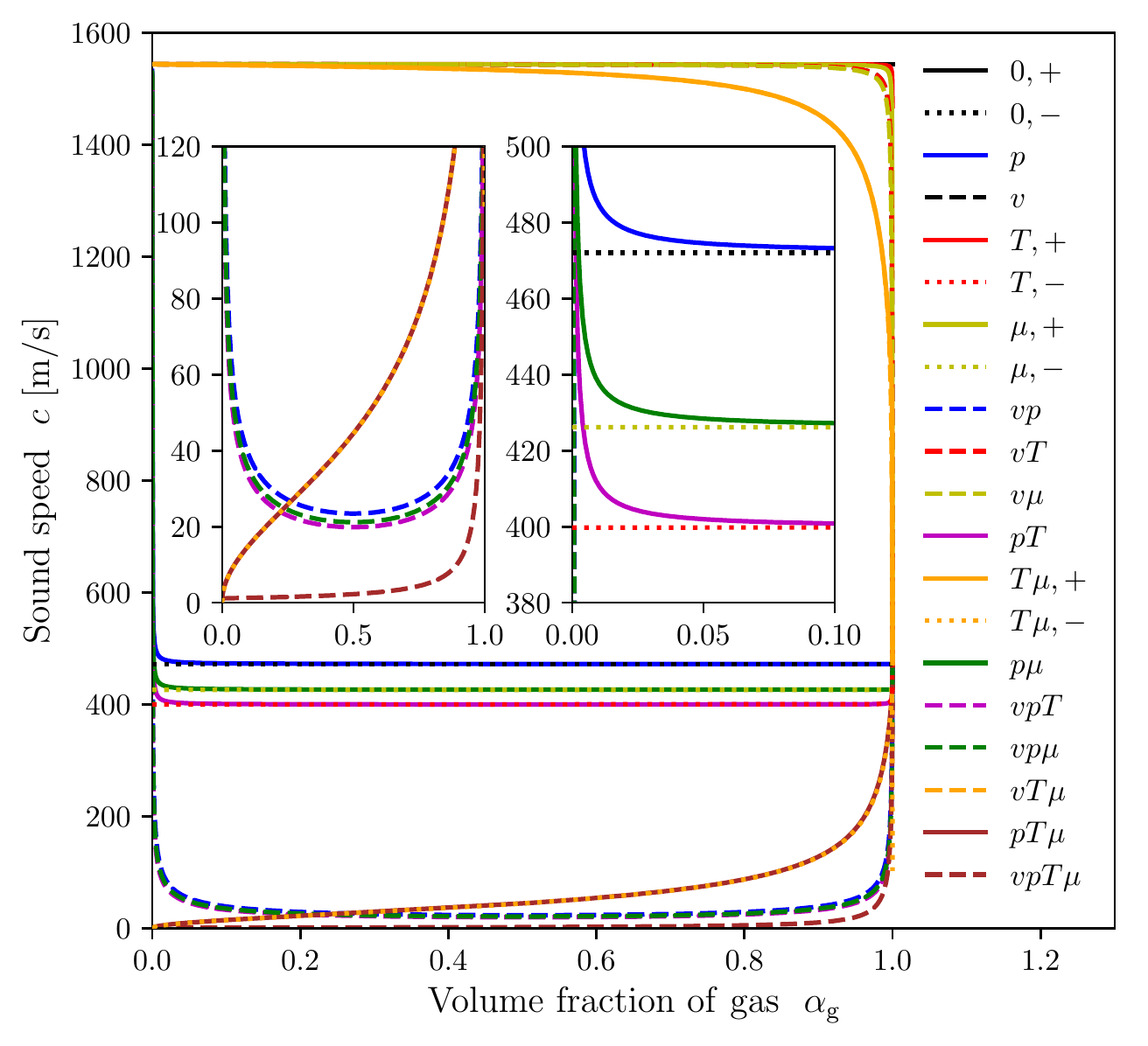}
  \caption{Sound velocities in a water-steam mixture at atmospheric conditions.
  The insets show close-ups of parts of the plots in the main panel.}
\label{fig:sound_vel_water}
\end{figure}
\begin{table}[htb]
\small
\caption{Parameters for a water-steam mixture at atmospheric pressure.}
\label{tab:water}
\centering
\begin{tabular}{lllcc}
  \toprule
  Quantity & Symbol & Unit & Gas & Liquid \\ \midrule
  Pressure & $p$ & \si{\mega\pascal}  & 0.1     & 0.1 \\
  Temperature & $T$ & \si{\kelvin}    & 372.76  & 372.76\\
  Density & $\rho$ & \si{\kilogram\per\cubic\meter}  & 0.59031 & 958.64 \\
  Speed of sound & $c$ & \si{\meter\per\second}   & 472.05  & 1543.4 \\
  Heat capacity & $C_p$ & \si{\joule\per\kilogram\per\kelvin} & 2075.9  & 4216.1 \\
  Entropy & $s$ & \si{\square\meter\per\square\second\per\kelvin}
  & 7358.8 & 1302.6 \\
  Grüneisen coefficient & $\Gamma$ & (--) & 0.33699 & 0.4 \\
\bottomrule
\end{tabular}
\end{table}

\begin{figure}[tbh]
  \centering
  \includegraphics[width=0.60\textwidth]{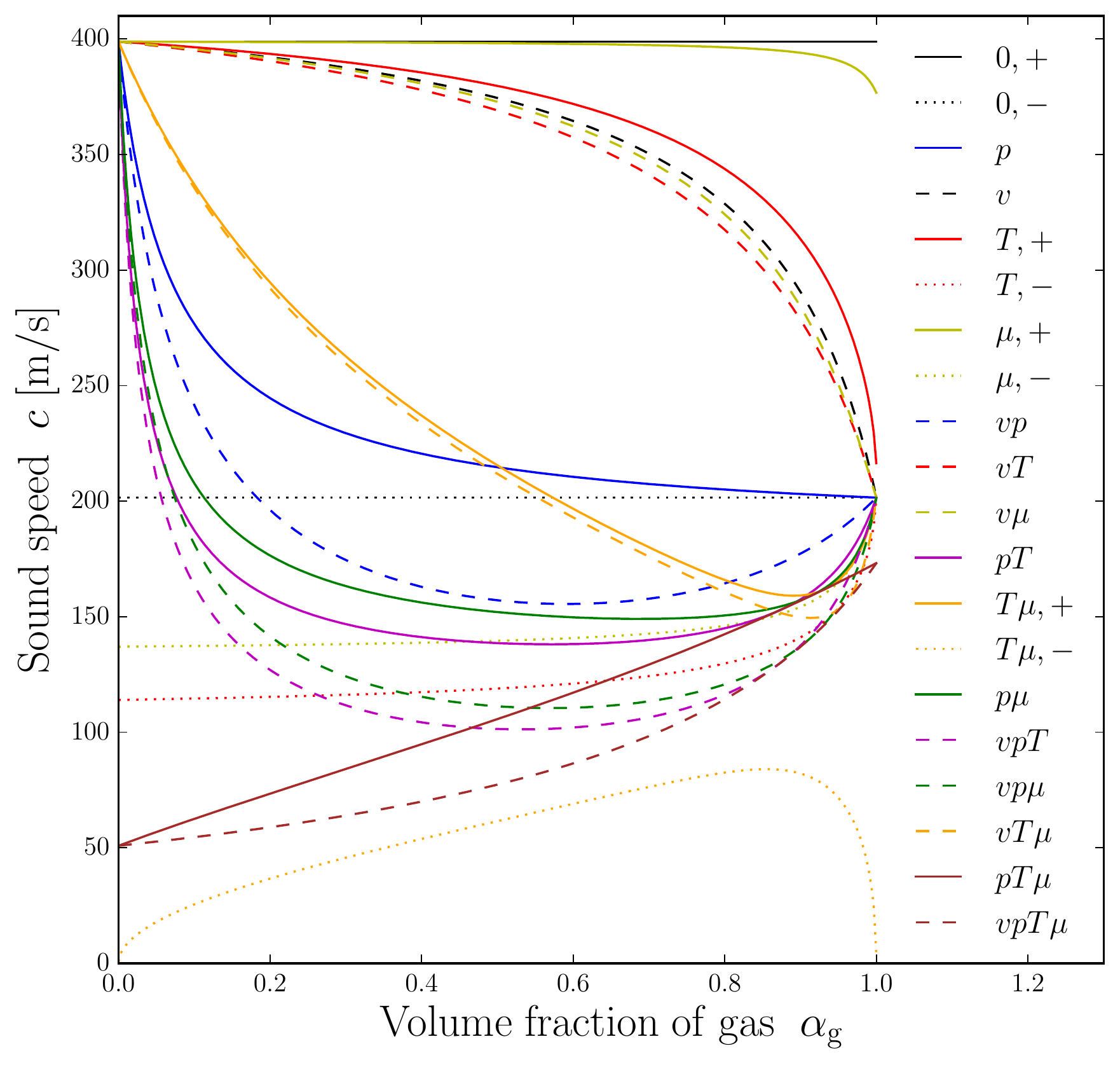}
  \caption{Sound speeds in a two-phase \coto{} mixture at 50 bar.}
\label{fig:sound_vel_CO2}
\end{figure}

In \cref{fig:sound_vel_CO2}, the sound speeds in the entire hierarchy are plotted for a \coto{} mixture at \SI{50}{\bar}, whose thermophysical properties are given in \cref{tab:coto}.
By close inspection, it can be seen that the subcharacteristic condition is everywhere respected.
In particular, the sound speeds of an equilibrium system are always interlaced between the sound speeds in the parent models.
Again, the pressure relaxation has the most prominent effect on the sound speed, but also combining thermal and chemical equilibrium seems to have a strong effect.

\subsection{Discontinuous sound speeds}
All the models considered in the present paper are only strictly valid when the gas fraction $\alpha\gas \in (0, 1)$.
One would expect the sound speeds of the models to be continuous at the phase boundary, i.e.~at the transition between single and two-phase flow, in the sense that the two-phase speed of sound should reduce to the single-phase speed of sound in the limit where one phase disappears:
\begin{equation}
  \lim_{\alpha_k \to 1} c_{X} \to c_k
\end{equation}
for a given model $X$ in the hierarchy.
However, some of the models have wave speeds that are discontinuous at the phase boundary.
In particular, this concerns the $pT\mu$- and $vpT\mu$-models, whose sound speeds are discontinuous in both limits $\alpha_k \to 1$, which can be seen directly by evaluating the analytic expressions in these limits (see Refs.~\cite{lund_hierarchy_2012,morin_two-fluid_2013}).


The $T$- and $\mu$-models have ``half-continuous'' sound speeds, in the sense that for the ``$\pm$'' sound waves, only one of them is continuous in the limit $\alpha_k \to 1$.
For the $\mu$-model, taking $\alpha\liq \to 1$ in \eqref{eq:c_mupm} yields
\begin{equation}
  \lim_{\alpha\liq\to 1} c_{\mu,\pm}^2 = \frac{ \tfrac{T\gas
        s\gas^2}{\Cpg} ( c\gas^2 + c\liq^2 ) + \tfrac{\xi\gas^2 c\liq^2}{c\gas^2} 
    \pm \left| \tfrac{T\gas s\gas^2}{\Cpg} ( c\gas^2 - c\liq^2 ) - \tfrac{\xi\gas^2 c\liq^2}{c\gas^2} \right| }{ 2
  \left[ \tfrac{T\gas s\gas^2}{\Cpg} +
    \tfrac{\xi\gas^2}{ c\gas^2} \right] },
\end{equation}
which is equivalent to
\begin{equation}
  \lim_{\alpha\liq\to 1} c_{\mu,+} = \max \left(  \frac{ c\gas }{ \sqrt{ 1 + \tfrac{\xi\gas^2}{ c\gas^2} \tfrac{\Cpg}{T\gas  s\gas^2} } } , c\liq \right), \quad \textrm{and} \quad
  \lim_{\alpha\liq\to 1} c_{\mu,-} = \min \left(  \frac{ c\gas }{ \sqrt{ 1 + \tfrac{\xi\gas^2}{ c\gas^2} \tfrac{\Cpg}{T\gas  s\gas^2} } } , c\liq \right).
\end{equation}
Clearly, only one of these approach the appropriate phasic value $c\liq^2$.
The result limits for $\alpha\gas\to 1$ are found by phase symmetry.
Similarly, we find for the $T$-model, from \cref{eq:c_T}, that
\begin{equation}
  \lim_{\alpha\liq\to 1} c_{T,+} = \min \left( \frac{ c\gas }{ \sqrt{ 1 + \tfrac{\Gamma\gas^2 \Cpg T}{c\gas^2} }}, c\liq \right), \quad \textrm{and} \quad
  \lim_{\alpha\liq\to 1} c_{T,-} = \min \left( \frac{ c\gas }{ \sqrt{ 1 + \tfrac{\Gamma\gas^2 \Cpg T}{c\gas^2} }}, c\liq \right),
\end{equation}
to which the same observation applies.

The remaining sound speeds are continuous at the phase boundary; for the $T\mu$-model in the sense that $\lim_{\alpha_k \to 1} c_{T\mu,+} = c_k$ and $\lim_{\alpha_k \to 1} c_{T\mu,-} = 0$, which can be deduced from the analytic expression \cref{eq:soundvel_Tmu}.

\begin{table}[bth]
\caption{Parameters for a two-phase \coto~mixture at 50 bar.}
\label{tab:coto}
\centering
\begin{tabular}{lllcc}
  \toprule
  Quantity & Symbol & Unit & Gas & Liquid \\ \midrule
  Pressure & $p$ & \si{\mega\pascal}  & 5.0 & 5.0 \\
  Temperature & $T$ & \si{\kelvin}    & 287.43 & 287.43\\
  Density & $\rho$ & \si{\kilogram\per\cubic\meter}  & 156.71 & 827.21 \\
  Speed of sound & $c$ & \si{\meter\per\second}   & 201.54 & 398.89 \\
  Heat capacity & $C_p$ & \si{\joule\per\kilogram\per\kelvin} & 3138.0 & 3356.9 \\
  Entropy & $s$ & \si{\square\meter\per\square\second\per\kelvin} 
  & 1753.9 & 1128.8 \\
  Grüneisen coefficient & $\Gamma$ & (--) & 0.30949 & 0.63175 \\
\bottomrule
\end{tabular}
\end{table}

\subsection{Physical considerations}
It is commonly argued that pressure relaxation is a much faster process than the other relaxation processes \cite{saurel_modelling_2008,zein2010modeling}.
Temperature relaxation, or heat flow, is associated with diffusion, which is an intrinsically slow process.
Chemical potential relaxation, i.e.~mass transfer, is also slow compared to pressure relaxation.
Zein et al.\ \cite{zein2010modeling} provide interesting discussions on the topic and argue that temperature relaxation is faster than chemical relaxation.
Generally, the magnitudes of the different relaxation times may be strongly problem-dependent.
Such considerations may have implications, e.g., for the mass flow through a nozzle, which has been shown to be linked to the subcharacteristic condition \cite{linga14}.

Apart from this, effects not captured by the coarse-grained flow models may come into play, and which model is more accurate may depend heavily on the flow regime under consideration.
The effects that arise from having independent phasic pressures may be of importance for the wave dynamics of the system, and thus models with different pressures may be sensible, even though the associated relaxation time is commonly thought to be comparatively short.
With regards to evaluating the physical relevance of the models presented herein, experimental data on sound speeds in two-phase flow can be found for various systems \cite{xu2000acoustic,kieffer1977sound,herbert2006gasdynamic,tosse2015experimental}.

\subsection{Numerical considerations}
A well-known problem with $p$-relaxed (one-pressure) two-fluid models is that they develop complex eigenvalues when $v\gas \neq v\liq$.
This is commonly resolved e.g.~by adding a regularising pressure which enforces hyperbolicity \cite{bestion_physical_1990,coquel_numerical_1997,evje_hybrid_2003,toumi_approximate_1996,ndjinga2007influence}.
It is worth noting that the two-fluid models with independent phasic pressures, i.e.\ the $T$-, $\mu$- and $T\mu$-models, are locally hyperbolic even for small perturbations away from velocity equilibrium, due to the following argument:
\emph{An eigenvalue of a matrix with real coefficients may only be complex if its complex conjugate is also an eigenvalue.
Since the eigenvalues of the individual phasic pressure models are real and distinct when $\veps=v\gas-v\liq=0$, they must remain so for sufficiently small $\veps$, as the eigenvalues may only become complex in a continuous way.}
In order to determine how large $\veps$ may be before hyperbolicity is lost, we must find the higher-order corrections in  $\veps$ to the eigenvalues, which is beyond the scope of this work.

%% file: conclusion.sec.tex
\section{Conclusions and further work}\label{sec:conclusion}
In this paper, we have presented and completed a hierarchy of relaxation models for two-phase flow, which arises when we impose instantaneous equilibrium in different combinations of velocity, pressure, temperature and chemical potential.
The starting point of the analysis has been the classic seven-equation Baer--Nunziato model \cite{baer_two-phase_1986} equipped with relaxation source terms. 
We have in the present work provided the $T$-, $\mu$-, $p\mu$-, and $T\mu$-models, which represent original contributions to the hierarchy.
Explicit expressions for the sound speeds of these models have been derived.
Using the new expressions and results from the literature, we have shown analytically that the subcharacteristic condition is always satisfied in the hierarchy, given velocity equilibrium between the phases.
To this end, we have contributed with 15 new subcharacteristic conditions, stated in propositions \ref{prop:SC_v_0}--\ref{prop:SC_pTmu_Tmu}.
Out of these, five have been shown in a strong sense, and nine hold in a weak sense, i.e.~given equilibrium in velocity.


In further work, the hierarchy could be extended to multi-component or multi-phase flow.
Moreover, the different models could be implemented and studied numerically for relevant cases (cf.\ \cite{saurel2018diffuse}).
Upon comparison with experimental data, one may then unravel under which conditions the different models are admissible.

%% file: appendix.sec.tex
\section{Coefficients in the $p\mu$-model}\label{sec:coeff_pmu}
The coefficients in the $p\mu$-model are given by
\begin{gather*}
  \mcl P_p^{p\mu} = \tfrac{\alpha\gas\alpha\liq \kappa_\mu}{\rho\gas c\gas^2 \rho\liq c\liq^2 \kappa_{p\mu}}, \quad
  \mcl G_p^{p\mu} = \tfrac{1}{\kappa_{p\mu}} \left[  \tfrac{\alpha\gas}{\rho\gas c\gas^2 }  + \tfrac{\alpha\liq}{\rho\liq c\liq^2 } + \tfrac{\alpha\liq \Gamma\gas T\gas + \alpha\gas\Gamma\liq T\liq}{\rho\gas c\gas^2 \rho\liq c\liq^2}  s^*\right],
  \\
  \mcl P^{p\mu}_{\mu} = \tfrac{\alpha\gas\alpha\liq}{\rho\gas c\gas^2 \rho\liq c\liq^2 \kappa_{p\mu}} \left( \tfrac{\xi\gas}{\alpha\gas}+ \tfrac{\xi\liq}{\alpha\liq}\right) ,\quad
  \mcl G^{p\mu}_\mu = - \tfrac{1}{\kappa_{p\mu}} \left( \tfrac{\alpha\gas}{\rho\gas c\gas^2 }  + \tfrac{\alpha\liq}{\rho\liq c\liq^2 } \right) , \quad
  \mcl V^{p\mu}_{\mu,\Gas} = - \tfrac{1}{\kappa_{p\mu}} \left( \tfrac{\xi\gas}{\rho\gas c\gas^2} - \tfrac{\xi\liq}{\rho\liq c\liq^2} \right),
\end{gather*}
\begin{multline*} \mcl I_v^{p\mu} = \tfrac{\alpha\gas\alpha\liq }{\rho\gas c\gas^2 \rho\liq c\liq^2 \kappa_{p\mu}} \Bigl\{ \Bigl( \tfrac{s\gas}{\mCpg} \left[ \tfrac{\xi\gas}{\alpha\gas}+\tfrac{\xi\liq}{\alpha\liq} \right] + \tfrac{\Gamma\gas}{\alpha\gas} \left[ \tfrac{T\gas s\gas^2}{\mCpg} + \tfrac{T\liq s\liq^2}{\mCpl}\right]  - \tfrac{\Gamma\gas \xi\liq}{\alpha\gas \alpha\liq}\left[ \tfrac{\xi\gas}{\rho\gas c\gas^2} - \tfrac{\xi\liq}{\rho\liq c\liq^2} \right] \\-
\Bigl[ \tfrac{\Gamma\gas \Gamma\liq}{\alpha\gas\alpha\liq} \left( \tfrac{\xi\gas}{\rho\gas c\gas^2} - \tfrac{\xi\liq}{\rho\liq c\liq^2} \right) + \tfrac{\Gamma\gas s\liq}{\alpha\gas \mCpl} - \tfrac{\Gamma\liq s\gas}{\alpha\liq \mCpg} \Bigr] T\liq s^* \Bigr) \Dpv\gas \\
+ \Bigl( \tfrac{s\liq}{\mCpl} \left[
  \tfrac{\xi\gas}{\alpha\gas}+\tfrac{\xi\liq}{\alpha\liq} \right] +
\tfrac{\Gamma\liq}{\alpha\liq} \left[ \tfrac{T\gas s\gas^2}{\mCpg} +
  \tfrac{T\liq s\liq^2}{\mCpl}\right] + \tfrac{\Gamma\liq
  \xi\gas}{\alpha\gas \alpha\liq}\left[ \tfrac{\xi\gas}{\rho\gas
    c\gas^2} - \tfrac{\xi\liq}{\rho\liq c\liq^2} \right] \\
 + \left[ \tfrac{\Gamma\gas \Gamma\liq}{\alpha\gas\alpha\liq} \left( \tfrac{\xi\gas}{\rho\gas c\gas^2} - \tfrac{\xi\liq}{\rho\liq c\liq^2} \right) + \tfrac{\Gamma\gas s\liq}{\alpha\gas \mCpl} - \tfrac{\Gamma\liq s\gas}{\alpha\liq \mCpg} \right] T\gas s^* \Bigr) \Dpv\liq \Bigr\},\end{multline*}
\begin{multline*}
  \mcl I_T^{p\mu} = \tfrac{\alpha\gas\alpha\liq }{\rho\gas c\gas^2 \rho\liq c\liq^2 \kappa_{p\mu}} \Bigl\{ \left( \tfrac{s\gas}{\mCpg} + \tfrac{s\liq}{\mCpl} \right)\left[ \tfrac{\xi\gas}{\alpha\gas}+\tfrac{\xi\liq}{\alpha\liq} \right] + \left( \tfrac{\Gamma\gas}{\alpha\gas} + \tfrac{\Gamma\liq}{\alpha\liq} \right) \left[ \tfrac{T\gas s\gas^2}{\mCpg} + \tfrac{T\liq s\liq^2}{\mCpl}\right] \\ - \tfrac{\Gamma\gas \xi\liq - \Gamma\liq \xi\gas}{\alpha\gas \alpha\liq}\left[ \tfrac{\xi\gas}{\rho\gas c\gas^2} - \tfrac{\xi\liq}{\rho\liq c\liq^2} \right] - \left[ \tfrac{\Gamma\gas \Gamma\liq}{\alpha\gas\alpha\liq} \left( \tfrac{\xi\gas}{\rho\gas c\gas^2} - \tfrac{\xi\liq}{\rho\liq c\liq^2} \right) + \tfrac{\Gamma\gas s\liq}{\alpha\gas \mCpl} - \tfrac{\Gamma\liq s\gas}{\alpha\liq \mCpg} \right] \left( T\liq - T\gas \right) s^* \Bigr\},\end{multline*}
\begin{equation*}
  \mcl K_v^{p\mu} = \tfrac{1}{\kappa_{p\mu}} \Bigl\{
\left[ \left( \tfrac{\alpha\gas}{\rho\gas c\gas^2 }  + \tfrac{\alpha\liq}{\rho\liq c\liq^2 } \right) \tfrac{s\gas}{\mCpg} - \tfrac{\Gamma\gas}{\rho\gas c\gas^2} \left( \tfrac{\xi\gas}{\rho\gas c\gas^2} - \tfrac{\xi\liq}{\rho\liq c\liq^2} \right) \right] \Dpv\gas + \left[ \left( \tfrac{\alpha\gas}{\rho\gas c\gas^2 }  + \tfrac{\alpha\liq}{\rho\liq c\liq^2 } \right) \tfrac{s\liq}{\mCpl} + \tfrac{\Gamma\liq}{\rho\liq c\liq^2} \left( \tfrac{\xi\gas}{\rho\gas c\gas^2} - \tfrac{\xi\liq}{\rho\liq c\liq^2} \right) \right] \Dpv\liq
\Bigr\},
\end{equation*}
\begin{gather*}
  \mcl K_T^{p\mu} = \tfrac{1}{\kappa_{p\mu}} \left\{
\left( \tfrac{\alpha\gas}{\rho\gas c\gas^2 }  + \tfrac{\alpha\liq}{\rho\liq c\liq^2 } \right) \left( \tfrac{s\gas}{\mCpg} + \tfrac{s\liq}{\mCpl} \right) - \left( \tfrac{\Gamma\gas}{\rho\gas c\gas^2} - \tfrac{\Gamma\liq}{\rho\liq c\liq^2} \right) \left( \tfrac{\xi\gas}{\rho\gas c\gas^2} - \tfrac{\xi\liq}{\rho\liq c\liq^2} \right)
\right\},
\end{gather*}
\begin{equation*}
  \Phi\gas = \tfrac{1}{\kappa_{p\mu}} \Bigl\{ \tfrac{\alpha\liq}{\rho\liq c\liq^2 } \left( \tfrac{T\gas s\gas^2}{\mCpg} + \tfrac{T\liq s\liq^2}{\mCpl}\right) -\tfrac{\xi\liq}{\rho\liq c\liq^2}\left( \tfrac{\xi\gas}{\rho\gas c\gas^2} - \tfrac{\xi\liq}{\rho\liq c\liq^2}
  \right) - \left[ \tfrac{\alpha\liq}{\rho\liq c\liq^2 } \left( \tfrac{T\gas s\gas}{\mCpg} + \tfrac{T\liq s\liq}{\mCpl}\right)
  + \left( \tfrac{\xi\gas}{\rho\gas c\gas^2} - \tfrac{\xi\liq}{\rho\liq c\liq^2}
  \right) \tfrac{\Gamma\liq T\liq}{\rho\liq c\liq^2} \right] s^* \Bigr\}.
  \end{equation*}
Herein, we have defined the shorthand denominator
\begin{equation*}
  \kappa_{p\mu} = \left( \tfrac{\alpha\gas}{\rho\gas c\gas^2 }  + \tfrac{\alpha\liq}{\rho\liq c\liq^2 } \right) \left( \tfrac{T\gas s\gas^2}{\mCpg} + \tfrac{T\liq s\liq^2}{\mCpl}\right)
  + \left( \tfrac{\xi\gas}{\rho\gas c\gas^2} - \tfrac{\xi\liq}{\rho\liq c\liq^2}
  \right)^2 - \left[ \left( \tfrac{\alpha\gas}{\rho\gas c\gas^2 }  + \tfrac{\alpha\liq}{\rho\liq c\liq^2 } \right)  \left( \tfrac{T\gas s\gas}{\mCpg} + \tfrac{T\liq s\liq}{\mCpl}\right)
  - \left( \tfrac{\xi\gas}{\rho\gas c\gas^2} - \tfrac{\xi\liq}{\rho\liq c\liq^2}
  \right) \left( \tfrac{\Gamma\gas T\gas}{\rho\gas c\gas^2} - \tfrac{\Gamma\liq T\liq}{\rho\liq c\liq^2}
  \right) \right] s^*
\label{eq:kappa_pmu}
\end{equation*}
and an entropy contribution due to velocity differences
$  s^* = (v\gas-v\liq)^2 / [2(\sqrt{T\gas}+\sqrt{T\liq})^2 ] $.
The coefficients related to the quasi-linear form are given by
\begin{gather*}
  P^{p\mu}_{\bar v} = \mcl P_p^{p\mu} - \left(\tfrac{\beta\liq}{\rho\gas}+\tfrac{\beta\gas}{\rho\liq}\right) \mcl P^{p\mu}_{\mu} , \quad
  G^{p\mu}_{\bar v} = \mcl G_p^{p\mu} - \left( \tfrac{\beta\liq}{\rho\gas} + \tfrac{\beta\gas}{\rho\liq} \right) \mcl G^{p\mu}_\mu , \quad
  V^{p\mu}_{\bar v,\Gas} =  \left( \tfrac{\beta\liq}{\rho\gas} + \tfrac{\beta\gas}{\rho\liq}\right)  \mcl V^{p\mu}_{\mu,\Gas} ,\\
  \Psi^p\gas = \tfrac{1}{\kappa_{p\mu}} \Bigl[  \tfrac{\alpha\gas}{\rho\gas c\gas^2 }  \left( \tfrac{T\gas s\gas (s\gas-s^*)}{\mCpg} + \tfrac{T\liq s\liq (s\liq-s^*)}{\mCpl}\right)
  + \tfrac{\xi\gas}{\rho\gas c\gas^2} \left( \tfrac{\xi\gas+\Gamma\gas T\gas s^*}{\rho\gas c\gas^2} - \tfrac{\xi\liq+\Gamma\liq T\liq s^*}{\rho\liq c\liq^2}
  \right) \Bigr] ,\\
  G^{p\mu}_p = \tfrac{1}{\kappa_{p\mu}} \left[ \tfrac{\xi\gas + \Gamma\gas T\gas s^*}{\rho\gas c\gas^2} - \tfrac{\xi\liq + \Gamma\liq T\liq s^*}{\rho\liq c\liq^2} \right] , \quad
  V^{p\mu}_{p} = \tfrac{1}{\kappa_{p\mu}} \left[ \tfrac{T\gas
      s\gas^2}{\mCpg} + \tfrac{T\liq s\liq^2}{\mCpl} - \left(
      \tfrac{T\gas s\gas}{\mCpg} + \tfrac{T\liq s\liq}{\mCpl} \right)
    s^* \right], \\
  P^{p\mu}_\mu = \tfrac{\xi\gas \xi\liq}{\rho\gas c\gas^2 \rho\liq c\liq^2 \kappa_{p\mu}} \Bigl\{ \tfrac{\xi\gas}{\rho\gas c\gas^2} - \tfrac{\xi\liq}{\rho\liq c\liq^2} + \tfrac{\alpha\gas T\gas s\gas^2}{\xi\gas \mCpg} - \tfrac{\alpha\liq T\liq s\liq^2}{\xi\liq \mCpl} - \left[ \tfrac{\Gamma\liq T\liq}{\rho\liq c\liq^2} - \tfrac{\alpha\liq T\liq s\liq}{\xi\liq \mCpl} - \tfrac{\Gamma\gas T\gas}{\rho\gas c\gas^2} + \tfrac{\alpha\gas T\gas s\gas}{\xi\gas \mCpg} \right] s^* \Bigr\}, 
\end{gather*}

\begin{gather*}
  \Psi^\mu\gas = \tfrac{1}{\kappa_{p\mu}} \left[ \left( \tfrac{\alpha\gas}{\rho\gas c\gas^2 }  + \tfrac{\alpha\liq}{\rho\liq c\liq^2 } \right) \tfrac{T\gas s\gas (s\gas-s^*)}{\mCpg} - \tfrac{\xi\liq}{\rho\liq c\liq^2} \left( \tfrac{\xi\gas + \Gamma\gas T\gas s^*}{\rho\gas c\gas^2} - \tfrac{\xi\liq + \Gamma\liq T\liq s^*}{\rho\liq c\liq^2}
  \right) \right],\\
  V^{p\mu}_\mu = \tfrac{1}{\kappa_{p\mu}} \left[\tfrac{\xi\liq}{\rho\liq c\liq^2} \tfrac{T\gas s\gas (s\gas-s^*)}{\mCpg} + \tfrac{\xi\gas}{\rho\gas c\gas^2} \tfrac{T\liq s\liq (s\liq-s^*)}{\mCpl}  \right].
\end{gather*}

\section{Coefficients in the $T\mu$-model}\label{sec:coeff_Tmu}
The coefficients in the $T\mu$-model are given by
\begin{gather*}
  \mcl A^{T\mu}_{\mu} = \tfrac{(p\gas-p\liq)(v\gas+v\liq)}{2} A^{T\mu}\gas,
 \quad
  \mcl G^{T\mu}_\mu = - \tfrac{m\gas m\liq}{c\gas^2 c\liq^2 T \kappa_{T\mu}} \left[ \tfrac{\zeta\gas}{\mCpg} + \tfrac{\zeta\liq}{\mCpl} \right],
  \\
  \mcl T^{T\mu}_\mu = \tfrac{m\gas m\liq}{c\gas^2 c\liq^2 T \kappa_{T\mu}} \left[ \tfrac{\Gamma\gas}{m\gas} + \tfrac{\Gamma\liq}{m\liq} - \tfrac{s\gas \zeta\gas}{\mCpg} - \tfrac{s\liq \zeta\liq}{\mCpl} \right], \\
  A^{T\mu}\gas = \tfrac{m\gas m\liq}{c\gas^2 c\liq^2 T \kappa_{T\mu}} \left[ \tfrac{\Gamma\liq \zeta\gas}{m\liq \mCpg} - \tfrac{\Gamma\gas \zeta\liq}{m\gas \mCpl} + \tfrac{\zeta\gas \zeta\liq \Delta h}{\mCpg \mCpl T} \right],\\
  \mcl V^{T\mu}_{\mu,\Gas} = \tfrac{m\gas m\liq}{c\gas^2 c\liq^2 T \kappa_{T\mu}} \left[ -\tfrac{c\gas^2}{\alpha\gas}\left( \tfrac{1}{\mCpg}+\tfrac{1}{\mCpl} \right) + \tfrac{c\gas^2 c\liq^2}{\alpha\gas \alpha\liq} \tfrac{\Gamma\liq T}{\rho\liq c\liq^2}\left(\tfrac{\Gamma\gas}{c\gas^2}-\tfrac{\Gamma\liq}{c\liq^2}\right) + \tfrac{\Gamma\gas \Delta h}{\alpha\gas} \left( \tfrac{1}{\mCpl} + \tfrac{ \Gamma\liq^2 T}{m\liq c\liq^2} \right) \right],\\
  \mcl V^{T\mu}_{\mu,\Liq}   = \tfrac{m\gas m\liq}{c\gas^2 c\liq^2 T \kappa_{T\mu}} \left[ -\tfrac{c\liq^2}{\alpha\liq}\left( \tfrac{1}{\mCpg}+\tfrac{1}{\mCpl} \right) - \tfrac{c\gas^2 c\liq^2}{\alpha\gas \alpha\liq} \tfrac{\Gamma\gas T}{\rho\gas c\gas^2}\left(\tfrac{\Gamma\gas}{c\gas^2}-\tfrac{\Gamma\liq}{c\liq^2}\right) - \tfrac{\Gamma\liq \Delta h}{\alpha\liq} \left( \tfrac{1}{\mCpg} + \tfrac{ \Gamma\gas^2 T}{m\gas c\gas^2} \right) \right],\\
  \mcl K^{T\mu}_p = -(p\gas-p\liq) A^{T\mu}\gas + \mcl V^{T\mu}_{\mu,\Gas} + \mcl V^{T\mu}_{\mu,\Liq},
  \quad
  \mcl K^{T\mu}_v = - A^{T\mu}\gas, 
\end{gather*}
\begin{gather*}
  \kappa_{T\mu} = \tfrac{\Delta h^2 m\gas m\liq}{\mCpg \mCpl T^2 c\gas^2
    c\liq^2} + \tfrac{\tfrac{m\liq}{c\liq^2} +
  \tfrac{m\gas}{c\gas^2} \left( 1 + \tfrac{\Gamma\liq}{c\liq^2} \Delta h
  \right)^2}{\mCpg T} + \tfrac{\tfrac{m\gas}{c\gas^2}
  + \tfrac{m\liq}{c\liq^2} \left( 1 - \tfrac{\Gamma\gas}{c\gas^2} \Delta
  h \right)^2 }{\mCpl T } + \left(\tfrac{\Gamma\liq}{c\liq^2} -
  \tfrac{\Gamma\gas}{c\gas^2} - \tfrac{\Gamma\gas\Gamma\liq}{c\gas^2
    c\liq^2} \Delta h \right)^2.
\end{gather*}

%% file: twofluid.bbl
\newcommand{\appenergy}{Appl.~Energy} \newcommand{\ijgc}{Int.~J.~Greenh.~Gas
  Con.} \newcommand{\arcon}{Annu.~Rev.~Control}
  \newcommand{\crasp}{C.~R.~Acad.~Sci.~Paris}
  \newcommand{\jcompp}{J.~Comput.~Phys.} \newcommand{\ijmf}{Int.~J.~Multiphase
  Flow} \newcommand{\numermath}{Numer.~Math.} \newcommand{\siamjsc}{{SIAM}
  J.~Sci.~Comput.} \newcommand{\physfluids}{Phys.~Fluids}
  \newcommand{\spepe}{{SPE} Prod.~Eng.} \newcommand{\egypro}{Energy Procedia}
  \newcommand{\nedes}{Nuclear Engineering and Design}
  \newcommand{\cpam}{Comm.~Pure Appl.~Math.} \newcommand{\mmmas}{Math.~Models
  Methods Appl.~Sci.} \newcommand{\siap}{{SIAM} J.~Appl.~Math.}
  \newcommand{\jgres}{J.~Geophys.~Res.}
  \newcommand{\jmaa}{J.~Math.~Anal.~Appl.}
  \newcommand{\cmphys}{Commun.~Math.~Phys.} \newcommand{\esaimmna}{{ESAIM}:
  Math.~Model.~Numer.~Anal.} \newcommand{\ijref}{Int.~J.~Refrig.}
  \newcommand{\jfmech}{J.~Fluid~Mech.}
  \newcommand{\jhde}{J.~Hyperbol.~Differ.~Eq.} \newcommand{\ijhmt}{Int.~J.~Heat
  Mass Tran.} \newcommand{\compfluids}{Comput. Fluids}
\begin{thebibliography}{65}
\providecommand{\natexlab}[1]{#1}
\providecommand{\url}[1]{\texttt{#1}}
\providecommand{\urlprefix}{URL }
\expandafter\ifx\csname urlstyle\endcsname\relax
  \providecommand{\doi}[1]{doi:\discretionary{}{}{}#1}\else
  \providecommand{\doi}[1]{doi:\discretionary{}{}{}\begingroup
  \urlstyle{rm}\url{#1}\endgroup}\fi
\providecommand{\bibinfo}[2]{#2}

\bibitem[{Bestion(1990)}]{bestion_physical_1990}
\bibinfo{author}{D.~Bestion}, \bibinfo{title}{The physical closure laws in the
  CATHARE code}, \bibinfo{journal}{\nedes}
  \bibinfo{volume}{124}~(\bibinfo{number}{3}) (\bibinfo{year}{1990})
  \bibinfo{pages}{229--245}, \doi{\bibinfo{doi}{10.1016/0029-5493(90)90294-8}}.

\bibitem[{Aarsnes et~al.(2016)Aarsnes, Fl{\aa}tten, and
  Aamo}]{aarsnes2016review}
\bibinfo{author}{U.~J.~F. Aarsnes}, \bibinfo{author}{T.~Fl{\aa}tten},
  \bibinfo{author}{O.~M. Aamo}, \bibinfo{title}{Review of two-phase flow models
  for control and estimation}, \bibinfo{journal}{\arcon} \bibinfo{volume}{42}
  (\bibinfo{year}{2016}) \bibinfo{pages}{50--62},
  \doi{\bibinfo{doi}{10.1016/j.arcontrol.2016.06.001}}.

\bibitem[{Bendiksen et~al.(1991)Bendiksen, Maines, Moe, Nuland
  et~al.}]{bendiksen_dynamic_1991}
\bibinfo{author}{K.~H. Bendiksen}, \bibinfo{author}{D.~Maines},
  \bibinfo{author}{R.~Moe}, \bibinfo{author}{S.~Nuland}, et~al.,
  \bibinfo{title}{The dynamic two-fluid model OLGA: Theory and application},
  \bibinfo{journal}{\spepe} \bibinfo{volume}{6}~(\bibinfo{number}{02})
  (\bibinfo{year}{1991}) \bibinfo{pages}{171--180},
  \doi{\bibinfo{doi}{10.2118/19451-PA}}.

\bibitem[{Pettersen et~al.(1998)Pettersen, Hafner, Skaugen, and
  Rekstad}]{pettersen98}
\bibinfo{author}{J.~Pettersen}, \bibinfo{author}{A.~Hafner},
  \bibinfo{author}{G.~Skaugen}, \bibinfo{author}{H.~Rekstad},
  \bibinfo{title}{Development of compact heat exchangers for {CO$_2$}
  air-conditioning systems}, \bibinfo{journal}{\ijref}
  \bibinfo{volume}{21}~(\bibinfo{number}{3}) (\bibinfo{year}{1998})
  \bibinfo{pages}{180--193},
  \doi{\bibinfo{doi}{10.1016/S0140-7007(98)00013-9}}.

\bibitem[{Saurel et~al.(2008)Saurel, Petitpas, and
  Abgrall}]{saurel_modelling_2008}
\bibinfo{author}{R.~Saurel}, \bibinfo{author}{F.~Petitpas},
  \bibinfo{author}{R.~Abgrall}, \bibinfo{title}{Modelling phase transition in
  metastable liquids: application to cavitating and flashing flows},
  \bibinfo{journal}{\jfmech} \bibinfo{volume}{607} (\bibinfo{year}{2008})
  \bibinfo{pages}{313--350}, \doi{\bibinfo{doi}{10.1017/S0022112008002061}}.

\bibitem[{Berstad et~al.(2011)Berstad, D{\o}rum, Jakobsen, Kragset, Li, Lund,
  Morin, Munkejord, M{\o}lnvik, Nordhagen et~al.}]{berstad_co2_2011}
\bibinfo{author}{T.~Berstad}, \bibinfo{author}{C.~D{\o}rum},
  \bibinfo{author}{J.~P. Jakobsen}, \bibinfo{author}{S.~Kragset},
  \bibinfo{author}{H.~Li}, \bibinfo{author}{H.~Lund},
  \bibinfo{author}{A.~Morin}, \bibinfo{author}{S.~T. Munkejord},
  \bibinfo{author}{M.~J. M{\o}lnvik}, \bibinfo{author}{H.~O. Nordhagen},
  et~al., \bibinfo{title}{CO2 pipeline integrity: A new evaluation
  methodology}, \bibinfo{journal}{\egypro} \bibinfo{volume}{4}
  (\bibinfo{year}{2011}) \bibinfo{pages}{3000--3007},
  \doi{\bibinfo{doi}{10.1016/j.egypro.2011.02.210}}.

\bibitem[{Munkejord et~al.(2016)Munkejord, Hammer, and
  L{\o}vseth}]{munkejord2016co2}
\bibinfo{author}{S.~T. Munkejord}, \bibinfo{author}{M.~Hammer},
  \bibinfo{author}{S.~W. L{\o}vseth}, \bibinfo{title}{CO2 transport: Data and
  models--A review}, \bibinfo{journal}{\appenergy} \bibinfo{volume}{169}
  (\bibinfo{year}{2016}) \bibinfo{pages}{499--523},
  \doi{\bibinfo{doi}{10.1016/j.apenergy.2016.01.100}}.

\bibitem[{Linga and Lund(2016)}]{linga2016two}
\bibinfo{author}{G.~Linga}, \bibinfo{author}{H.~Lund}, \bibinfo{title}{A
  two-fluid model for vertical flow applied to CO2 injection wells},
  \bibinfo{journal}{\ijgc} \bibinfo{volume}{51} (\bibinfo{year}{2016})
  \bibinfo{pages}{71--80}, \doi{\bibinfo{doi}{10.1016/j.ijggc.2016.05.009}}.

\bibitem[{Ishii and Hibiki(2011)}]{ishii_thermo-fluid_2011}
\bibinfo{author}{M.~Ishii}, \bibinfo{author}{T.~Hibiki},
  \bibinfo{title}{Thermo-fluid dynamics of two-phase flow},
  \bibinfo{publisher}{Springer-Verlag}, \bibinfo{year}{2011}.

\bibitem[{Drew and Passman(1999)}]{drew99}
\bibinfo{author}{D.~A. Drew}, \bibinfo{author}{S.~L. Passman},
  \bibinfo{title}{Theory of multicomponent fluids}, vol. \bibinfo{volume}{135},
  \bibinfo{publisher}{Springer-Verlag}, \bibinfo{year}{1999}.

\bibitem[{Natalini(1999)}]{natalini_recent_1998}
\bibinfo{author}{R.~Natalini}, \bibinfo{title}{Recent mathematical results on
  hyperbolic relaxation problems}, in: \bibinfo{booktitle}{Chapman \& Hall/CRC
  Monogr.~Surv.~Pure Appl.~Math.}, vol.~\bibinfo{volume}{99},
  \bibinfo{publisher}{Chapman \& Hall/CRC}, \bibinfo{year}{1999}.

\bibitem[{Solem et~al.(2015)Solem, Aursand, and Fl{\aa}tten}]{solem14}
\bibinfo{author}{S.~Solem}, \bibinfo{author}{P.~Aursand},
  \bibinfo{author}{T.~Fl{\aa}tten}, \bibinfo{title}{Wave dynamics of linear
  hyperbolic relaxation systems}, \bibinfo{journal}{\jhde}
  \bibinfo{volume}{12}~(\bibinfo{number}{04}) (\bibinfo{year}{2015})
  \bibinfo{pages}{655--670}, \doi{\bibinfo{doi}{10.1142/S0219891615500186}}.

\bibitem[{Chen et~al.(1994)Chen, Levermore, and Liu}]{chen_hyperbolic_1994}
\bibinfo{author}{G.-Q. Chen}, \bibinfo{author}{C.~D. Levermore},
  \bibinfo{author}{T.-P. Liu}, \bibinfo{title}{Hyperbolic conservation laws
  with stiff relaxation terms and entropy}, \bibinfo{journal}{\cpam}
  \bibinfo{volume}{47}~(\bibinfo{number}{6}) (\bibinfo{year}{1994})
  \bibinfo{pages}{787--830}, \doi{\bibinfo{doi}{10.1002/cpa.3160470602}}.

\bibitem[{Leray(1953)}]{leray53}
\bibinfo{author}{J.~Leray}, \bibinfo{title}{Hyperbolic differential equations},
  \bibinfo{publisher}{The Institute for {A}dvanced {S}tudy},
  \bibinfo{year}{1953}.

\bibitem[{Whitham(1974)}]{whitham74}
\bibinfo{author}{G.~B. Whitham}, \bibinfo{title}{Linear and nonlinear waves},
  vol.~\bibinfo{volume}{42}, \bibinfo{publisher}{John Wiley \& Sons},
  \bibinfo{year}{1974}.

\bibitem[{Liu(1987)}]{liu_hyperbolic_1987}
\bibinfo{author}{T.-P. Liu}, \bibinfo{title}{Hyperbolic conservation laws with
  relaxation}, \bibinfo{journal}{\cmphys}
  \bibinfo{volume}{108}~(\bibinfo{number}{1}) (\bibinfo{year}{1987})
  \bibinfo{pages}{153--175}, \doi{\bibinfo{doi}{10.1007/BF01210707}}.

\bibitem[{Yong(2001)}]{yong_basic_2001}
\bibinfo{author}{W.-A. Yong}, \bibinfo{title}{Basic aspects of hyperbolic
  relaxation systems}, in: \bibinfo{booktitle}{Advances in the Theory of Shock
  Waves}, \bibinfo{publisher}{Springer}, \bibinfo{pages}{259--305},
  \bibinfo{year}{2001}.

\bibitem[{Baudin et~al.(2005{\natexlab{a}})Baudin, Berthon, Coquel, Masson, and
  Tran}]{baudin05a}
\bibinfo{author}{M.~Baudin}, \bibinfo{author}{C.~Berthon},
  \bibinfo{author}{F.~Coquel}, \bibinfo{author}{R.~Masson},
  \bibinfo{author}{Q.~H. Tran}, \bibinfo{title}{A relaxation method for
  two-phase flow models with hydrodynamic closure law},
  \bibinfo{journal}{\numermath} \bibinfo{volume}{99}~(\bibinfo{number}{3})
  (\bibinfo{year}{2005}{\natexlab{a}}) \bibinfo{pages}{411--440},
  \doi{\bibinfo{doi}{10.1007/s00211-004-0558-1}}.

\bibitem[{Baudin et~al.(2005{\natexlab{b}})Baudin, Coquel, and
  Tran}]{baudin05b}
\bibinfo{author}{M.~Baudin}, \bibinfo{author}{F.~Coquel},
  \bibinfo{author}{Q.-H. Tran}, \bibinfo{title}{A semi-implicit relaxation
  scheme for modeling two-phase flow in a pipeline},
  \bibinfo{journal}{\siamjsc} \bibinfo{volume}{27}~(\bibinfo{number}{3})
  (\bibinfo{year}{2005}{\natexlab{b}}) \bibinfo{pages}{914--936},
  \doi{\bibinfo{doi}{10.1137/030601624}}.

\bibitem[{Fl{\aa}tten et~al.(2010)Fl{\aa}tten, Morin, and
  Munkejord}]{flatten_wave_2010}
\bibinfo{author}{T.~Fl{\aa}tten}, \bibinfo{author}{A.~Morin},
  \bibinfo{author}{S.~T. Munkejord}, \bibinfo{title}{Wave propagation in
  multicomponent flow models}, \bibinfo{journal}{\siap}
  \bibinfo{volume}{70}~(\bibinfo{number}{8}) (\bibinfo{year}{2010})
  \bibinfo{pages}{2861--2882}, \doi{\bibinfo{doi}{10.1137/090777700}}.

\bibitem[{Gvozdeva and Predvoditeleva(1969)}]{gvozdeva1969triple}
\bibinfo{author}{L.~Gvozdeva}, \bibinfo{author}{O.~Predvoditeleva},
  \bibinfo{title}{Triple configurations of detonation waves in gases},
  \bibinfo{journal}{Combustion, Explosion and Shock Waves}
  \bibinfo{volume}{5}~(\bibinfo{number}{4}) (\bibinfo{year}{1969})
  \bibinfo{pages}{309--316}, \doi{\bibinfo{doi}{10.1007/BF00742065}}.

\bibitem[{Picard and Bishnoi(1987)}]{picard1987calculation}
\bibinfo{author}{D.~Picard}, \bibinfo{author}{P.~Bishnoi},
  \bibinfo{title}{Calculation of the thermodynamic sound velocity in two-phase
  multicomponent fluids}, \bibinfo{journal}{\ijmf}
  \bibinfo{volume}{13}~(\bibinfo{number}{3}) (\bibinfo{year}{1987})
  \bibinfo{pages}{295--308}, \doi{\bibinfo{doi}{10.1016/0301-9322(87)90050-4}}.

\bibitem[{Fl{\aa}tten and Lund(2011)}]{flatten_relaxation_2011}
\bibinfo{author}{T.~Fl{\aa}tten}, \bibinfo{author}{H.~Lund},
  \bibinfo{title}{Relaxation two-phase flow models and the subcharacteristic
  condition}, \bibinfo{journal}{\mmmas}
  \bibinfo{volume}{21}~(\bibinfo{number}{12}) (\bibinfo{year}{2011})
  \bibinfo{pages}{2379--2407}, \doi{\bibinfo{doi}{10.1142/S0218202511005775}}.

\bibitem[{Baer and Nunziato(1986)}]{baer_two-phase_1986}
\bibinfo{author}{M.~Baer}, \bibinfo{author}{J.~Nunziato}, \bibinfo{title}{A
  two-phase mixture theory for the deflagration-to-detonation transition (DDT)
  in reactive granular materials}, \bibinfo{journal}{\ijmf}
  \bibinfo{volume}{12}~(\bibinfo{number}{6}) (\bibinfo{year}{1986})
  \bibinfo{pages}{861--889}, \doi{\bibinfo{doi}{10.1016/0301-9322(86)90033-9}}.

\bibitem[{Lund(2012)}]{lund_hierarchy_2012}
\bibinfo{author}{H.~Lund}, \bibinfo{title}{A hierarchy of relaxation models for
  two-phase flow}, \bibinfo{journal}{\siap}
  \bibinfo{volume}{72}~(\bibinfo{number}{6}) (\bibinfo{year}{2012})
  \bibinfo{pages}{1713--1741}, \doi{\bibinfo{doi}{10.1137/12086368X}}.

\bibitem[{Ferrer et~al.(2012)Ferrer, Fl{\aa}tten, and
  Munkejord}]{martinez_ferrer_effect_2012}
\bibinfo{author}{P.~J.~M. Ferrer}, \bibinfo{author}{T.~Fl{\aa}tten},
  \bibinfo{author}{S.~T. Munkejord}, \bibinfo{title}{On the effect of
  temperature and velocity relaxation in two-phase flow models},
  \bibinfo{journal}{\esaimmna} \bibinfo{volume}{46}~(\bibinfo{number}{2})
  (\bibinfo{year}{2012}) \bibinfo{pages}{411--442},
  \doi{\bibinfo{doi}{10.1051/m2an/2011039}}.

\bibitem[{Morin and Fl{\aa}tten(2016)}]{morin_two-fluid_2013}
\bibinfo{author}{A.~Morin}, \bibinfo{author}{T.~Fl{\aa}tten}, \bibinfo{title}{A
  two-fluid four-equation model with instantaneous thermodynamical
  equilibrium}, \bibinfo{journal}{\esaimmna}
  \bibinfo{volume}{50}~(\bibinfo{number}{4}) (\bibinfo{year}{2016})
  \bibinfo{pages}{1167--1192}, \doi{\bibinfo{doi}{10.1051/m2an/2015074}}.

\bibitem[{Gallou{\"e}t et~al.(2004)Gallou{\"e}t, H{\'e}rard, and
  Seguin}]{gallouet2004numerical}
\bibinfo{author}{T.~Gallou{\"e}t}, \bibinfo{author}{J.-M. H{\'e}rard},
  \bibinfo{author}{N.~Seguin}, \bibinfo{title}{Numerical modeling of two-phase
  flows using the two-fluid two-pressure approach}, \bibinfo{journal}{\mmmas}
  \bibinfo{volume}{14}~(\bibinfo{number}{05}) (\bibinfo{year}{2004})
  \bibinfo{pages}{663--700}, \doi{\bibinfo{doi}{10.1142/S0218202504003404}}.

\bibitem[{Allaire et~al.(2002)Allaire, Clerc, and Kokh}]{allaire2002five}
\bibinfo{author}{G.~Allaire}, \bibinfo{author}{S.~Clerc},
  \bibinfo{author}{S.~Kokh}, \bibinfo{title}{A five-equation model for the
  simulation of interfaces between compressible fluids},
  \bibinfo{journal}{\jcompp} \bibinfo{volume}{181}~(\bibinfo{number}{2})
  (\bibinfo{year}{2002}) \bibinfo{pages}{577--616},
  \doi{\bibinfo{doi}{10.1006/jcph.2002.7143}}.

\bibitem[{Kapila et~al.(2001)Kapila, Menikoff, Bdzil, Son, and
  Stewart}]{kapila2001two}
\bibinfo{author}{A.~Kapila}, \bibinfo{author}{R.~Menikoff},
  \bibinfo{author}{J.~Bdzil}, \bibinfo{author}{S.~Son}, \bibinfo{author}{D.~S.
  Stewart}, \bibinfo{title}{Two-phase modeling of deflagration-to-detonation
  transition in granular materials: Reduced equations},
  \bibinfo{journal}{\physfluids} \bibinfo{volume}{13}~(\bibinfo{number}{10})
  (\bibinfo{year}{2001}) \bibinfo{pages}{3002--3024},
  \doi{\bibinfo{doi}{10.1063/1.1398042}}.

\bibitem[{Bdzil et~al.(1999)Bdzil, Menikoff, Son, Kapila, and
  Stewart}]{bdzil1999two}
\bibinfo{author}{J.~Bdzil}, \bibinfo{author}{R.~Menikoff},
  \bibinfo{author}{S.~Son}, \bibinfo{author}{A.~Kapila}, \bibinfo{author}{D.~S.
  Stewart}, \bibinfo{title}{Two-phase modeling of deflagration-to-detonation
  transition in granular materials: A critical examination of modeling issues},
  \bibinfo{journal}{\physfluids} \bibinfo{volume}{11}~(\bibinfo{number}{2})
  (\bibinfo{year}{1999}) \bibinfo{pages}{378--402},
  \doi{\bibinfo{doi}{10.1063/1.869887}}.

\bibitem[{Coquel et~al.(2002)Coquel, Gallou{\"e}t, H{\'e}rard, and
  Seguin}]{coquel2002closure}
\bibinfo{author}{F.~Coquel}, \bibinfo{author}{T.~Gallou{\"e}t},
  \bibinfo{author}{J.-M. H{\'e}rard}, \bibinfo{author}{N.~Seguin},
  \bibinfo{title}{Closure laws for a two-fluid two-pressure model},
  \bibinfo{journal}{\crasp} \bibinfo{volume}{334}~(\bibinfo{number}{10})
  (\bibinfo{year}{2002}) \bibinfo{pages}{927--932},
  \doi{\bibinfo{doi}{10.1016/S1631-073X(02)02366-X}}.

\bibitem[{Zein et~al.(2010)Zein, Hantke, and Warnecke}]{zein2010modeling}
\bibinfo{author}{A.~Zein}, \bibinfo{author}{M.~Hantke},
  \bibinfo{author}{G.~Warnecke}, \bibinfo{title}{Modeling phase transition for
  compressible two-phase flows applied to metastable liquids},
  \bibinfo{journal}{\jcompp} \bibinfo{volume}{229}~(\bibinfo{number}{8})
  (\bibinfo{year}{2010}) \bibinfo{pages}{2964--2998},
  \doi{\bibinfo{doi}{10.1016/j.jcp.2009.12.026}}.

\bibitem[{Saurel et~al.(2009)Saurel, Petitpas, and Berry}]{saurel2009simple}
\bibinfo{author}{R.~Saurel}, \bibinfo{author}{F.~Petitpas},
  \bibinfo{author}{R.~A. Berry}, \bibinfo{title}{Simple and efficient
  relaxation methods for interfaces separating compressible fluids, cavitating
  flows and shocks in multiphase mixtures}, \bibinfo{journal}{\jcompp}
  \bibinfo{volume}{228}~(\bibinfo{number}{5}) (\bibinfo{year}{2009})
  \bibinfo{pages}{1678--1712}, \doi{\bibinfo{doi}{10.1016/j.jcp.2008.11.002}}.

\bibitem[{Murrone and Guillard(2005)}]{murrone2005five}
\bibinfo{author}{A.~Murrone}, \bibinfo{author}{H.~Guillard}, \bibinfo{title}{A
  five equation reduced model for compressible two phase flow problems},
  \bibinfo{journal}{\jcompp} \bibinfo{volume}{202}~(\bibinfo{number}{2})
  (\bibinfo{year}{2005}) \bibinfo{pages}{664--698},
  \doi{\bibinfo{doi}{10.1016/j.jcp.2004.07.019}}.

\bibitem[{Stewart and Wendroff(1984)}]{stewart1984two}
\bibinfo{author}{H.~B. Stewart}, \bibinfo{author}{B.~Wendroff},
  \bibinfo{title}{Two-phase flow: models and methods},
  \bibinfo{journal}{\jcompp} \bibinfo{volume}{56}~(\bibinfo{number}{3})
  (\bibinfo{year}{1984}) \bibinfo{pages}{363--409},
  \doi{\bibinfo{doi}{10.1016/0021-9991(84)90103-7}}.

\bibitem[{Menikoff and Plohr(1989)}]{menikoff1989riemann}
\bibinfo{author}{R.~Menikoff}, \bibinfo{author}{B.~J. Plohr},
  \bibinfo{title}{The Riemann problem for fluid flow of real materials},
  \bibinfo{journal}{Reviews of Modern Physics}
  \bibinfo{volume}{61}~(\bibinfo{number}{1}) (\bibinfo{year}{1989})
  \bibinfo{pages}{75}, \doi{\bibinfo{doi}{10.1103/RevModPhys.61.75}}.

\bibitem[{Clerc(2000)}]{clerc2000numerical}
\bibinfo{author}{S.~Clerc}, \bibinfo{title}{Numerical simulation of the
  homogeneous equilibrium model for two-phase flows},
  \bibinfo{journal}{\jcompp} \bibinfo{volume}{161}~(\bibinfo{number}{1})
  (\bibinfo{year}{2000}) \bibinfo{pages}{354--375},
  \doi{\bibinfo{doi}{10.1006/jcph.2000.6515}}.

\bibitem[{Voss and Dahmen(2005)}]{voss2005exact}
\bibinfo{author}{A.~Voss}, \bibinfo{author}{W.~Dahmen}, \bibinfo{title}{Exact
  Riemann solution for the Euler equations with nonconvex and nonsmooth
  equation of state}, Ph.D. thesis, \bibinfo{school}{Fakult{\"a}t f{\"u}r
  Mathematik, Informatik und Naturwissenschaften}, \bibinfo{year}{2005}.

\bibitem[{Helluy and Seguin(2006)}]{helluy2006relaxation}
\bibinfo{author}{P.~Helluy}, \bibinfo{author}{N.~Seguin},
  \bibinfo{title}{Relaxation models of phase transition flows},
  \bibinfo{journal}{\esaimmna} \bibinfo{volume}{40}~(\bibinfo{number}{2})
  (\bibinfo{year}{2006}) \bibinfo{pages}{331--352},
  \doi{\bibinfo{doi}{10.1051/m2an:2006015}}.

\bibitem[{Helluy and Mathis(2011)}]{helluy2011pressure}
\bibinfo{author}{P.~Helluy}, \bibinfo{author}{H.~Mathis},
  \bibinfo{title}{Pressure laws and fast Legendre transform},
  \bibinfo{journal}{\mmmas} \bibinfo{volume}{21}~(\bibinfo{number}{04})
  (\bibinfo{year}{2011}) \bibinfo{pages}{745--775},
  \doi{\bibinfo{doi}{10.1142/S0218202511005209}}.

\bibitem[{Faccanoni et~al.(2012)Faccanoni, Kokh, and
  Allaire}]{faccanoni2012modelling}
\bibinfo{author}{G.~Faccanoni}, \bibinfo{author}{S.~Kokh},
  \bibinfo{author}{G.~Allaire}, \bibinfo{title}{Modelling and simulation of
  liquid-vapor phase transition in compressible flows based on thermodynamical
  equilibrium}, \bibinfo{journal}{\esaimmna}
  \bibinfo{volume}{46}~(\bibinfo{number}{5}) (\bibinfo{year}{2012})
  \bibinfo{pages}{1029--1054}, \doi{\bibinfo{doi}{10.1051/m2an/2011069}}.

\bibitem[{Coquel et~al.(1997)Coquel, El~Amine, Godlewski, Perthame, and
  Rascle}]{coquel_numerical_1997}
\bibinfo{author}{F.~Coquel}, \bibinfo{author}{K.~El~Amine},
  \bibinfo{author}{E.~Godlewski}, \bibinfo{author}{B.~Perthame},
  \bibinfo{author}{P.~Rascle}, \bibinfo{title}{A numerical method using upwind
  schemes for the resolution of two-phase flows}, \bibinfo{journal}{\jcompp}
  \bibinfo{volume}{136}~(\bibinfo{number}{2}) (\bibinfo{year}{1997})
  \bibinfo{pages}{272--288}, \doi{\bibinfo{doi}{10.1006/jcph.1997.5730}}.

\bibitem[{Hammer and Morin(2014)}]{hammer2014method}
\bibinfo{author}{M.~Hammer}, \bibinfo{author}{A.~Morin}, \bibinfo{title}{A
  method for simulating two-phase pipe flow with real equations of state},
  \bibinfo{journal}{\compfluids} \bibinfo{volume}{100} (\bibinfo{year}{2014})
  \bibinfo{pages}{45--58},
  \doi{\bibinfo{doi}{10.1016/j.compfluid.2014.04.030}}.

\bibitem[{Cortes et~al.(1998)Cortes, Debussche, and Toumi}]{cortes1998density}
\bibinfo{author}{J.~Cortes}, \bibinfo{author}{A.~Debussche},
  \bibinfo{author}{I.~Toumi}, \bibinfo{title}{A density perturbation method to
  study the eigenstructure of two-phase flow equation systems},
  \bibinfo{journal}{\jcompp} \bibinfo{volume}{147}~(\bibinfo{number}{2})
  (\bibinfo{year}{1998}) \bibinfo{pages}{463--484},
  \doi{\bibinfo{doi}{10.1006/jcph.1998.6096}}.

\bibitem[{Toumi et~al.(1999)Toumi, Kumbaro, and
  Paillere}]{toumi1999approximate}
\bibinfo{author}{I.~Toumi}, \bibinfo{author}{A.~Kumbaro},
  \bibinfo{author}{H.~Paillere}, \bibinfo{title}{Approximate Riemann solvers
  and flux vector splitting schemes for two-phase flow},
  \bibinfo{publisher}{CEA Saclay, Direction de l'information scientifique et
  technique}, \bibinfo{year}{1999}.

\bibitem[{Gallou{\"e}t et~al.(2010)Gallou{\"e}t, Helluy, H{\'e}rard, and
  Nussbaum}]{gallouet2010hyperbolic}
\bibinfo{author}{T.~Gallou{\"e}t}, \bibinfo{author}{P.~Helluy},
  \bibinfo{author}{J.-M. H{\'e}rard}, \bibinfo{author}{J.~Nussbaum},
  \bibinfo{title}{Hyperbolic relaxation models for granular flows},
  \bibinfo{journal}{\esaimmna} \bibinfo{volume}{44}~(\bibinfo{number}{2})
  (\bibinfo{year}{2010}) \bibinfo{pages}{371--400},
  \doi{\bibinfo{doi}{10.1051/m2an/2010006}}.

\bibitem[{Stuhmiller(1977)}]{stuhmiller1977influence}
\bibinfo{author}{J.~Stuhmiller}, \bibinfo{title}{The influence of interfacial
  pressure forces on the character of two-phase flow model equations},
  \bibinfo{journal}{\ijmf} \bibinfo{volume}{3}~(\bibinfo{number}{6})
  (\bibinfo{year}{1977}) \bibinfo{pages}{551--560},
  \doi{\bibinfo{doi}{10.1016/0301-9322(77)90029-5}}.

\bibitem[{Ndjinga(2007)}]{ndjinga2007influence}
\bibinfo{author}{M.~Ndjinga}, \bibinfo{title}{Influence of interfacial pressure
  on the hyperbolicity of the two-fluid model}, \bibinfo{journal}{\crasp}
  \bibinfo{volume}{344}~(\bibinfo{number}{6}) (\bibinfo{year}{2007})
  \bibinfo{pages}{407--412}, \doi{\bibinfo{doi}{10.1016/j.crma.2007.02.006}}.

\bibitem[{Roe(1981)}]{roe1981approximate}
\bibinfo{author}{P.~L. Roe}, \bibinfo{title}{Approximate Riemann solvers,
  parameter vectors, and difference schemes}, \bibinfo{journal}{\jcompp}
  \bibinfo{volume}{43}~(\bibinfo{number}{2}) (\bibinfo{year}{1981})
  \bibinfo{pages}{357--372}, \doi{\bibinfo{doi}{10.1016/0021-9991(81)90128-5}}.

\bibitem[{Saurel and Abgrall(1999)}]{saurel_multiphase_1999}
\bibinfo{author}{R.~Saurel}, \bibinfo{author}{R.~Abgrall}, \bibinfo{title}{A
  Multiphase {G}odunov Method for Compressible Multifluid and Multiphase
  Flows}, \bibinfo{journal}{\jcompp}
  \bibinfo{volume}{150}~(\bibinfo{number}{2}) (\bibinfo{year}{1999})
  \bibinfo{pages}{425--467}, ISSN \bibinfo{issn}{0021-9991},
  \doi{\bibinfo{doi}{10.1006/jcph.1999.6187}}.

\bibitem[{Saurel and Pantano(2018)}]{saurel2018diffuse}
\bibinfo{author}{R.~Saurel}, \bibinfo{author}{C.~Pantano},
  \bibinfo{title}{Diffuse-Interface Capturing Methods for Compressible
  Two-Phase Flows}, \bibinfo{journal}{Annual Review of Fluid Mechanics}
  \bibinfo{volume}{50}~(\bibinfo{number}{1}),
  \doi{\bibinfo{doi}{10.1146/annurev-fluid-122316-050109}}.

\bibitem[{Gavrilyuk and Gouin(1999)}]{gavrilyuk1999new}
\bibinfo{author}{S.~Gavrilyuk}, \bibinfo{author}{H.~Gouin}, \bibinfo{title}{A
  new form of governing equations of fluids arising from Hamilton's principle},
  \bibinfo{journal}{International Journal of Engineering Science}
  \bibinfo{volume}{37}~(\bibinfo{number}{12}) (\bibinfo{year}{1999})
  \bibinfo{pages}{1495--1520},
  \doi{\bibinfo{doi}{10.1016/S0020-7225(98)00131-1}}.

\bibitem[{Gavrilyuk and Saurel(2002)}]{gavrilyuk2002mathematical}
\bibinfo{author}{S.~Gavrilyuk}, \bibinfo{author}{R.~Saurel},
  \bibinfo{title}{Mathematical and numerical modeling of two-phase compressible
  flows with micro-inertia}, \bibinfo{journal}{\jcompp}
  \bibinfo{volume}{175}~(\bibinfo{number}{1}) (\bibinfo{year}{2002})
  \bibinfo{pages}{326--360}, \doi{\bibinfo{doi}{10.1006/jcph.2001.6951}}.

\bibitem[{Morin(2012)}]{morin2012mathematical}
\bibinfo{author}{A.~Morin}, \bibinfo{title}{Mathematical modelling and
  numerical simulation of two-phase multi-component flows of {CO$_2$} mixtures
  in pipes}, Ph.D. thesis, \bibinfo{school}{Norges teknisk-naturvitenskapelige
  universitet, Fakultet for ingeni{\o}rvitenskap og teknologi, Institutt for
  energi-og prosessteknikk}, \bibinfo{year}{2012}.

\bibitem[{Pelanti and Shyue(2014)}]{pelanti2014mixture}
\bibinfo{author}{M.~Pelanti}, \bibinfo{author}{K.-M. Shyue}, \bibinfo{title}{A
  mixture-energy-consistent six-equation two-phase numerical model for fluids
  with interfaces, cavitation and evaporation waves},
  \bibinfo{journal}{\jcompp} \bibinfo{volume}{259} (\bibinfo{year}{2014})
  \bibinfo{pages}{331--357}, \doi{\bibinfo{doi}{10.1016/j.jcp.2013.12.003}}.

\bibitem[{Lund(2013)}]{lund_phd_thesis}
\bibinfo{author}{H.~Lund}, \bibinfo{title}{Relaxation models for two-phase flow
  with applications to {CO$_2$} transport}, Ph.D. thesis,
  \bibinfo{school}{Norges teknisk-naturvitenskapelige universitet, Fakultet for
  ingeni{\o}rvitenskap og teknologi, Institutt for energi-og prosessteknikk},
  \bibinfo{year}{2013}.

\bibitem[{Barberon and Helluy(2005)}]{barberon2005finite}
\bibinfo{author}{T.~Barberon}, \bibinfo{author}{P.~Helluy},
  \bibinfo{title}{Finite volume simulation of cavitating flows},
  \bibinfo{journal}{\compfluids} \bibinfo{volume}{34}~(\bibinfo{number}{7})
  (\bibinfo{year}{2005}) \bibinfo{pages}{832--858},
  \doi{\bibinfo{doi}{10.1016/j.compfluid.2004.06.004}}.

\bibitem[{Linga et~al.(2015)Linga, Aursand, and Fl{\aa}tten}]{linga14}
\bibinfo{author}{G.~Linga}, \bibinfo{author}{P.~Aursand},
  \bibinfo{author}{T.~Fl{\aa}tten}, \bibinfo{title}{Two-phase nozzle flow and
  the subcharacteristic condition}, \bibinfo{journal}{\jmaa}
  \bibinfo{volume}{426}~(\bibinfo{number}{2}) (\bibinfo{year}{2015})
  \bibinfo{pages}{917--934}, \doi{\bibinfo{doi}{10.1016/j.jmaa.2015.01.065}}.

\bibitem[{Xu and Chen(2000)}]{xu2000acoustic}
\bibinfo{author}{J.~Xu}, \bibinfo{author}{T.~Chen}, \bibinfo{title}{Acoustic
  wave prediction in flowing steam--water two-phase mixture},
  \bibinfo{journal}{\ijhmt} \bibinfo{volume}{43}~(\bibinfo{number}{7})
  (\bibinfo{year}{2000}) \bibinfo{pages}{1079--1088},
  \doi{\bibinfo{doi}{10.1016/S0017-9310(99)00213-6}}.

\bibitem[{Kieffer(1977)}]{kieffer1977sound}
\bibinfo{author}{S.~W. Kieffer}, \bibinfo{title}{Sound speed in liquid-gas
  mixtures: Water-air and water-steam}, \bibinfo{journal}{\jgres}
  \bibinfo{volume}{82}~(\bibinfo{number}{20}) (\bibinfo{year}{1977})
  \bibinfo{pages}{2895--2904}, \doi{\bibinfo{doi}{10.1029/JB082i020p02895}}.

\bibitem[{St{\"a}dtke(2006)}]{herbert2006gasdynamic}
\bibinfo{author}{H.~St{\"a}dtke}, \bibinfo{title}{Gasdynamic aspects of
  two-phase flow: Hyperbolicity, wave propagation phenomena and related
  numerical methods}, \bibinfo{publisher}{John Wiley \& Sons},
  \bibinfo{year}{2006}.

\bibitem[{Tosse et~al.(2015)Tosse, V{\aa}gs{\ae}ther, and
  Bjerketvedt}]{tosse2015experimental}
\bibinfo{author}{S.~Tosse}, \bibinfo{author}{K.~V{\aa}gs{\ae}ther},
  \bibinfo{author}{D.~Bjerketvedt}, \bibinfo{title}{An experimental
  investigation of rapid boiling of {CO$_2$}}, \bibinfo{journal}{Shock Waves}
  \bibinfo{volume}{25}~(\bibinfo{number}{3}) (\bibinfo{year}{2015})
  \bibinfo{pages}{277--282}, \doi{\bibinfo{doi}{10.1007/s00193-014-0523-6}}.

\bibitem[{Evje and Fl{\aa}tten(2003)}]{evje_hybrid_2003}
\bibinfo{author}{S.~Evje}, \bibinfo{author}{T.~Fl{\aa}tten},
  \bibinfo{title}{Hybrid flux-splitting schemes for a common two-fluid model},
  \bibinfo{journal}{\jcompp} \bibinfo{volume}{192}~(\bibinfo{number}{1})
  (\bibinfo{year}{2003}) \bibinfo{pages}{175--210},
  \doi{\bibinfo{doi}{10.1016/j.jcp.2003.07.001}}.

\bibitem[{Toumi and Kumbaro(1996)}]{toumi_approximate_1996}
\bibinfo{author}{I.~Toumi}, \bibinfo{author}{A.~Kumbaro}, \bibinfo{title}{An
  approximate linearized Riemann solver for a two-fluid model},
  \bibinfo{journal}{\jcompp} \bibinfo{volume}{124}~(\bibinfo{number}{2})
  (\bibinfo{year}{1996}) \bibinfo{pages}{286--300},
  \doi{\bibinfo{doi}{10.1006/jcph.1996.0060}}.

\end{thebibliography}
